\documentclass[sigconf, 9pt,  language=french,
language=german, language=spanish, language=english, review=false, anonymous=false]{acmart}

\AtBeginDocument{%
  }

\setcopyright{acmcopyright}
\copyrightyear{2024}
\acmYear{2024}
\acmDOI{XXXXXXX.XXXXXXX}

\acmConference[LICS'24]{Thirty-Ninth Annual ACM/IEEE Symposium on
Logic in Computer Science}{8–12 July 2024}{Tallinn, Estonia}
\acmPrice{15.00}
\acmISBN{978-1-4503-XXXX-X/18/06}

\usepackage{S_Packages}
\usepackage{S_Comments}
\usepackage{S_Macro}
\usepackage{S_Macro_CBN}
\usepackage{S_Macro_CBV}
\usepackage{S_Macro_Generic}
\usetikzlibrary{calc}

\AtEndPreamble{%
        \theoremstyle{acmdefinition}
        
        \newtheorem*{notation}{Notation}

  \theoremstyle{acmdefinition}
  \declaretheoremstyle[notebraces={\color{black}{\textbf{(}}}{\color{black}{\textbf{)}}}, notefont=\bfseries, headpunct=\color{black}{\textbf{.}}]{mystyle}
  \declaretheorem[style=mystyle, name=\color{olive}{\textbf{Axiom}}]{assumption}
  \crefname{assumption}{\color{olive}{\textbf{Axiom}}}{\color{olive}{\textbf{Axioms}}}
  \Crefname{assumption}{\color{olive}{\textbf{Axiom}}}{\color{olive}{\textbf{Axioms}}}

}

\begin{document}

\NewDocumentCommand{\CtxtPlugOptionParam}{ m t{'}t{'} e{_^} t{'}t{'} d{<}{>} }%
  {%
		\apostropheNbTrue{#2}{#3}{#6}{#7}%
		\IfValueT{#4}{_{#4}}%
		\IfValueT{#5}{^{#5}}%
		\IfValueT{#8}{#1{#8}}%
  }
\newcommand{\bangCtxtPlugOption}{\CtxtPlugOptionParam{\bangCtxtPlug}}
\newcommand{\cbnCtxtPlugOption}{\CtxtPlugOptionParam{\cbnCtxtPlug}}
\newcommand{\cbvCtxtPlugOption}{\CtxtPlugOptionParam{\cbvCtxtPlug}}
\newcommand{\wcbnCtxtPlugOption}{\CtxtPlugOptionParam{\wcbnCtxtPlug}}
\newcommand{\genCtxtPlugOption}{\CtxtPlugOptionParam{\genCtxtPlug}}
\newcommand{\CtxtPlugOption}{\CtxtPlugOptionParam{\CtxtPlug}}
\NewDocumentCommand{\cbneedCtxtPlugOption}{ t{'}t{'} e{_^} t{'}t{'} d{<}{>} d{|}{|} }%
	{%
		\apostropheNbTrue{#1}{#2}{#5}{#6}%
		\IfValueT{#3}{_{#3}}%
		\IfValueT{#4}{^{#4}}%
		\IfValueT{#7}{\cbneedCtxtPlug{#7}}%
    \IfValueT{#8}{\cbneedCtxtDblPlug{#8}}%
	}

        \title{Genericity Through  Stratification}

\author{Victor Arrial}
\orcid{0000-0002-1607-7403}
\affiliation{%
  \institution{Université Paris Cité, CNRS, IRIF}
  \city{Paris}
  \country{France}
}

\author{Giulio Guerrieri}
\orcid{0000-0002-0469-4279}
\affiliation{%
  \institution{Aix Marseille Univ, CNRS, LIS}
  \city{Marseille}
  \country{France}
}

\author{Delia Kesner}
\orcid{0000-0003-4254-3129}
\affiliation{%
  \institution{Université Paris Cité, CNRS, IRIF}
  \city{Paris}
  \country{France}
}

\renewcommand{\shortauthors}{V.Arrial, G.Guerrieri and D.Kesner}


\begin{abstract}
A fundamental issue in the $\lambda$-calculus is to find appropriate
notions for \emph{meaningfulness}. It is well-known that in the call-by-name $\lambda$-calculus (CbN) the meaningful terms can be identified 
with the \emph{solvable} ones, and that this notion is
not appropriate in the call-by-value $\lambda$-calculus
(CbV). This paper validates the challenging claim that yet another
notion, previously introduced in the literature as \emph{potential
valuability}, appropriately represents meaningfulness in CbV. Akin to
CbN, this claim is corroborated by proving two essential properties.
The first one is \emph{genericity}, stating that meaningless subterms
have no bearing on evaluating normalizing terms. 
To prove this (which was an open problem), we use
a novel approach based on \emph{stratified} reduction, indifferently
applicable to CbN and CbV, and in a \emph{quantitative} way. The
second property concerns \emph{consistency} of the smallest congruence
relation resulting from equating all meaningless terms.
While the consistency result is not new, we
provide the first direct operational proof of it. We also show that
such a congruence has a unique consistent and maximal extension, which
coincides with a well-known notion of observational equivalence. 
Our results thus supply the formal concepts and tools that validate 
the informal notion of meaningfulness underlying~CbN~and~CbV.
\end{abstract}



\received{20 February 2007}
\received[revised]{12 March 2009}
\received[accepted]{5 June 2009}


\maketitle

\section{Introduction}
\label{sec:intro}

\paragraph{Meaningfulness and Meaninglessness}
One of the aims of programming language theory is to establish
semantics describing the meaning of computer programs (see
\cite{Hutton23} for a survey). It is then natural to search for
criteria that  distinguish, among syntactically correct (\ie
compiling) programs, the meaningful ones (those computing something)
from the  meaningless ones (those not computing anything, for instance
because they always diverge for every input).

In the context of the pure untyped $\lambda$-calculus (the
prototypical kernel common to all functional programming languages and
many proof assistants), this question has been extensively studied
since the 70s
\cite{Scott70,ScottStrachey71,Barendregt71,Wadsworth71,Scott72,Hyland75,Wadsworth76,Barendregt92a}.
Setting up a semantic for the $\lambda$-calculus amounts to fixing an
\emph{equivalence} relation $\equiv$ on $\lambda$-terms, with the
underlying intuition that two equivalent $\lambda$-terms share the
\emph{same meaning}. Two primary requirements for $\equiv$
must~be~considered:

\begin{enumerate}[leftmargin=*,labelindent=0pt]
\item \label{p:invariance} \emph{Invariance under computation:}
  $\equiv$ must include $\beta$-equivalence $=_\beta$, which is the
  reflexive, symmetric and transitive closure of $\beta$-reduction. In
  other words, if $t =_\beta u$ then $t \equiv u$. Roughly, if $t
  =_\beta u$ then $t$ and $u$ have a common $\beta$-reduct (by the
  Church-Rosser theorem), meaning that $t$ and $u$ represent different
  stages of the same computation and thus share the same meaning.

\item \label{p:compositionality} \emph{Compositionality:} $\equiv$
	must be preserved when plugging $\lambda$-terms into the same
	context, \ie if $t \equiv u$ then $\genCCtxt<t> \equiv \genCCtxt<u>$
	for every context $\genCCtxt$. Essentially this means that the
	meaning of a compound $\lambda$-term is defined solely by the
	meaning of its subterms.
\end{enumerate}

Equivalence relations fulfilling
\Cref{p:invariance,p:compositionality} are usually called
\emph{$\lambda$-theories}~\cite{BarendregtManzonetto22}. In programming language theory, there are at least three more
desiderata for them:
\begin{enumerate}[leftmargin=*,labelindent=0pt]
	\setcounter{enumi}{2}%
\item \label{p:consistency} \emph{Consistency:} Equivalence $\equiv$
  must not equate all $\lambda$-terms, \ie
  $t \not\equiv u$ for at least two $\lambda$-terms  $t$ and $u$,
  otherwise the semantics would be trivial and convey no information.
	
\item\label{p:sensibility} \emph{Sensibility:} All meaningless
  $\lambda$-terms should be equated by $\equiv$. Essentially, there
  should be no semantic difference between two $\lambda$-terms that
  have no meaning!
	
\item\label{p:genericity} \emph{Genericity:} Meaningless subterms are
  computationally irrelevant ---in the sense that they do not play any role--- in the
  evaluation of $\beta$-norma\-li\-zing $\lambda$-terms 
\end{enumerate}

Genericity says that a meaningless subterm $t$ of a $\beta$-normalizing term
$\genCCtxt<t>$ does not
impact on the computation from $\genCCtxt<t>$ to its $\beta$-normal form $u$, and
hence $t$ is \emph{generic}, that is, it can be replaced by any other
$\lambda$-term $s$ and still $\genCCtxt<s>$ $\beta$-reduces to $u$.


A $\lambda$-theory gives insights into denotational and
operational semantics of the $\lambda$-calculus, \eg it may stem from
a denotational model (by equating exactly the $\lambda$-terms with the
same interpretation), or may reflect a particular observational
equivalence (by equating exactly the observationally equivalent
$\lambda$-terms).  Two $\lambda$-terms $t$ and $u$ are
\emph{observationally equivalent} with respect to a notion of
observable if $\genCCtxt<t>$ reduces to an observable exactly when
$\genCCtxt<u>$ does.  In programming language theory, this is a way to
determine when programs can be considered equivalent as they have the
same~behavior.

\paragraph{A Naive Notion of Meaningfulness}
A naive approach to set a semantics for the pure untyped
$\lambda$-calculus is to define the meaning of a
$\beta$-normalizing $\lambda$-term as its normal form, and equating
all $\lambda$-terms that do not $\beta$-normalize. The underlying idea
is that, as $\beta$-reduction represents evaluation and a normal form
stands for its outcome, all \emph{non}-$\beta$-normalizing
$\lambda$-terms are considered as meaningless (intuitively, they
represent non-terminating~programs). 

Evidence that this naive approach to semantics is flawed is that it
equates every fixed point combinator (any $\lambda$-term $Y$ such that
$t(Yt) =_\beta t$ for every $\lambda$-term $t$) with the prototypical
diverging $\lambda$-term $\Omega = (\abs{x}xx) \abs{x}xx$. While both
$Y$  and $\Omega$ fail to $\beta$-normalize, the former is essential
for representing recursive functions in the $\lambda$-calculus,
suggesting it should carry some non-trivial semantics, while $\Omega$
is indisputably meaningless. This suggests the idea that the
$\beta$-normalizing $\lambda$-terms are not the only meaningful
$\lambda$-terms.

Indeed, as thoroughly discussed   in \cite[Chapter 2]{barendregt84nh},
the naive approach has some major drawbacks. In particular, any
$\lambda$-theory equating all non-$\beta$-normalizing $\lambda$-terms
is \emph{inconsistent}: it equates all $\lambda$-terms, not just the
meaningless ones!

\paragraph{Solvable as meaningful}
In the 70s, Wadsworth \cite{Wadsworth71,Wadsworth76} and
Barendregt
\cite{Barendregt71,Barendregt73,Barendregt75,barendregt84nh} showed
that the meaningful $\lambda$-terms can be identified with the
\emph{solvable} ones. The definition of solvability seems
technical: a $\lambda$-term $t$ is \emph{solvable} if there is a
special kind of context, called \emph{head} context $\genHCtxt$,
sending $t$ to the identity function $\Id = \lambda z.z$, meaning that
$\genHCtxt<t>$ $\beta$-reduces to $\Id$. Roughly, a solvable
$\lambda$-term $t$ may be divergent, but its diverging subterms can be
removed via interaction with a suitable head context $\genHCtxt$ which
cannot merely discard $t$. For instance, $x\Omega$ is divergent but
solvable (take the head context $\genHCtxt =
(\abs{x}{\Hole})\abs{y}\Id$). It turns out that \emph{unsolvable}
$\lambda$-terms constitutes a strict subset of the
non-$\beta$-normalizing ones (notably all fixed point combinators are
solvable). Moreover, the smallest $\lambda$-theory that equates all
unsolvable $\lambda$-terms ---let us call it $\mathcal{H}$--- is
\emph{consistent} (\Cref{p:consistency}). Following the key idea in Barendregt's book \cite{barendregt84nh}, if we identify
\emph{solvable as meaningful}, and \mbox{hence~\emph{unsolvable as
meaningless},~we~have~that}:
\begin{itemize}[leftmargin=*,labelindent=0pt]
\item The $\lambda$-theory $\mathcal{H}$ gives rise to a
  \emph{sensible} semantics (\Cref{p:sensibility}), and

\item \emph{Genericity} holds (\Cref{p:genericity}), meaning that if
  $t$ is \emph{unsolvable} and $\genCCtxt<t>$ $\beta$-reduces to
  \emph{some}  $\beta$-normal term $u$ for some context $\genCCtxt$,
  then $\genCCtxt<s>$ $\beta$-reduces to $u$ for \emph{every}
  $\lambda$-term~$s$. 
\end{itemize} 

Not only is genericity an interesting operational property \emph{per
se}, but it is also a crucial ingredient to prove consistency of
$\mathcal{H}$ in a \emph{direct way} (not relying on consistency of other $\lambda$-theories),~following~\cite{barendregt84nh}.

The fact that \Cref{p:consistency,p:sensibility,p:genericity} hold is
the main rationale for embracing solvability as a formal
representation of the informal notion of meaningfulness in the pure
untyped $\lambda$-calculus, as extensively discussed
in~\cite{barendregt84nh}. This 
semantic approach gives rise to an elegant and successful theory of
the $\lambda$-calculus, see
\cite{Wadsworth71,Barendregt71,Barendregt73,Barendregt75,Wadsworth76,CDC78,CDC80,barendregt84nh,Barendregt92a,BarendregtManzonetto22}. 
%
%

\paragraph{Call-by-Name \& Call-by-Value}

Solvability has been extensively studied in the standard pure untyped
$\lambda$-calculus, which we refer to as \emph{call-by-name} (CbN for
short).  This approach does not require the argument to be evaluated
before being passed to the calling function.  In the 70s,
Plotkin~\cite{Plotkin75} introduced the \emph{call-by-value} (CbV, for
short) variant of the pure untyped $\lambda$-calculus, which is closer
to actual implementation of functional programming languages and proof
assistants. The key distinction lies in its reduction rule $\beta_v$:
a restriction of $\beta$-reduction that allows a redex to be fired
only when the argument is a \emph{value}, that is, a variable or an
abstraction. Among $\lambda$-theories, we distinguish
\emph{$\lambda_n$-theories} (for CbN) and \emph{$\lambda_v$-theories}
(for CbV) if their invariance is under $\beta$- and
$\beta_v$-reduction, respectively.

It is tempting to adapt the notion of solvability to CbV, by simply
replacing $\beta$-reduction with $\beta_v$-reduction in its
definition.  It has been proved that CbV solvability has appealing
properties~\cite{paolini99tia,RoccaP04,AP12,CarraroGuerrieri14,GuerrieriPaoliniRonchi17,AccattoliGuerrieri22bis}.
Unfortunately, as it is well-known, 
any $\lambda_v$-theory (\Cref{p:invariance,p:compositionality}) that
equates all CbV unsolvable $\lambda$-terms (\Cref{p:sensibility})
cannot be consistent (\Cref{p:consistency}), see
\eg~\cite{AccattoliGuerrieri22bis}. Furthermore, Accattoli and
Guerrieri~\cite{AccattoliGuerrieri22bis} have showed that genericity
(\Cref{p:genericity}) fails~for CbV unsolvable $\lambda$-terms.
Therefore, meaningfulness in CbV cannot be identified with CbV
solvability, unless we are willing to relinquish to the desired
properties outlined in
\Cref{p:consistency,p:sensibility,p:genericity}. Identifying
appropriate notions to capture CbV meaningful $\lambda$-terms and
formally validating these notions is a longstanding and
challenging~open~question.

\paragraph{Our First Contribution: Scrutable as Meaningful in CbV}
Facing these negative results, our paper investigates the notion of
meaningfulness in CbV using a different approach, in the vein of what
Barendregt \cite{barendregt84nh} did for CbN. Indeed, we look for a
suitable of CbV meaningfulness such that 
the smallest $\lambda_v$-theory it \emph{generates} (according to
\Cref{p:sensibility,p:invariance,p:compositionality}) fulfills the
other \mbox{desired properties (\Cref{p:consistency,p:genericity})}. 

Paolini and Ronchi Della Rocca \cite{paolini99tia,RoccaP04} introduced
the notion of \emph{potentially valuability} for CbV, also studied in
\cite{PaoliniPimentelRonchi06,AP12,CarraroGuerrieri14,Garcia-PerezN16}
and renamed \emph{(CbV) scrutability} in
\cite{AccattoliGuerrieri22bis}: a $\lambda$-term is \emph{scrutable}
if there is a \emph{head} context $\genHCtxt$ sending $t$ to some
value, that is, such that $\genHCtxt<t>$ \emph{$\beta_v$-reduces} to
some value.
Scrutable $\lambda$-terms form a \emph{proper superset} of the CbV
solvable ones, for instance $\abs{x}\Omega$ is scrutable but not CbV
solvable. Akin to CbV solvability,  scrutability has 
appealing characterizations
\cite{paolini99tia,RoccaP04,PaoliniPimentelRonchi06,AP12,CarraroGuerrieri14,GuerrieriPaoliniRonchi17,AccattoliGuerrieri22bis}.
But there is more than that.
Indeed, one of the main contributions of our paper is to prove that we
can identify \emph{scrutable as CbV meaningful}, and hence
\emph{inscrutable as CbV meaningless}, in that we prove the following
properties:
\begin{itemize}[leftmargin=*,labelindent=0pt]
\item \emph{Genericity} (\Cref{p:genericity}): in CbV, if $t$ is \emph{CbV
  meaningless} and $\genCCtxt<t>$ reduces to a
 normal form $u$ for some context $\genCCtxt$, then
  $\genCCtxt<s>$ reduces to $u$ for
  \mbox{every~$\lambda$-term~$s$};

\item \emph{Consistency} (\Cref{p:consistency}): the smallest
	$\lambda_v$-theory (\Cref{p:invariance,p:compositionality}) that
	equates all CbV meaningless  $\lambda$-terms (\Cref{p:sensibility})
	is consistent.
\end{itemize}

The main and non-trivial challenge is to prove genericity, conjectured but not proved
in~\cite{AccattoliGuerrieri22bis}. Concerning consistency, it is
already proved in the literature ---or can be easily derived--- from the
consistency of some related
$\lambda_v$-theories~\cite{EgidiHonsellRonchi92,paolini99tia,CarraroGuerrieri14,AccattoliGuerrieri22bis,AL24}.
But we provide the first direct operational proof of consistency,
using genericity.

\paragraph{Our Second Contribution: Stratified Genericity, Modularly}
We claim that
our approach to prove genericity is \emph{finer} than the ones in the
literature, and above all \emph{robust}, in that it indifferently
applies to CbN and CbV in a modular way.

Why our approach is finer? What is called genericity in the
literature, is called here \emph{full genericity}, because it is about
the behavior of meaningless $\lambda$-terms that are plugged in
\emph{full} contexts (\ie contexts without any restriction) when they
reduce
to \emph{full} normal forms (\ie normal forms for unrestricted
reduction). We also study \emph{non-full} genericity, by considering
non-full contexts and non-full normal forms, which amounts to looking at
contexts and normal forms until a certain \emph{depth}. Thus, a
natural notion of \emph{stratified genericity} arises, depending on
how deep we consider contexts and normal forms. Our approach proves
stratified genericity for \emph{any} depth, the full case (the usual
one considered in the literature) being only a special case where
there is no limit to the depth. Another special case is \emph{surface
genericity}, that is, stratified genericity where~the~depth~is~$0$.

Moreover, our approach does not only bring forward  \emph{qualitative}
aspects of genericity, but also \emph{quantitative} ones: if $t$ is
meaningless and $\genCCtxt<t>$ normalizes to $u$ for some context
$\genCCtxt$ in $m$ steps, then $\genCCtxt<s>$ normalizes to $u$ for
\emph{every} $\lambda$-term $s$ in exactly $m$ steps. Thus, $t$ is not
only qualitative irrelevant, but also quantitatively irrelevant. 

\paragraph{Our third contribution: $\mathcal{H}$ and $\mathcal{H}^*$
in CbV}


The literature on $\lambda_n$-theories (for CbN) is
vast~\cite{barendregt84nh,Barendregt92a,Berline00,BerlineManzonettoSalibra07,BerlineManzonettoSalibra09,
Breuvart14,BreuvartManzonettoPolonskyRuoppolo16,BarendregtManzonetto22},
and the set of $\lambda_n$-theories forms a very rich
mathematical~structure.
In particular, the aforesaid $\mathcal{H}$ has a unique maximal consistent extension called $\mathcal{H}^*$~\cite{Hyland75}.

Concerning CbV, instead, the study of $\lambda_v$-theories is still in
its early
stages~\cite{EgidiHonsellRonchi92,RoccaP04,AccattoliGuerrieri22bis},
also because of the intricacies of CbV semantics. We already mentioned
that the smallest $\lambda_v$-theory equating all CbV meaningless
$\lambda$-terms is consistent. As a further contribution, thanks again
to genericity, we prove that such a $\lambda_v$-theory has a unique
maximal consistent extension, which coincides with a well known notion
of CbV observational equivalence \cite{Plotkin75}, the one where the
notion of observable is to reduce to some value. 

All these results in CbV are perfectly analogous to well known results
for CbN \cite{Hyland75,barendregt84nh,BarendregtManzonetto22},  
suggesting that there is a sort of \emph{canonicity} in the proposal
to identify CbV meaningfulness with scrutability.

\paragraph{Methodological Remark}
In this paper, we do not work in the standard CbN~\cite{Church} and
CbV~\cite{Plotkin75} $\lambda$-calculi, but in their respective
variants with \emph{explicit substitution}s and reduction \emph{at a
distance}~\cite{AK10,AP12}. Indeed, our technique for both CbN and CbV
relies on an existing \emph{operational characterization} of
meaningfulness: a $\lambda$-term $t$ is meaningful if and only if a
\emph{certain} reduction starting at $t$ terminates. Such a reduction
is given by the notion of \emph{head reduction} for CbN~\cite{Wadsworth71,Wadsworth76}, which is \emph{internal} to the
calculus (a subreduction of $\beta$). For CbV, instead, it is not
possible to have an operational characterization of meaningfulness
internally to Plotkin's original CbV $\lambda$-calculus, as already
remarked by Paolini and Ronchi Della Rocca~\cite{paolini99tia,RoccaP04}. This weakness of Plotkin's CbV is widely
accepted in the literature on CbV, which provides many alternative
proposals to extend CbV~\cite{GregoireLeroy02,AP12,CarraroGuerrieri14,AccattoliGuerrieri16,HerbelinZimmermann09}
in such a way that an internal characterization of scrutability
becomes possible. Among them we chose the formalism in~\cite{AP12},
called here \CBVSymb, as it has good rewriting properties. For
uniformity, we also took its CbN version~\cite{AK12}.

\paragraph{Outline} \Cref{sec:cbv} presents the \CBVSymb-Calculus and
a first result of (qualitative) surface genericity. A general
methodology to prove quantitative full genericity is introduced
in~\Cref{sec:methodology}, and we instantiate it to the \CBVSymb-Calculus
in~\Cref{sec:Cbv_Stratified_Genericity}. Call-by-name genericity is
developed in~\Cref{sec:cbn}, still using the methodology proposed
in~\Cref{sec:methodology}. \Cref{sec:theories} presents the results of
consistency and sensibility for the associated theories.
\Cref{sec:conclusion} discusses related works and concludes.
Omitted proofs are in the Appendix.

\section{The \CBVSymb-Calculus}
\label{sec:cbv}

This section introduces the syntax and operational semantics  of
the \CBVSymb-calculus, including the new notion of \emph{stratified
reduction} (a key ingredient of our approach). We also recall
the $\cbvTypeSysBKRV$ type system~\cite{BKRV20}, another key tool in this work.
Adequacy of the \CBVSymb-calculus versus non-adequacy of Plotkin's CbV
is  also discussed, to justify the use of \CBVSymb.
At the end of the section, a first qualitative surface genericity result
(\Cref{lem:Cbv_Qualitative_Surface_Genericity}) is obtained \mbox{thanks to the type system}.

\paragraph*{\textbf{Syntax}} We present the syntax of the \CBVSymb
Calculus, a call-by-value $\lambda$-calculus with explicit
substitutions originally introduced in~\cite{AP12},  where it is
called the \emph{Value Substitution Calculus}. Given a countably
infinite set $\mathcal{X}$ of variables $x, y, z, \dots$, the sets of
terms ($\setCbvTerms$) and values ($\setCbvValues$) are given by the
following inductive definitions:
\begin{equation*}
    \begin{array}{r rcl}
        \textbf{(Terms)} \quad
        &{\tt t}, u&\;\Coloneqq\;& v
            \vsep \app{t}{u}
            \vsep t\esub{x}{u}
    \\
        \textbf{(Values)} \quad
        &v&\;\Coloneqq\;& x \in \setCbvVariables
            \vsep \abs{x}{t}
    \end{array}
\end{equation*}
The set $\setCbvTerms$ includes \defn{variables} $x$,
\textbf{abstractions} $\abs{x}{t}$, \defn{applications} $\app{t}{u}$
and \defn{closures} $t\esub{x}{u}$, the latter representing a pending
\defn{explicit substitution (ES)} $\esub{x}{u}$ on $t$. Terms without
ES are called \defn{\pure}. Abstractions $\abs{x}{t}$ and closures
$t\esub{x}{u}$ bind the variable $x$ in $t$.

The set of \defn{free} variables $\freeVar{t}$  of a term $t$ is
defined as expected, in particular $\freeVar{\abs{x}{t}} = \freeVar{t}
\backslash \{x\}$ and $\freeVar{t\esub{x}{u}} = \freeVar{u} \cup
(\freeVar{t} \backslash \{x\})$. The usual notion of
$\alpha$-conversion \cite{barendregt84nh} is extended to the whole set
$\setCbvTerms$, and terms are identified up to $\alpha$-conversion.
For example $x\esub{x}{\abs{y}{y}} = z\esub{z}{\abs{x}{x}}$. We denote
by $t\isub{x}{u}$ the term obtained by the usual (capture-avoiding)
meta-level substitution of the term $u$ for all free occurrences of
the variable $x$ in the term $t$.

\begin{notation}
	We set  $\Id \coloneqq \abs{x}{x}$ and $\Delta \coloneqq
	\abs{x}{\app{x}{x}}$ and $\Omega \coloneqq \app{\Delta}{\Delta}$.
\end{notation}

\paragraph*{\textbf{(Stratified) Contexts and Equality}}
  The operational semantics of \CBVSymb is often specified by
means of a \emph{weak} reduction relation~\cite{AP12} that forbids
evaluation under abstractions, called here \emph{surface} reduction.
In this work however, we equip \CBVSymb with a much finer
\emph{stratified} operational semantics which generalizes surface
reduction to different \emph{levels}. For example, given a term $t =
\abs{x}{(\app{u_1}{(\abs{y}{u_2})})}$, reduction at level $1$ allows
to evaluate inside one binder $\lambda$, but not inside two or more,
so that $u_1$ can be evaluated while $u_2$ cannot. This extended
semantics sheds a light on the layering intrinsically standing out in
the study of meaningfulness. From now on, we denote by $\setOrdinals$
the set $\mathbb{N} \cup \{\omega\}$ (with $i < \omega$ for all $i \in
 \mathbb{N}$), which will be used to
index all the forthcoming notions with respect to a level.

To define such stratified operational semantics, we need a fine-grained notion of context (a terms containing exactly one hole $\Hole$).
The set of \defn{list contexts} $(\cbvLCtxt)$ and
\defn{stratified contexts} $(\cbvSCtxt_\indexK)$ for $\indexK \in
\setOrdinals$ are inductively defined as follows (where $\indexI \in
\mathbb{N}$): 
\begin{equation*}
    \begin{array}{c}
        \begin{array}{rcl}
            \cbvLCtxt &\;\Coloneqq\;& \Hole
                \vsep \cbvLCtxt\esub{x}{t}
        \\[0.2cm]
            \cbvSCtxt_0 &\;\Coloneqq\;& \Hole
                \vsep \app[\,]{\cbvSCtxt_0}{t}
                \vsep \app[\,]{t}{\cbvSCtxt_0}
                \vsep \cbvSCtxt_0\esub{x}{t}
                \vsep t\esub{x}{\cbvSCtxt_0}
        \\
            \cbvSCtxt_{\indexI+1} &\;\Coloneqq\;& \Hole
                \vsep \abs{x}{\cbvSCtxt_\indexI}
                \vsep \app[\,]{\cbvSCtxt_{\indexI+1}}{t}
                \vsep \app[\,]{t}{\cbvSCtxt_{\indexI+1}}
                \vsep \cbvSCtxt_{\indexI+1}\esub{x}{t}
                \vsep t\esub{x}{\cbvSCtxt_{\indexI+1}}
        \\
            \cbvSCtxt_\indexOmega, \cbvCCtxt &\;\Coloneqq\;& \Hole
                \vsep \abs{x}{\cbvSCtxt_\indexOmega}
                \vsep \app[\,]{\cbvSCtxt_\indexOmega}{t}
                \vsep \app[\,]{t}{\cbvSCtxt_\indexOmega}
                \vsep \cbvSCtxt_\indexOmega\esub{x}{t}
                \vsep t\esub{x}{\cbvSCtxt_\indexOmega}
        \end{array}
    \end{array}
\end{equation*}
A stratified context is \emph{\Pure}  if it is without ES.
Special cases of stratified contexts are 
\defn{surface contexts} (whose hole is not in the
scope of any $\lambda$) and \defn{full contexts} (the
restriction-free one-hole contexts, also denoted by $\cbvCCtxt$): the former (\resp the latter) are 
stratified contexts for the minimal (\resp transfinite) level, \ie
$\cbvSCtxt_0$ (\resp $\cbvSCtxt_\indexOmega$). In a stratified context
$\cbvSCtxt_\indexK$, the hole can occur everywhere as long as its
depth is at most $\indexK$, so that in $\cbvSCtxt_\indexOmega$ the
hole can occur everywhere and in $\cbvSCtxt_0$ it cannot occur under
any abstraction. Thus for example $\app[\,]{y}{(\abs{x}{\Hole})}$ is
in $\cbvSCtxt_1$, and thus in every $\cbvSCtxt_\indexK$ with $\indexK
\geq 1$, including $\cbvSCtxt_\indexOmega$, but it is not in
$\cbvSCtxt_0$. For all $\indexK \in \setOrdinals$, we write
$\cbvSCtxt_\indexK<t>$ for the term obtained by replacing the hole in
$\cbvSCtxt_\indexK$ with the term $t$, and similarly for $\cbvLCtxt$. 

We define the \defn{stratified equality} comparing terms up to a given
depth by the following inductive rules (where $\indexK \in
\setOrdinals$ and $\indexI \in \mathbb{N})$:
\begin{equation*}
    \begin{array}{c}
    \hspace{-0.2cm}
        \begin{prooftree}
            \hypo{\phantom{\cbvEq[\indexI]}}
            \inferCbvEqVar[1]{x \cbvEq[\indexK] x}
        \end{prooftree}
    \hspace{0.4cm}
        \begin{prooftree}
            \hypo{t_1 \cbvEq[\indexK] u_1}
            \hypo{t_2 \cbvEq[\indexK] u_2}
            \inferCbvEqApp{\app{t_1}{t_2} \cbvEq[\indexK] \app{u_1}{u_2}}
        \end{prooftree}
    \hspace{0.4cm}
        \begin{prooftree}
            \hypo{t_1 \cbvEq[\indexK] u_1}
            \hypo{t_2 \cbvEq[\indexK] u_2}
            \inferCbvEqEs{t_1\esub{x}{t_2} \cbvEq[\indexK] u_1\esub{x}{u_2}}
        \end{prooftree}
    \\[0.4cm]
      
        \begin{prooftree}
            \hypo{\phantom{t \cbvEq[\indexK] u}}
            \infer1[\cbnRNPTStyle{\cbnBKRVAbsRuleName}]{\abs{x}{t} \cbvEq[0] \abs{y}{u}}
        \end{prooftree}  
    \hspace{0.8cm}
    	\begin{prooftree}
    			\hypo{t \cbvEq[\indexI] u }
    			\inferCbvEqAbs{\abs{x}{t} \cbvEq[\indexI+1] \abs{x}{u}}
    	\end{prooftree}
      \hspace{0.8cm}
      \begin{prooftree}
            \hypo{t \cbvEq[\indexOmega] u }
            \inferCbvEqAbs{\abs{x}{t} \cbvEq[\indexOmega] \abs{x}{u}}
        \end{prooftree}
    \end{array}
\end{equation*}
Notice that $t \cbvEq[\indexOmega]u $ if and only if $t = u$.
\begin{example}
    Let $t_0 := \app{(\abs{x}{\app{x}{(\abs{y}{x})}})}{z}$ and $t_1 :=
    \app{(\abs{x}{\app{x}{(\abs{z}{z})}})}{z}$. We have $t_0 \cbvEq[0]
    t_1$ and $t_0 \cbvEq[1] t_1$ but not $t_0 \cbvEq[\indexK] t_1$ for
    any $\indexK \geq 2$.
\end{example}

\paragraph*{\textbf{(Stratified) Operational Semantics.}} The notions
of contexts already established, we can now equip \CBVSymb with an operational semantics. The following rewriting
rules are the base components of such operational semantics. Any term
having the shape of the left-hand side of one of these three rules is
called a \textbf{redex}.
\begin{equation*}
    \begin{array}{c c c}
        \app{\cbvLCtxt<\abs{x}{s}>}{t}
            \;\mapstoR[\cbvSymbBeta]\;
        \cbvLCtxt<s\esub{x}{t}>
    &\hspace{0.25cm}&
        t\esub{x}{\cbvLCtxt<v>}
            \;\mapstoR[\cbvSymbSubs]\;
        \cbvLCtxt<t\isub{x}{v}>
    \end{array}
\end{equation*}

Rules $\cbvSymbBeta$ and $\cbvSymbSubs$ are assumed to be
\emph{capture free}, so no free variable of $t$ is captured by the
context $\cbvLCtxt$.  The $\cbvSymbBeta$-rule fires a $\beta$-redex
(with no restriction on the argument $t$) and generates an ES.  The
$\cbvSymbSubs$-rule encapsulates the call-by-value behavior: it fires
an ES provided that its argument is a \emph{value} $v$, possibly
wrapped in a list context $\cbvLCtxt$. In both rewrite rules,
reduction acts \emph{at a distance}~\cite{AK10}: the main constructors
involved in the rule can be separated by a finite---possibly
empty---list $\cbvLCtxt$ of ES. This mechanism unblocks redexes that
otherwise would be stuck, \eg $\app{(\abs{x}{x})\esub{y}{zz}}{v}
\mapstoR[\cbvSymbBeta] x\esub{x}{v}\esub{y}{zz}$ fires a $\beta$-redex
by taking $\cbvLCtxt= \Hole\esub{y}{zz}$ as the list context in
between the function $\abs{x}{x}$ and the argument $v$. Another
example is $(xx) \esub{x}{\Id \esub{y}{zz}} \mapstoR[\cbvSymbSubs]
(\Id\Id) \esub{y}{zz}$ which substitutes the value $\Id$ for the
variable $x$ by pushing out the list context $\cbvLCtxt=
\Hole\esub{y}{zz}$.  

The \defn{$\cbvSymbSurface_\indexK$-stratified reduction} 
$\cbvArr_{S_\indexK}$ is the
$\cbvSCtxt_\indexK$-closure of the two rewriting rules
$\mapstoR[\cbvSymbBeta]$ and $\mapstoR[\cbvSymbSubs]$. In particular,
$\cbvArr_{S_0}$ (\resp $\cbvArr_{S_\indexOmega}$) is called
\defn{surface} (resp. \defn{full}) \defn{reduction}. We
write $\cbvArr*_{S_\indexK}$ for the reflexive and transitive \emph{closure}
of the reduction $\cbvArr_{S_\indexK}$, and $\cbvArr^i_{S_\indexK}$ (with $i \in \mathbb{N}$) for the
\emph{composition} $i$ times of $\cbvArr_{S_\indexK}$. As an example, consider
\begin{equation*}
   \begin{array}{l}
            \abs{x}{\big(\abs{y}{\app{x}{w}}\big)\esub{w}{\abs{z}{\Omega}}}
        \;\cbvArr_{S_2}\;
            \abs{x}{\abs{y}{\big(\app[\,]{x}{\abs{z}{\Omega}}\big)}}
    \\
        \hspace{0.5cm}
        \;\cbvArr_{S_3}\;
            \abs{x}{\abs{y}{\big(\app[\,]{x}{\abs{z}{((\app{w}{w})\esub{w}{\Delta})}}\big)}}
        \;\cbvArr_{S_3}\;
            \abs{x}{\abs{y}{\big(\app[\,]{x}{\abs{z}{\Omega}}\big)}}
   \end{array}
\end{equation*}
The first and third reduction steps are
$\cbvSymbSubs$-steps while the second is a
$\cbvSymbBeta$-step. Moreover, the first step is indeed a
$\cbvArr_{S_\indexK}$-step for any $\indexK \geq 2$. However, the
second and third steps are not $\cbvArr_{S_\indexK}$-steps for $0 \leq
\indexK \leq 2$.
\giulio{In general, $\cbvArr_{S_\indexI} \,\subsetneq\, \cbvArr_{S_{\indexI+1}} \,\subsetneq\, \cbvArr_{S_\indexOmega}$ for all $i \in \mathbb{N}$. For $\indexK \notin \{0,\indexOmega\}$, $\cbvSymbSurface_\indexK$-stratified reduction is not confluent}.

\begin{lemma}[\cite{AP12}]
    \label{lem:Cbv_Surface_Confluence}%
    The \CBVSymb \tSurface and \tFull reductions \mbox{are
    confluent}. 
\end{lemma}

\paragraph*{\textbf{Normal Forms.}} A term $t$ is $\cbvSymbSurface_\indexK$-\defn{normal} (or a
\defn{$\cbvSymbSurface_\indexK$-normal form}), noted
$\cbvPredSNF_\indexK{t}$, if there is no $u$ such that $t
\cbvArr_{S_\indexK} u$.  Note that $\cbvPredSNF_{\indexI+1}{t}$
implies $\cbvPredSNF_\indexI{t}$. A term $t$ is a 
\defn{surface} (resp. \defn{full}) \defn{normal form} if
it is $\cbvSymbSurface_0$-normal (resp.
$\cbvSymbSurface_\omega$-normal). A term $t$ is
\defn{$\cbvSymbSurface_\indexK$-normalizing} if $t
\cbvArr*_{S_\indexK} u$  for some $\cbvSymbSurface_\indexK$-normal
form $u$. In particular, a term $t$ is  \defn{surface} (resp. \defn{full})
\defn{normalizing} if it is
$\cbvSymbSurface_0$-normalizing (resp.
$\cbvSymbSurface_\omega$-normalizing). Since $\cbvArr_{S_\indexI}
\subsetneq \cbvArr_{S_{\indexI+1}}$, then some terms may be
$\cbvArr*_{S_\indexI}$-normalizing but not
$\cbvArr*_{S_{\indexI+1}}$-normalizing, \eg 
$\abs{x}{((\abs{z}{\app{x}{y}})\esub{y}{\abs{w}{\Omega}})}$ is
$\cbvSymbSurface_0$, $\cbvSymbSurface_1$ and
$\cbvSymbSurface_2$-normalizing but not
$\cbvSymbSurface_3$-normalizing. As a consequence, there are terms
that are surface but not full normalizing.


Normal forms can be \emph{syntactically} characterized by the
following grammar $\cbvNoS_\indexK$ ($\indexK \in \setOrdinals$),
where the subgrammar $\cbvNeS_\indexK\ (\indexK \in \setOrdinals)$
generates special normal forms called \emph{neutral}, which never
create a redex when applied to an argument. 
\begin{equation*}
        \begin{array}{rcl}
            \cbvVrS_\indexK &\coloneqq& x
                \vsep \cbvVrS_\indexK\esub{x}{\cbvNeS_\indexK}
        \\[0.2cm]
            \cbvNeS_\indexK &\coloneqq& \app[\,]{\cbvVrS_\indexK}{\cbvNoS_\indexK}
                \vsep \app[\,]{\cbvNeS_\indexK}{\cbvNoS_\indexK}
                \vsep \cbvNeS_\indexK\esub{x}{\cbvNeS_\indexK}
        \\[0.2cm]
            \cbvNoS_0 &\coloneqq& \abs{x}{t}
                \vsep \cbvVrS_0
                \vsep \cbvNeS_0
                \vsep \cbvNoS_0\esub{x}{\cbvNeS_0}
        \\
            \cbvNoS_{\indexI+1} &\coloneqq& \abs{x}{\cbvNoS_\indexI}
                \vsep \cbvVrS_{\indexI+1}
                \vsep \cbvNeS_{\indexI+1}
                \vsep \cbvNoS_{\indexI+1}\esub{x}{\cbvNeS_{\indexI+1}}\; (\indexI \in \Nat)
        \\
            \cbvNoS_\indexOmega &\coloneqq& \abs{x}{\cbvNoS_\indexOmega}
                \vsep \cbvVrS_\indexOmega
                \vsep \cbvNeS_\indexOmega
                \vsep \cbvNoS_\indexOmega\esub{x}{\cbvNeS_\indexOmega}
        \end{array}
\end{equation*}

\begin{restatable}{lemma}{CbvNoSnCharacterization}
    \LemmaToFromProof{cbv_NoSn_Characterization}%
    Let $\indexK \in \setOrdinals$. Then $t \in \setCbvTerms$ is $\cbvSymbSurface_\indexK$-normal iff $t \in
    \cbvNoS_\indexK$.
\end{restatable}

\paragraph*{\textbf{Quantitative Types.}} We now present the
quantitative typing system $\cbvTypeSysBKRV$, originally introduced
in~\cite{Ehrhard12} for the $\lambda$-calculus
and extended to \CBVSymb in~\cite{BKRV20}. It
contains functional and intersection types.  Intersection is
considered to be associative, commutative but not idempotent, thus an
intersection type is represented by a (possibly empty) finite multiset
$\mset{\sigma_i}_{i \in I}$. Formally, given a countably infinite set
$\setTypeVariables$ of type variables $\alpha, \beta, \gamma, \dots$,
we inductively define:
\begin{equation*}
    \begin{array}{rcl}
        \sigma, \tau, \rho &\Coloneqq& \alpha \in \setTypeVariables
            \vsep \M
            \vsep \M \typeArrow \sigma
    \\
        \M &\Coloneqq& \mset{\sigma_i}_{i \in I}\ (I \mbox{ finite})
    \end{array}
\end{equation*}

The typing rule of system $\cbvTypeSysBKRV$ are defined as
follows:
\begin{equation*}
    \begin{array}{ccc}
        \begin{prooftree}
            \inferCbvBKRVVar{x : \M \vdash x : \M}
        \end{prooftree}
    \hspace{0.5cm}
        \begin{prooftree}
            \hypo{\Pi_i \cbvTrBKRV &\Gamma_i; x : \M_i \vdash t : \sigma_i}
            \delims{\left(}{\right)_{i \in I}}
            \inferCbvBKRVAbs{\;\;+_{i \in I} \Gamma_i &\vdash \abs{x}{t} : \mset{\M_i \typeArrow \sigma_i}_{i \in  I}\;\;}
        \end{prooftree}
    \\[0.5cm]
        \begin{prooftree}
            \hypo{\Pi_1 \cbvTrBKRV \Gamma_1 \vdash t : \mset{\M \typeArrow \sigma}}
            \hypo{\Pi_2 \cbvTrBKRV \Gamma_2 \vdash u : \M}
            \inferCbvBKRVApp{\Gamma_1 + \Gamma_2 \vdash \app{t}{u} : \sigma}
        \end{prooftree}
    \\[0.5cm]
        \begin{prooftree}
            \hypo{\Pi_1 \cbvTrBKRV \Gamma_1; x : \M \vdash t : \sigma}
            \hypo{\Pi_2 \cbvTrBKRV \Gamma_2 \vdash u : \M}
            \inferCbvBKRVEs{\Gamma_1 + \Gamma_2 \vdash t\esub{x}{u} : \sigma}
        \end{prooftree}
    \end{array}
\end{equation*}
with $I$ finite in rule $(\cbvBKRVAbsRuleName)$. A \textbf{(type)
derivation} is a tree obtained by applying the (inductive) typing
rules of system $\cbvTypeSysBKRV$. The notation $\Pi \cbvTrBKRV \Gamma
\vdash t : \sigma$ means there is a derivation of the judgment $\Gamma
\vdash t : \sigma$ in system $\cbvTypeSysBKRV$.  We say that $t$ is
\textbf{$\cbvTypeSysBKRV$-typable} if  $\Pi \cbvTrBKRV \Gamma \vdash t
: \sigma$ holds for some $\Gamma, \sigma$. For example, the term
$\abs{x}{\app{y}{z}}$ is typable
with the context$\Gamma = y
: \mset{\emptymset \typeArrow \alpha}, z : \emptymset$ and type
$\alpha$, as illustrated below. 
    \begin{equation}\label{eq:derivation}
        \begin{prooftree}
            \inferCbvBKRVVar{y : \mset{\emptymset \typeArrow \alpha} \vdash y : \mset{\emptymset \typeArrow \alpha}}

            \inferCbvBKRVAbs[0]{\vdash \abs{z}\Omega : \emptymset}

            \inferCbvBKRVApp{y : \mset{\emptymset \typeArrow \alpha} \vdash \app{y}{(\abs{z}\Omega)} : \alpha}

            \inferCbvBKRVAbs{y : \mset{\emptymset \typeArrow \alpha} \vdash \abs{x}{\app{y}{(\abs{z}\Omega})} : \mset{\emptymset \typeArrow \alpha}}
        \end{prooftree}
    \end{equation}




System $\cbvTypeSysBKRV$ \emph{characterizes} surface normalization
(see \Cref{lem:cbvBKRV_characterizes_meaningfulness}).


\paragraph*{\textbf{Meaningfulness.}} Meaningfulness is about the
possibility to differentiate operationally relevant terms from those
that are irrelevant. Operational relevance can be understood as the
ability, under appropriate circumstances, to produce some element that
is \emph{observable} within the corresponding language. Scrutability
in \CBVSymb~\cite{AccattoliGuerrieri22bis} aligns precisely with this
principle. Throughout this paper, call-by-value scrutability is called
\emph{meaningfulness}. Formally,
\begin{definition}\label{def:cbv-meaningful}
    A term $t\in \setCbvTerms$ is said \CBVSymb-\defn{meaningful} if
    there exists a testing context $\cbvTCtxt$ and a value $v \in
    \setCbvValues$ such that $\cbvTCtxt<t> \cbvArr*_{S_0} v$, where
    testing contexts are defined by $\cbvTCtxt\;\Coloneqq\; \Hole
    \vsep \app[\,]{\cbvTCtxt}{u} \vsep
    \app[\,]{(\abs{x}{\cbvTCtxt})}{u}$.
\end{definition}

The intuition behind meaningfulness is that values are the observable
elements in call-by-value, while testing contexts $\cbvTCtxt$ supply
arguments to their plugged term, possibly binding them to variables.
For example $t = \app{x}{(\abs{y}{z})}$ is meaningful as
$\cbvTCtxt<t> \cbvArr*_{S_0} \abs{y}{z}$ for $\cbvTCtxt =
\app{(\abs{x}{\Hole})}{(\abs{z}{z})}$, while $\Omega$ and
$\app{x}{\Omega}$ are meaningless as for whatever testing context $\Omega$ and $x\Omega$ are plugged into, $\Omega$ will not be erased.

Meaningfulness can be characterized \emph{operationally},
through the notion of surface-normalization, and \emph{logically},
through typability in type system $\cbvTypeSysBKRV$. The logical
characterization of meaningfulness plays a central role in this work
as it yields the initial genericity result
(\Cref{lem:t_meaningless_and_Pi_|>_F<t>_==>_Pi'_|>_F<u>}) upon which
all the other results are founded.

\begin{lemma}[\CBVSymb Meaningful Characterizations~\cite{AP12,AccattoliGuerrieri22bis,BKRV20}]
    \label{lem:cbvBKRV_characterizes_meaningfulness}%
    Let $t \in \setCbvTerms$. Then the following characterizations hold:
    \begin{itemize}[leftmargin=6em]
    \item[(Operational)] \label{lem:cbvBKRV_characterizes_meaningfulness_operational}%
        $t$ is \CBVSymb-meaningful iff $t$ is \CBVSymb
        surface-normalizing. 
    \item[(Logical)] \label{lem:cbvBKRV_characterizes_meaningfulness_logical}%
        $t$ is \CBVSymb-meaningful iff $t$ is
        $\cbvTypeSysBKRV$-typable.
  \end{itemize}
\end{lemma}

\paragraph*{\textbf{Adequacy of the \CBVSymb-Calculus.}}

Plotkin's seminal work~\cite{Plotkin75} established the foundations
and key concepts about the \emph{call-by-value} $\lambda$-calculus
(\emph{CbV}).  The syntax of Plotkin's CbV is based on
\pure\ terms. 
Using the notations introduced above, the reduction relations $\cbvArr_{\mathrm{pv}_0}$
and $\cbvArr_{\mathrm{pv}_\omega}$ in Plotkin's CbV are the closures
under \pure\  surface and full contexts, of the \emph{$\cbvSymbBetaPlot$-rule}~below:
\begin{equation*}
    \app[\,]{(\abs{x}{t})}{v} \;\mapstoR[\cbvSymbBetaPlot]\; t\isub{x}{v}
    \quad
    \text{where $v$ is a \pure\  value}
\end{equation*}
For example, $\app{(\abs{x}{x})}{\abs{y}{z}} \cbvArr_{\mathrm{pv}_0}
\abs{y}{z}$ and $\abs{x}{(\app{(\abs{y}{z})}{x})}
\cbvArr_{\mathrm{pv}_\omega} \abs{x}{z}$. As expected,
$\cbvArr^*_{\mathrm{pv}_\omega}$ denotes the reflexive-transitive
closure of 
$\cbvArr_{\mathrm{pv}_\omega}$. 
It is straightforward to see that reduction in Plotkin's CbV can be
easily simulated by reductions in \CBVSymb.
The converse fails (hence \CBVSymb reduces \emph{more} that Plotkin's CbV), even when considering \pure\  terms: take 
	$t \coloneqq \app{\app{(\abs{x}{\Delta})}{(\app{y}{y})}}{\Delta}$,
	where $\delta \coloneqq \abs{z}{\app{z}{z}}$.  Note that
	$\abs{x}{\Delta}$ is applied to $\app{y}{y}$ ---which is not a
	value and cannot reduce to a value--- so $t$ is normal in
	Plotkin's CbV, but it diverges in \CBVSymb:
\begin{align*}\label{eq:mismatch}
	t \cbvArr_{S_0}  \app{{\Delta}\esub{x}{\app{y}{y}}}{\Delta} \cbvArr_{S_0} (zz) \esub{z}{\Delta} \esub{x}{\app{y}{y}} \cbvArr_{S_0} \Delta\Delta \esub{x}{\app{y}{y}} \cbvArr_{S_0}  \dots
\end{align*}

Concerning  the notion of meaningfulness, it can be adapted to Plotkin's CbV: a \pure\  term  is \defn{Plotkin-meaningful} if there is a testing context $\cbvTCtxt$ and a \pure\  value  $v$ such that $\cbvTCtxt<t> \cbvArr*_{\mathrm{pv}_0} v$.

Plotkin~\cite{Plotkin75} also introduced the notion of
\emph{observational equivalence} to identify \pure\ terms having the
same behavior.  Two \pure\ terms $t, u$ are \defn{observationally
equivalent in Plotkin's CbV}, noted $t \cong^p u$, if for every \pure\
full context $\cbvCCtxt$ such that $\cbvCCtxt<t>$ and $\cbvCCtxt<u>$
are \emph{closed}, $\cbvCCtxt<t> \cbvArr*_{\mathrm{pv}_\omega} v_1$
iff $\cbvCCtxt<u> \cbvArr*_{\mathrm{pv}_\omega} v_2$ for some \pure\
values $v_1,v_2$. This notion can be naturally extended to ES: two
terms $t, u \in \setCbvTerms$ are \defn{observationally equivalent in
\CBVSymb}, noted $t \cong u$, if for every full context $\cbvCCtxt$
such that $\cbvCCtxt<t>$ and $\cbvCCtxt<u>$ are \emph{closed},
$\cbvCCtxt<t> \cbvArr*_{S_\indexOmega} v_1$ iff $\cbvCCtxt<u>
\cbvArr*_{S_\indexOmega} v_2$ for some values $v_1, v_2$. Accattoli
and Guerrieri~\cite{AccattoliGuerrieri22bis} proved that the notions
of meaningfulness and observational equivalence coincide in
Plotkin's~CbV~and~\CBVSymb.

\begin{lemma}[\cite{AccattoliGuerrieri22bis}]%
	\label{lemma:robust}
    Let $t,u$ be two \pure\ terms. 
    \begin{itemize}
	\item Observational equivalence in Plotkin and \CBVSymb coincide.

	\item $t$ is Plotkin-meaningful if and only if $t$ is
		\CBVSymb-meaningful.
\end{itemize}
\end{lemma}

Thus, the semantic notions of meaningfulness and observational
equivalence are \emph{robust} in CbV, they do not depend on the
syntactic formulation of the calculus (Plotkin's or \CBVSymb). Said
differently, \CBVSymb is \emph{conservative} with respect to Plotkin's
CbV: it does not change its semantics. However, there is a crucial
difference between Plotkin's CbV and \CBVSymb. As Paolini and Ronchi
Della Rocca~\cite{paolini99tia} first observed (see
also~\cite{RoccaP04,AP12,CarraroGuerrieri14,AccattoliGuerrieri16,AccattoliGuerrieri22bis}),
Plotkin's CbV is not \emph{adequate}:
\begin{enumerate}[leftmargin=*,labelindent=0pt]
	\item observational equivalence may equate terms having different
	operational behaviors in Plotkin's CbV, 
	normal form and a diverging term, \eg $t \coloneqq
	\app{\app{(\abs{x}{\Delta})}{(\app{y}{y})}}{\Delta}$ seen above is
	normal, but $t \cong \Delta\Delta$ and $\Delta\Delta$ is
	diverging;

	\item meaningfulness cannot be characterized operationally inside
	Plotkin's CbV, for instance again the term $t \coloneqq
	\app{\app{(\abs{x}{\Delta})}{(\app{y}{y})}}{\Delta}$ seen above is
	normal for Plotkin, but meaningless.
\end{enumerate}

The inadequacy of Plotkin's CbV is due to \emph{open} terms, that is,
terms containing free variables (such as $t$ above). Indeed, Plotkin's
CbV works well only when reduction is \emph{weak} (it does not go
under abstractions)  and restricted to \emph{closed} terms (without
free variables): this is enough to properly model evaluation  of
functional programs. But Plotkin's does not model adequately
proof-assistants~\cite{GregoireLeroy02}, where evaluation proceeds
under abstraction, and hence on possibly (locally) open terms. On the
contrary, \CBVSymb is adequate because meaningfulness can be
characterized operationally  inside the calculus
(\Cref{lem:cbvBKRV_characterizes_meaningfulness}) and observational
equivalence does not equate terms with a different operational
behavior (\Cref{sec:theories}).

\paragraph*{\textbf{Surface Genericity.}} Genericity can be stated in
various ways, according to the previous different characterizations of
meaningfulness (\Cref{lem:cbvBKRV_characterizes_meaningfulness}). As
we shall see, an operational presentation of genericity is more
informative than a logical one, it particularly allows for the
extraction of quantitative information of the reduction sequences to
normal form, as well as observational information about these normal
forms. Nonetheless, proving logical genericity, called here
\emph{typed genericity}, is a more straightforward endeavor (just an
induction on full contexts), so that we start with it.

\begin{restatable}[\CBVSymb Typed Genericity]{lemma}{CbvSurfaceTypedGenericity}
    \LemmaToFromProof{t_meaningless_and_Pi_|>_F<t>_==>_Pi'_|>_F<u>}%
    Let $t \in \setCbvTerms$ be \CBVSymb-meaningless. If $\Pi
    \cbvTrBKRV \Gamma \vdash
    \cbvCCtxt<t> : \sigma$, then there
    is $\Pi' \cbvTrBKRV \Gamma \vdash
    \cbvCCtxt<u> : \sigma$ for all $u
    \in \setCbvTerms$.%
\end{restatable}

For instance, in the derivation \eqref{eq:derivation}, we can replace $\Omega$, which is \CBVSymb-meaningless, with any $u \in \setCbvTerms$ without changing any type.

The logical characterization of meaningfulness
(\Cref{lem:cbvBKRV_characterizes_meaningfulness}) together with the
previous \CBVSymb Typed Genericity result suffice to derive the
following first (weak) surface genericity theorem.

\begin{restatable}[\CBVSymb Qualitative Surface Genericity]{theorem}{RecCbvQualitativeSurfaceGenericity}
    \label{lem:Cbv_Qualitative_Surface_Genericity}%
    Let  $t \in
    \setCbvTerms$ be \CBVSymb-meaningless. If $\cbvCCtxt<t>$ is
    \CBVSymb-meaningful then $\cbvCCtxt<u>$ is also
    \CBVSymb-meaningful for all $u \in \setCbvTerms$.
\end{restatable}

\begin{proof}
    Let $t$ be \CBVSymb-meaningless and
    $\cbvCCtxt<t>$ be
    \CBVSymb-meaningful. By 
    \Cref{{lem:cbvBKRV_characterizes_meaningfulness}},
    $\cbvCCtxt<t>$ is
    $\cbvTypeSysBKRV$-typable. By
    \Cref{lem:t_meaningless_and_Pi_|>_F<t>_==>_Pi'_|>_F<u>},
    $\cbvCCtxt<u>$ is also
    $\cbvTypeSysBKRV$-typable,  hence
    $\cbvCCtxt<u>$ is
    \CBVSymb-meaningful by
    \Cref{{lem:cbvBKRV_characterizes_meaningfulness}}.
\end{proof}

These preliminary outcomes will be further bolstered and refined in
the subsequent sections of the paper: on one hand we will prove full
genericity, as explained in \Cref{sec:intro}, on the other hand, we
will enrich the genericity property quantitatively.

\section{Towards  Full Genericity: Key Tools}
\label{sec:methodology}

In this paper, we prove genericity results for \CBVSymb and \CBNSymb,
starting from \emph{qualitative surface genericity} up to
\emph{quantitative stratified genericity}, thus including full
genericity as a particular case, an important result for the theory of
call-by-value that has never been proved before. All these results
stand upon a first \emph{typed genericity} result
(\Cref{lem:t_meaningless_and_Pi_|>_F<t>_==>_Pi'_|>_F<u>}) that we
swiftly proved using the existing logical characterization of
meaningfulness in~\cite{AccattoliGuerrieri22bis}
(\Cref{lem:cbvBKRV_characterizes_meaningfulness}). For the rest of our
results, we use a technique based on some additional key ingredients
that we discuss in what follows. We use a \emph{generic} vocabulary
``calculus'', ``stratified context'', ``stratified reduction'', etc, the
intention being that the same methodology will be applied to
call-by-value (\Cref{sec:Cbv_Stratified_Genericity}) and to
call-by-name (\Cref{sec:cbn}).

\paragraph{\defn{Specifying Unknown Subterms: Meaningful Approximants}.}  
The set of terms $\setTerms$ of the calculus is extended with a new
constant $\genAprxBot$ representing an \emph{unknown piece of term}.
The resulting set of \emph{\tPartial terms} is noted $\setPTerms$,
using bold font. All the previous notions are lifted consequently:
\tPartial stratified contexts, \tPartial stratified reduction,
\tPartial normalization, etc. They are all defined as conservative
extensions to \tPartial elements of their underlying original ones.

\tPartial^ terms are endowed with a \emph{\tPartial order}
$\genAprxLeq$, defined as the \tPartial\ full context closure of the
relation $\genAprxBot \genAprxLeq \aprxT$, for all $\aprxT \in
\setPTerms$. Intuitively, $\aprxT \genAprxLeq \aprxU$ means that
$\aprxU$ carries more information than $\aprxT$.

We use the partial calculus to abstract from meaningless subterms. For
that purpose, for each term $t$ we construct its \emph{meaningful approximant} $\MeanApprox{t}$, a \tPartial term obtained by replacing all meaningless subterms of $t$ with $\cbvAprxBot$.
As expected, $\MeanApprox{t}$ (under-)approximates $t$, in
that $\MeanApprox{t} \genAprxLeq t$. Meaningful approximants are the
cornerstone of our approach: they will be used as intermediary
stakeholders to transfer a reduction sequence between two terms into
another reduction sequence between two  other related terms.

\paragraph{\defn{\Assumps}.}

The method used to derive quantitative stratified genericity from the
preliminary typed genericity result incorporates all the previously
mentioned elements, in addition to some general statements on the
\tPartial calculus and meaningful approximants. These statements,
called \emph{\assumps}, \emph{uniformly} outline the prerequisites
that are necessary for achieving quantitative and full genericity from
qualitative and surface genericity. Nonetheless, they require
validation for each individual calculus, given that their suitability
is closely tied to their respective concept of stratified reduction.

We
now detail the four  \assumps.
Our approach strongly relies on the operational characterization of
meaningfulness. Thus, all the \assumps are related to
reduction sequences and normal forms.

Numerous natural approaches exist to extend a given calculus into some
corresponding \tPartial calculus
(\Cref{subsec:Cbv_Partial_Calculus,sec:Cbn_Partial_Terms}).
However, not all of these methods consistently allow for an accurate
approximation of the calculus's dynamic behavior. What we mean by this
is the capacity to relate reductions in the calculus to those
occurring in its \tPartial counterpart, and vice versa.

We assume that $\indexK$-reduction sequences on terms can be
approximated by means of $\indexK$-reduction sequences on meaningful
approximants. Ideally, this means that $t \genArr_{\indexK} u$ implies
$\MeanApprox{t} \rightarrow_{\indexK} \MeanApprox{u}$. However, this
is not always true in our frameworks, and it is sufficient to
approximate reduction sequences in the following weaker way.
We denote by $\genArr_{\indexK}^=$ the reflexive closure of
$\genArr_{\indexK}$.

\begin{assumption}[Dynamic Approximation]
    \label{ax:Gen_t_->Sn_u_=>_OA(t)_->*Sn_|u_OA(u)}%
    Let $\indexK \in \setOrdinals$ and $t, u \in \setTerms$.
    If $t \genArr_{\indexK} u$ then there is
    $\aprxU \in \setPTerms$ such that $\MeanApprox{t}
    \rightarrow_{\indexK}^= \aprxU \genAprxGeq \MeanApprox{u}$.%
\end{assumption}

We also  assume that the \tPartial order gives a natural lifting of
\tPartial $\indexK$-reduction sequences into other reduction sequence
(on terms or \tPartial terms) of same length.

\begin{assumption}[Dynamic Partial Lifting]
    \label{ax:Gen_t_->Sn_u_and_t_<=_t'_==>_t'_->Sn_u'_and_u_<=_u'}%
    Let $\indexK \in \setOrdinals$ and $\aprxT, \aprxU, \aprxT' \in
    \setPTerms$. If\, $\aprxT \genArr_{\indexK} \aprxU$ and $\aprxT
    \genAprxLeq \aprxT'$, then $\aprxT' \genArr_{\indexK} \aprxU'
    \genAprxGeq \aprxU$ for some $\aprxU' \in \setPTerms$.%
\end{assumption}

Stratified normal forms are exactly the meaningful and observable
outcomes behind the notion of stratified computation. Let
${\tt no}_\indexK$ be the set of $\indexK$-normal forms. 
\tPartial^ $\indexK$-normal forms without any occurrence of $\cbvAprxBot$ at
depth $\indexK$, noted ${\tt \textbf{bno}}_\indexK$, provide perfect
approximations of $\indexK$-normal forms by keeping all their
meaningful content. Our next \assump asserts that the meaningful
approximant of an $\indexK$-normal form does not loose any meaningful
information.

\begin{assumption}[Observability of Normal Form Approximants]
    \label{ax:Gen_t_in_BnoSn_==>_OA(t)_in_noSn}%
    Let $\indexK \in \setOrdinals$ and $\aprxT \in {\tt
    no}_\indexK$, then $\MeanApprox{t} \in
    {\tt\textbf{bno}}_\indexK$.%
\end{assumption}

Finally, \tPartial terms in  ${\tt \textbf{bno}}_\indexK$ are supposed
to provide perfect approximants for $\indexK$-normal forms. The last
\assump asserts that they are stable by increasing information. 

\begin{assumption}[Stability of Meaningful Observables]
    \label{ax:Gen_|t_in_bnoSn_and_t_<=_u_==>_u_in_BnoSn}%
    Let $\indexK \in \setOrdinals$ and $\aprxT \in
    {\tt\textbf{bno}}_\indexK$. If $\aprxU \in \setPTerms$
    and $\aprxT \genAprxLeq \aprxU$, then $\aprxU \in
    {\tt\textbf{bno}}_\indexK$ and $\aprxT \genEq^\indexK \aprxU$.%
\end{assumption}

\section{Call-by-Value Stratified Genericity}
\label{sec:Cbv_Stratified_Genericity}

In this section we first introduce the \emph{\tPartial}\
\CBVSymb-calculus (\Cref{subsec:Cbv_Partial_Calculus}), an extension
of \CBVSymb enabling the simplification of some (meaningless)
subterms. Then, we prove that the axioms presented in
\Cref{sec:methodology} hold in the partial \CBVSymb
(\Cref{lem:Cbv_t_->Sn_u_=>_MA(t)_->*Sn_MA(u),lem:cbvAGK_Aprx_Simulation_MultiStep,lem:Cbv_t_in_BnoSn_==>_OA(t)_in_noSn,lem:Cbv_|t_in_bnoSn_and_t_<=_u_==>_u_in_BnoSn}).
We use these properties to strengthen our first \CBVSymb qualitative
surface genericity result into a quantitative stratified one
(\Cref{lem:Cbv_Quantitative_Stratified_Genericity}).

\subsection{The Partial \CBVSymb-Calculus}
\label{subsec:Cbv_Partial_Calculus}%


\paragraph*{\textbf{Syntax.}} Terms of the \tPartial\
\CBVSymb-calculus share the same structure as those of \CBVSymb, with
the additional constant $\cbvAprxBot$. Formally, the sets
$\setCbvAprxTerms \supsetneq \setCbvTerms$ of \textbf{\tPartTerm+} and
$\setCbvAprxValues \supsetneq \setCbvValues$ of \textbf{\tPartValue+}
for the \tPartial\ \CBVSymb-calculus, denoted by the same symbols for
terms but distinguished by a bold font, are defined inductively as
follows:
\begin{equation*}
    \begin{array}{r rcl}
        \textbf{(\tPartial^ terms)} \quad
        &\aprxT, \aprxU&\Coloneqq& \aprxV
            \vsep \app[\,]{\aprxT}{\aprxU}
            \vsep \aprxT\esub{x}{\aprxU}
            \vsep \cbvAprxBot
    \\
        \textbf{(\tPartial^ values)} \quad
        &\aprxV&\Coloneqq& x \in \setCbvVariables
            \vsep \abs{x}{\aprxT}
    \end{array}
\end{equation*}


\paragraph*{\textbf{Contexts.}} As in \CBVSymb, there is a stratified
operational semantics, but now defined by means of \tPartial contexts.
The sets of \textbf{\tPartLCtxt+} $(\cbvAprxLCtxt)$ and
\textbf{\tPartSnCtxt+} $(\cbvAprxSCtxt_\indexK)$, for $\indexK \in
\setOrdinals$, are inductively defined as follows (where $\indexI \in
\setIntegers$): 
\begin{equation*}
    \begin{array}{c}
        \begin{array}{rcl}
            \cbvAprxLCtxt &\Coloneqq& \Hole
                \vsep \cbvAprxLCtxt\esub{x}{\aprxT}
        \\[0.2cm]
            \cbvAprxSCtxt_0 &\;\Coloneqq& \Hole
                \vsep \app[\,]{\cbvAprxSCtxt_0}{\aprxT}
                \vsep \app[\,]{\aprxT}{\cbvAprxSCtxt_0}
                \vsep \cbvAprxSCtxt_0\esub{x}{\aprxT}
                \vsep \aprxT\esub{x}{\cbvAprxSCtxt_0}
        \\
            \cbvAprxSCtxt_{\indexI + 1} &\Coloneqq& \Hole
                \vsep \abs{x}{\cbvAprxSCtxt_\indexI}
                \vsep \app[\,]{\cbvAprxSCtxt_{\indexI + 1}}{\aprxT}
                \vsep \app[\,]{\aprxT}{\cbvAprxSCtxt_{\indexI + 1}}
                \vsep \cbvAprxSCtxt_{\indexI + 1}\esub{x}{\aprxT}
                \vsep \aprxT\esub{x}{\cbvAprxSCtxt_{\indexI + 1}}
        \\
            \cbvAprxSCtxt_\indexOmega &\Coloneqq& \Hole
                \vsep \abs{x}{\cbvAprxSCtxt_\indexOmega}
                \vsep \app[\,]{\cbvAprxSCtxt_\indexOmega}{\aprxT}
                \vsep \app[\,]{\aprxT}{\cbvAprxSCtxt_\indexOmega}
                \vsep \cbvAprxSCtxt_\indexOmega\esub{x}{\aprxT}
                \vsep \aprxT\esub{x}{\cbvAprxSCtxt_\indexOmega}
        \end{array}
    \end{array}
\end{equation*}
In particular, we call \defn{\tPartial surface} (\resp \defn{full})
\defn{context} the base (\resp transfinite) level $\cbvAprxSCtxt_0$
(\resp $\cbvAprxSCtxt_\indexOmega$).

\paragraph*{\textbf{(Stratified) Operational Semantics.}} The
following rewriting rules are the base components of our reduction
relations. Any \tPartTerm having the shape of the left-hand side of
one of these three rules is called a \textbf{redex}.
\begin{equation*}
    \begin{array}{c c c}
        \app{\cbvAprxLCtxt\cbvCtxtPlug{\abs{x}{\aprxT}}}{\aprxU}
            \;\mapstoR[\cbvAprxSymbBeta]\;
        \cbvAprxLCtxt\cbvCtxtPlug{\aprxT\esub{x}{\aprxU}}
    &\hspace{0.25cm}&
        \aprxT\esub{x}{\cbvAprxLCtxt\cbvCtxtPlug{\aprxV}}
            \;\mapstoR[\cbvAprxSymbSubs]\;
        \cbvAprxLCtxt\cbvCtxtPlug{\aprxT\isub{x}{\aprxV}}
    \end{array}
\end{equation*}

Building upon the established pattern used for terms, the
\defn{\tPartSnRed} (or $\cbvAprxSymbSurface_\indexK$-reduction in bold) $\cbvAprxArr_{S_\indexK}$ is the \tPartial
$\indexK$-stratified closure ($\cbvAprxSCtxt_\indexK$-closure) of the
two rewriting rules $\mapstoR[\cbvAprxSymbBeta]$ and
$\mapstoR[\cbvAprxSymbSubs]$. In particular, the reduction
$\cbvAprxArr_{S_0}$ (\resp $\cbvAprxArr_{S_\indexOmega}$) is called
\defn{\tPartial surface} (resp. \defn{full}) \defn{reduction}. Its
reflexive and transitive closure is again denoted by $\cbvAprxArr*_{S_\indexK}$. A \tPartTerm
$\aprxT$ is a \defn{$\cbvAprxSymbSurface_\indexK$-normal form} (or \defn{$\cbvAprxSymbSurface_\indexK$-normal}) if
there is no $\aprxU \in \setCbvAprxTerms$ such that $\aprxT \cbvAprxArr_{S_\indexK}
\aprxU$. A \tPartTerm is
\defn{$\cbvAprxSymbSurface_\indexK$-normalizing} if $\aprxT
\cbvAprxArr*_{S_\indexK} \aprxU$ for some 
$\cbvAprxSymbSurface_\indexK$-normal form $\aprxU$, and similarly for the
\tPartial surface and full reductions.
\begin{example} \label{example:partial_reduction} 
    $\aprxT_0 =
    \big(\abs{y}{\app{\cbvAprxBot}{w}}\big)\esub{w}{\abs{z}{\cbvAprxBot}}
    \cbvAprxArr_{S_\indexK}
    \abs{y}{\app[\,]{\cbvAprxBot}{\big(\abs{z}{\cbvAprxBot}}\big)} =
    \aprxT_1$ for any $\indexK \in \setOrdinals$. In particular,
    $\aprxT_1$ is \tPartial surface normal.
\end{example}

\subsection{Quantitative Stratified Genericity}
\label{subsec:Meaningful_Approximation}%

\paragraph*{\textbf{Meaningful Approximation.}} As explained before,
our approach is based on abstracting meaningless subterms by means of
\tPartial terms. This is implemented through the tool of
\emph{meaningful approximation}, which associates with each term a
partial term obtained by replacing all its meaningless subterms with
$\cbvAprxBot$. Formally, the
\defn{\CBVSymb-meaningful approximant}
$\cbvTApprox{t}$ of a term $t \in \setCbvTerms$ is defined as
$\cbvAprxBot$ if $t$ is \CBVSymb-meaningless, and inductively
on $t$ as~follows~otherwise:
\begin{align*}\hspace{-0.3cm}
    \hspace{-0.3cm}
        \cbvTApprox{x}
            &\coloneqq x
     &&&
     \hspace{-0.4cm}
     \cbvTApprox{t\esub{x}{u}}
     &\coloneqq \cbvTApprox{t}\esub{x}{\cbvTApprox{u}}
    \\
    \hspace{-0.3cm}
        \cbvTApprox{\abs{x}{t}}
            &\coloneqq \abs{x}{\cbvTApprox{t}}
     &&&
     \hspace{-0.4cm}
     \cbvTApprox{\app{t}{u}}
     &\coloneqq \app{\cbvTApprox{t}}{\cbvTApprox{u}}
\end{align*}

\begin{example}\label{ex:approximant}
$\cbvTApprox{\app{(\abs{x}{\app{x}{\Omega}})}{(\abs{y}{\Omega})}} =
\app{(\abs{x}{\cbvAprxBot})}{(\abs{y}{\cbvAprxBot})}$.
\end{example}

Another \tPartial calculus already exists in a CbV
setting~\cite{KerinecManzonettoPagani20,KerinecMR21,ManzonettoPR2019},
which only approximates values. Their approach seems incompatible with
our definition of meaningful approximant: it cannot be used to
approximate meaningless terms such as $\Omega$.

We equip the set of \tPartial terms $\setPTerms$ with a
\defn{\tPartial order} $\genAprxLeq$ defined as the \tPartial\ full
context closure of $\genAprxBot \genAprxLeq \aprxT$ for any \tPartial
term $\aprxT \in \setPTerms$. This definition does not only allow us
to approximate values, as done in~\cite{KerinecManzonettoPagani20},
but any kind of term. As expected, the meaningful approximation of a
term approximates it.%
\begin{restatable}{lemma}{CbvMeaningfulApproximationLeqTerm}
    \LemmaToFromProof{Cbv_MA(t)_<=_t}%
    Let $t \in \setCbvTerms$, then $\cbvTApprox{t} \cbvAprxLeq t$.
\end{restatable}

Using qualitative surface genericity
(\Cref{lem:Cbv_Qualitative_Surface_Genericity}), we obtain a
genericity result for meaningful approximants which states that
replacing a meaningless subterm by any other kind of term can only
refine its meaningful approximant.

\begin{restatable}[\CBVSymb Approximant Genericity]{lemma}{CbvApproximantGenericity}
    \LemmaToFromProof{Cbv_t_meaningless_=>_MA(F<t>)_<=_MA(F<u>)}%
    Let $\cbvCCtxt$ be a full context. If $t \in \setCbvTerms$
    is \CBVSymb-meaningless, then $\cbvTApprox{\cbvCCtxt<t>}
    \cbvAprxLeq \cbvTApprox{\cbvCCtxt<u>}$ for all $u \in
    \setCbvTerms$.
\end{restatable}

\begin{example}\label{ex:approximant-genericty}
	One has $\cbvTApprox{\app{(\abs{x}{\app{x}{\Omega}})}{(\abs{y}{\Omega})}} =
	\app{(\abs{x}{\cbvAprxBot})}{(\abs{y}{\cbvAprxBot})} \cbvAprxLeq 
	\app{(\abs{x}{x\Id})}{(\abs{y}{\cbvAprxBot})} = \cbvTApprox{\app{(\abs{x}{\app{x}{\Id}})}{(\abs{y}{\Omega})}}
	$ since $\Id$ is \CBVSymb-meaningful.
\end{example}

The previous property constitutes the core of our quantitative
stratified genericity result. It will be used in particular to
construct a common \tPartial term approximating two distinct terms
$\cbvCCtxt<t>$ and $\cbvCCtxt<u>$, where $t$ is necessarily a
meaningless subterm. This common \tPartial term will then be used to
transform a reduction sequence from $\cbvCCtxt<t>$ into a reduction
sequence from $\cbvCCtxt<u>$.

\paragraph*{\textbf{Dynamic Approximation and Lifting.}} We now seek
to transfer reduction sequences between two terms by using a common
\tPartial term approximating both of them.
First, we show that term reduction steps are simulated by the
corresponding meaningful approximants. This corresponds to
\Cref{ax:Gen_t_->Sn_u_=>_OA(t)_->*Sn_|u_OA(u)} for \CBVSymb.

\begin{restatable}[\CBVSymb Dynamic Approximation]{lemma}{CbvDynamicApproximation}
    \LemmaToFromProof{Cbv_t_->Sn_u_=>_MA(t)_->*Sn_MA(u)}%
    Let $\indexK \in \setOrdinals$ and $t, u \!\in\! \setCbvTerms$.
    If\, $t \cbvArr_{S_\indexK}\! u$ then \mbox{$\cbvTApprox{t}
    \cbvAprxArr^=_{S_\indexK} \!\aprxU \cbvAprxGeq \cbvTApprox{u}$ for some  $\aprxU \!\in\!
    \setCbvAprxTerms$}.
\end{restatable}

Thus, the meaningful approximant $\cbvTApprox{t}$ replicates the
\emph{essence} of the reduction sequence from $t$. Still, there are two
important differences. First, some steps may be
equated by their meaningful approximants, \eg $t_0 \coloneqq \Omega
\cbvArr_{S_0} (\app{x}{x})\esub{x}{\Delta} \eqqcolon t_1$ while
$\cbvTApprox{t_0} = \cbvAprxBot = \cbvTApprox{t_1}$. Secondly, some
reduction steps do not yield the meaningful approximation but an
over-approximation of it, \eg consider $t_2 \coloneqq
(\app{x}{(\abs{y}{z\Delta})})\esub{z}{\Delta} \cbvArr_{S_0}
\app{x}{(\abs{y}{\Omega})} \eqqcolon t_3$, then $\cbvTApprox{t_1} =
(\app{x}{(\abs{y}{z\Delta})})\esub{z}{\Delta} \cbvAprxArr_{S_0}
\app{x}{(\abs{y}{\Omega})}$ while $\app{x}{(\abs{y}{\Omega})}
\cbvAprxGeq \app{x}{\cbvAprxBot} = \cbvTApprox{t_3}$.

\tPartial^ reduction sequences can also be lifted to greater \tPartial
terms, and in particular to terms. This is
\Cref{ax:Gen_t_->Sn_u_and_t_<=_t'_==>_t'_->Sn_u'_and_u_<=_u'} for
\CBVSymb.

\begin{restatable}[\CBVSymb Dynamic \tPartial^ Lift]{lemma}{CbvDynamicLift}
    \LemmaToFromProof{cbvAGK_Aprx_Simulation_MultiStep}%
    Let $\indexK \in \setOrdinals$ and $\aprxT, \aprxU, \aprxT'
    \!\in\! \setCbvAprxTerms$. If\, $\aprxT \cbvAprxArr_{S_\indexK}\!
    \aprxU$ and $\aprxT \cbvAprxLeq \aprxT'$, then $\aprxT'
    \!\cbvAprxArr_{S_\indexK}\! \aprxU' \!\cbvAprxGeq \aprxU$
    \mbox{for some $\aprxU' \!\in\! \setCbvAprxTerms$}.%
    \label{lem:RecApproximantSimulation}
\end{restatable}

\noindent Recalling \Cref{example:partial_reduction}, let
$\aprxT'_0 \coloneqq
\big(\abs{y}{\app{\Id}{w}}\big)\esub{w}{\abs{z}{\abs{x}{\cbvAprxBot}}}$,
so that $\aprxT_0 \cbvAprxLeq \aprxT'_0$. Then $\aprxT'_0
\cbvAprxArr_{S_\indexK}
\abs{y}{\app[\,]{\Id}{\big(\abs{z}{\abs{x}{\cbvAprxBot}}}\big)} =:
\aprxT'_1$ for all $\indexK \in \setOrdinals$, and $\aprxT_1
\cbvAprxLeq \aprxT'_1$.

\paragraph*{\textbf{Observability.}} Stratified normal forms are
precisely the meaningful and observable outcomes behind the notion of
stratified computation. For $\indexK \in \setOrdinals$, the set $\cbvAprxBnoS_\indexK$ of \tPartial terms introduced below aims
to provide perfect approximation of them, being
$\cbvAprxSymbSurface_\indexK$-normal forms without any occurrence of
$\cbvAprxBot$ at depth $\indexK$.
\begin{align*}
        \cbvAprxBvrS_\indexK &\coloneqq \, x
            \vsep \cbvAprxBvrS_\indexK\esub{x}{\cbvAprxBneS_\indexK}
    \\[0.2cm]
        \cbvAprxBneS_\indexK &\coloneqq \, \app[\,]{\cbvAprxBvrS_\indexK}{\cbvAprxBnoS_\indexK}
            \vsep \app[\,]{\cbvAprxBneS_k}{\cbvAprxBnoS_\indexK}
            \vsep \cbvAprxBneS_k\esub{x}{\cbvAprxBneS_\indexK}
    \\[0.2cm]
        \cbvAprxBnoS_0 &\coloneqq \, \abs{x}{\aprxT}
            \vsep \cbvAprxBvrS_0
            \vsep \cbvAprxBneS_0
            \vsep \cbvAprxBnoS_0\esub{x}{\cbvAprxBneS_0}
    \\
        \cbvAprxBnoS_{\indexI+1} &\coloneqq \, \abs{x}{\cbvAprxBnoS_\indexI}
            \vsep \cbvAprxBvrS_{\indexI+1}
            \vsep \cbvAprxBneS_{\indexI+1}
            \vsep \cbvAprxBnoS_{\indexI+1}\esub{x}{\cbvAprxBneS_{\indexI+1}} \ \ (\indexI \in \setIntegers)
     \\
        \cbvAprxBnoS_\indexOmega &\coloneqq \, \abs{x}{\cbvAprxBnoS_\indexOmega}
            \vsep \cbvAprxBvrS_\indexOmega
            \vsep \cbvAprxBneS_\indexOmega
            \vsep \cbvAprxBnoS_\indexOmega\esub{x}{\cbvAprxBneS_\indexOmega}
\end{align*}

Remarkably, our meaningful approximants are precise enough to
perfectly observe the important level of these meaningful results. We
thus obtain the \CBVSymb instance of
\Cref{ax:Gen_t_in_BnoSn_==>_OA(t)_in_noSn}.

\begin{restatable}[\CBVSymb Observability of Normal Form Approximants]{lemma}{CbvObservabilityNormalFormApproximant}
    \LemmaToFromProof{Cbv_t_in_BnoSn_==>_OA(t)_in_noSn}%
    Let $\indexK \in \setOrdinals$ and $\aprxT \in
    \cbvNoS_\indexK$, then $\cbvTApprox{t} \in \cbvAprxBnoS_\indexK$.%
\end{restatable}

\begin{example}\label{ex:normal-approximant}
$t_0 \coloneqq
\abs{x}{(\app{\app{x}{(\abs{y}{\app{\Id}{\Id}})}}{(\abs{z}{\app{\Id}{\Omega}})})}$
is $\cbvSymbSurface_1$-normal and its meaningful approximant is
\mbox{$\cbvTApprox{t_0} =
	\abs{x}{(\app{\app{x}{(\abs{y}{\app{\Id}{\Id}})}}{(\abs{z}{\cbvAprxBot})})}
	\in \cbvAprxBnoS_1$.}
\end{example}

Finally, we show that \tPartial terms in $\cbvAprxBnoS_\indexK$
entirely determinate the first $\indexK$ levels of the
$\indexK$-normal forms they approximate. Consequently, any greater
approximant aligns with its first $\indexK$ levels. This is the
\CBVSymb\ precisely the instance of
\Cref{ax:Gen_|t_in_bnoSn_and_t_<=_u_==>_u_in_BnoSn}.

\begin{restatable}[\CBVSymb Stability of Meaningful Observables]{lemma}{CbvStabilityMeaningfulObservables}
    \LemmaToFromProof{Cbv_|t_in_bnoSn_and_t_<=_u_==>_u_in_BnoSn}%
    Let $\indexK \in \setOrdinals$. If\, $\aprxT \in
    \cbvAprxBnoS_\indexK$ and $\aprxT \cbvAprxLeq \aprxU \in
    \setCbvAprxTerms$, then $\aprxU \in \cbvAprxBnoS_\indexK$ and
    $\aprxT \cbvEq[\indexK] \aprxU$.%
\end{restatable}

Note that $\aprxU$ can be a term. Taking
$\cbvTApprox{t_0}$ from \Cref{ex:normal-approximant}, where
$\cbvTApprox{t_0} \genAprxLeq t_0$ as expected, one verifies that both
\tPartial terms $\cbvTApprox{t_0}$ and $t_0$ share the same structure
on the first $2$-levels. And similarly for any other \tPartial term
greater than $\cbvTApprox{t_0}$.

\paragraph*{\textbf{Quantitative Stratified Genericity.}} We now have
all the material to enhance our first \emph{qualitative} surface
genericity (\Cref{lem:Cbv_Qualitative_Surface_Genericity}) into a
stronger \emph{quantitative} stratified genericity result.

\begin{theorem}[\CBVSymb Quantitative Stratified Genericity]
    \label{lem:Cbv_Quantitative_Stratified_Genericity}%
    Let $k \in \setOrdinals$. If $t \in \setCbvTerms$ is
    \CBVSymb-meaningless and $\cbvCCtxt<t>$ is
    $\cbvSymbSurface_\indexK$-normalizing, then there is $i \in
    \mathbb{N}$ such that, for every $u \in \setCbvTerms$, there are
    $t', u' \in \cbvNoS_\indexK$ such that $\cbvCCtxt<t>
    \cbvArr^{i}_{S_\indexK} t' \;\cbvEq[\indexK]\; u'
    \;\sideset{^{i}_{\mathrm{v}_\indexK\!\!\!\!}}{}\revArr\;
	    \cbvCCtxt<u>$.
\end{theorem}
\begin{proof}
    By definition, there is $s \in \cbvNoS_\indexK$ such that
    $\cbvCCtxt<t> \cbvArr*_{S_\indexK} s$. By iterating
    \Cref{lem:Cbv_t_->Sn_u_=>_MA(t)_->*Sn_MA(u),lem:cbvAGK_Aprx_Simulation_MultiStep},
    there are $\aprxS \in \setCbvAprxTerms$ and $i \in \mathbb{N}$
    such that $\cbvTApprox{\cbvCCtxt<t>} \cbvAprxArr^i_{S_\indexK}
    \aprxS \cbvAprxGeq \cbvTApprox{s}$. Since $s \in \cbvNoS_\indexK$,
    then $\cbvTApprox{s} \in \cbvAprxBnoS_\indexK$ using
    \Cref{lem:Cbv_t_in_BnoSn_==>_OA(t)_in_noSn}, thus $\aprxS \in
    \cbvAprxBnoS_\indexK$ using
    \Cref{lem:Cbv_|t_in_bnoSn_and_t_<=_u_==>_u_in_BnoSn}. 
    On one hand, $\cbvTApprox{\cbvCCtxt<t>} \cbvAprxLeq \cbvCCtxt<t>$ 
    by \Cref{lem:Cbv_MA(t)_<=_t}, thus by
    iterating \Cref{lem:cbvAGK_Aprx_Simulation_MultiStep}, one has
    $\cbvCCtxt<t> \cbvArr^i_{S_\indexK} t'$ for some $t' \in
    \setCbvTerms$ such that $\aprxS \cbvAprxLeq t'$. By
    \Cref{lem:Cbv_|t_in_bnoSn_and_t_<=_u_==>_u_in_BnoSn}, $t' \in \cbvNoS_\indexK$ and $t' \cbvEq[\indexK] \aprxS$. On
    the other hand, $\cbvTApprox{\cbvCCtxt<t>} \cbvAprxLeq 
    \cbvTApprox{\cbvCCtxt<u>}$ according to
    \Cref{lem:Cbv_t_meaningless_=>_MA(F<t>)_<=_MA(F<u>)},
    thus $\cbvTApprox{\cbvCCtxt<t>} \cbvAprxLeq \cbvCCtxt<u>$ by
    \Cref{lem:Cbv_MA(t)_<=_t} and transitivity. By iterating
    \Cref{lem:cbvAGK_Aprx_Simulation_MultiStep} again, one has that
    $\cbvCCtxt<u> \cbvArr^i_{S_\indexK} u'$ for some $u' \in
    \setCbvTerms$ such that $\aprxS \cbvAprxLeq u'$. By
    \Cref{lem:Cbv_|t_in_bnoSn_and_t_<=_u_==>_u_in_BnoSn}, $u' \in \cbvNoS_\indexK$ and $u' \cbvEq[\indexK] \aprxS$.
    Finally, $t' \cbvEq[\indexK] u'$ by transitivity of $\cbvEq[\indexK]$.
\end{proof}

    Graphically, these are the main steps in the proof of \Cref{lem:Cbv_Quantitative_Stratified_Genericity}:
    \begin{center}
        \begin{tikzpicture}
            \node               (A)         at (0, 0.25)    {$\cbvTApprox{\cbvCCtxt<t>}$};
            \node               (A')        at (0, -2)      {$\aprxS\!$};

            \node               (Ct)        at (-1, -0.5)   {$\cbvCCtxt<t>$};
            \node               (t')        at (-1, -2.5)   {$t'\!$};
            \node               (t'_in)     at (-1.25, -2.53)  {$\ni$};
            \node               (t'_set)    at (-1.7, -2.58)   {$\cbvNoS_\indexK$};

            \node               (Cu)        at (1, -0.5)    {$\cbvCCtxt<u>$};
            \node               (u')        at (1, -2.5)    {$u'\!$};
            \node[rotate=0]     (u'_in)     at (1.25, -2.53)   {$\in$};
            \node               (u'_set)    at (1.7, -2.58)    {$\cbvNoS_\indexK$};

            \draw[->>, line width=0.2mm] (A) to (A'.north);
            \node at ([yshift=-2.05cm, xshift=1.8mm]A.north)
                {{$\scriptscriptstyle n$}};
            \node at ([yshift=-2.05cm, xshift=-1.8mm]A.north)
                {{$\scriptscriptstyle \cbvAprxSymbSurface_\indexK$}};

            \draw[->>, line width=0.2mm] (Ct) to (t'.north);
            \node at ([yshift=-1.8cm, xshift=1.8mm]Cu.north)
                {{$\scriptscriptstyle n$}};
            \node at ([yshift=-1.8cm, xshift=-1.8mm]Cu.north)
                {{$\scriptscriptstyle \cbvSymbSurface_\indexK$}};

            \draw[->>, line width=0.2mm] (Cu) to (u'.north);
            \node at ([yshift=-1.8cm, xshift=1.8mm]Ct.north)
                {{$\scriptscriptstyle n$}};
            \node at ([yshift=-1.8cm, xshift=-1.8mm]Ct.north)
                {{$\scriptscriptstyle \cbvSymbSurface_\indexK$}};

            \node[rotate=25]    (deduceArrowLeft)   at (-0.5, -1.25) {$\leftsquigarrow$};
            \node at ([xshift=1mm, yshift=0.8mm]deduceArrowLeft.north)
            {{$\scriptscriptstyle {\tt Lem}.\ref{lem:cbvAGK_Aprx_Simulation_MultiStep}$}};


            \node[rotate=-25]   (deduceArrowRight)  at (0.5, -1.25)  {$\rightsquigarrow$};
            \node at ([xshift=-1mm, yshift=0.8mm]deduceArrowRight.north)
            {{$\scriptscriptstyle {\tt Lem}.\ref{lem:cbvAGK_Aprx_Simulation_MultiStep}$}};

            \node[rotate=35]    (LeqCt)     at (-0.6, -0.15)    {$\cbvAprxGeq$};
            \node   at ([yshift=1mm, xshift=-5mm]LeqCt)
                {$\scriptscriptstyle {\tt Lem.} \ref{lem:Cbv_MA(t)_<=_t}$};

            \node[rotate=35]    (LeqT')     at (-0.5, -2.13)    {$\cbvAprxGeq$};

            \node[rotate=-40]   (LeqCu)  at (0.6, -0.15)     {$\cbvAprxLeq$}; 
            \node   at ([yshift=1mm, xshift=9mm]LeqCu)
                {$\scriptscriptstyle {\tt Lem.} \ref{lem:Cbv_t_meaningless_=>_MA(F<t>)_<=_MA(F<u>)} \;\&\; {\tt Lem.} \ref{lem:Cbv_MA(t)_<=_t}$};

            \node[rotate=-40]   (LeqU')  at (0.5, -2.13)     {$\cbvAprxLeq$};


            \node               (eq)         at (0, -2.5)     {$\cbvLongEq[\indexK]$};
            \node   at ([yshift=-2.75mm]eq)
                {$\scriptscriptstyle {\tt Lem.} \ref{lem:Cbv_|t_in_bnoSn_and_t_<=_u_==>_u_in_BnoSn}$};
        \end{tikzpicture}
        \hfill\qedhere
    \end{center}

A particular case is given by the following statement since
$\indexK$-equality for $\indexK = \indexOmega$ is just the equality
between terms.

\begin{corollary}[\CBVSymb Quantitative Full Genericity]
    \label{lem:Cbv_Quantitative_Full_Genericity}%
    If $t \in \setCbvTerms$ is \CBVSymb-meaningless and $\cbvCCtxt<t>
    \cbvArr*_{S_\indexOmega} s$ where $s$ is
    $\cbvSymbSurface_\indexOmega$-normal, then there is $i \in
    \mathbb{N}$ such that, for every $u \in \setCbvTerms$,
    $\cbvCCtxt<t> \cbvArr^i_{S_\indexOmega} s
    \;\sideset{^{i}_{\mathrm{v}_\indexOmega\!\!\!\!}}{}\revArr\;
    \cbvCCtxt<u>$.
\end{corollary}

For instance, $(\abs{y}\Id)\abs{z}\Omega \cbvArr^2_{S_\indexOmega} \Id
\;\sideset{^{2}_{\mathrm{v}_\indexOmega\!\!\!\!}}{}\revArr\;
(\abs{y}\Id)\abs{z}u$ for any $u \in \setCbvTerms$.

\Cref{lem:Cbv_Quantitative_Stratified_Genericity},
\Cref{lem:Cbv_Quantitative_Full_Genericity} and the general
methodology used to prove such results are among the main
contributions of this~paper.

\section{Call-by-Name Stratified Genericity}
\label{sec:cbn}

\Cref{sec:syntax-cbn} introduces the syntax and operational semantics
of the \CBNSymb-calculus \cite{AK10,AK12}, including the new notion of \emph{stratified
reduction}. We also recall the $\cbnTypeSysBKRV$ type system~\cite{KesnerV14}. The
qualitative and quantitative stratified  genericity results follow in
\Cref{sec:genericity-cbn}.

\subsection{Syntax and Stratified Operational Semantics}
\label{sec:syntax-cbn}

%
\paragraph*{\textbf{Syntax.}} The term syntax of the \CBNSymb-calculus
is exactly the same as in \CBVSymb. We keep the same notations for
terms ($s, t, u, \dots$) and the set of terms ($\setCbnTerms$). In the
operational semantics of \CBNSymb (see below), values do not play any
role, so we do not distinguish them. Free variables,
$\alpha$-conversion and substitution are defined as in \CBVSymb.

\paragraph*{\textbf{Contexts.}}  The operational semantics of the
\CBNSymb-calculus is often specified by means of a \emph{head}
reduction that forbids evaluation inside arguments, called here
\emph{surface} reduction. Here again, we equip the \CBNSymb-calculus
with a new \emph{stratified} operational semantics which generalizes
surface reduction to different \emph{levels}.

The set of \textbf{list contexts} $(\cbnLCtxt)$ and \textbf{stratified
contexts} $(\cbnSCtxt_\indexK)$ for $\indexK \in \setOrdinals$, that
can be seen as terms containing exactly one hole $\Hole$, are
inductively defined as follows (where $\indexI \in \setIntegers$):
\begin{equation*}
        \begin{array}{rcl}
            \cbnLCtxt &\;\Coloneqq\;& \Hole
                \vsep \cbnLCtxt\esub{x}{t}
        \\[0.2cm]
            \cbnSCtxt_0 &\;\Coloneqq\;& \Hole
                \vsep \abs{x}{\cbnSCtxt_0}
                \vsep \app[\,]{\cbnSCtxt_0}{t}
                \vsep \cbnSCtxt_0\esub{x}{t}
        \\
            \cbnSCtxt_{\indexI+1} &\;\Coloneqq\;& \Hole
                \vsep \abs{x}{\cbnSCtxt_{\indexI+1}}
                \vsep \app[\,]{\cbnSCtxt_{\indexI+1}}{t}
                \vsep \app[\,]{t}{\cbnSCtxt_\indexI}
                \vsep \cbnSCtxt_{\indexI+1}\esub{x}{t}
                \vsep t\esub{x}{\cbnSCtxt_\indexI}\; 
        \\
            \cbnSCtxt_\indexOmega, \cbnCCtxt &\;\Coloneqq\;& \Hole
                \vsep \abs{x}{\cbnSCtxt_\indexOmega}
                \vsep \app[\,]{\cbnSCtxt_\indexOmega}{t}
                \vsep \app[\,]{t}{\cbnSCtxt_\indexOmega}
                \vsep \cbnSCtxt_\indexOmega\esub{x}{t}
                \vsep t\esub{x}{\cbnSCtxt_\indexOmega}
        \end{array}
\end{equation*}

In particular, \textbf{surface contexts} (whose hole is not in an
argument) and \textbf{full contexts} (the restriction-free one-hole
contexts, also denoted by $\cbnCCtxt$) are special cases of stratified context: surface contexts
are base level $\cbnSCtxt_0$ and full contexts the transfinite level
$\cbnSCtxt_\indexOmega$. In stratified contexts $\cbnSCtxt_\indexK$,
the hole can occur everywhere as long as its depth is at most
$\indexK$, while in $\cbnSCtxt_\indexOmega$ it can occur everywhere
and in $\cbnSCtxt_0$ it cannot occur in a top level argument. Thus for
example $\app[\,]{y}{(\abs{x}{\Hole})}$ is in $\cbnSCtxt_1$, and thus
in every $\cbnSCtxt_\indexK$ with $\indexK \geq 1$, including
$\cbnSCtxt_\indexOmega$, but it is not in
$\cbnSCtxt_0$. For all $\indexK \in
\setOrdinals$, we write $\cbnSCtxt_\indexK<t>$ for the term obtained
by replacing the hole in $\cbnSCtxt_\indexK$ with the term $t$, and
similarly for $\cbnLCtxt$.

We define the \defn{stratified equality} comparing terms up to a
given depth. It is defined by the following inductive rules (with
$\indexI \in \mathbb{N})$:%
\begin{center}
        \begin{prooftree}
            \hypo{\phantom{\cbvEq[\indexI]}}
            \inferCbnEqVar[1]{x \cbvEq[\indexI] x}
        \end{prooftree}
    \hspace{1.5cm}
        \begin{prooftree}
            \hypo{t \cbvEq[\indexI] u }
            \inferCbnEqAbs{\abs{x}{t} \cbvEq[\indexI] \abs{x}{u}}
        \end{prooftree}
    \\[0.2cm]
        \begin{prooftree}
            \hypo{t_1 \cbvEq[0] u_1}
            \inferCbnEqApp[1]{\app{t_1}{t_2} \cbvEq[0] \app{u_1}{u_2}}
        \end{prooftree}
    \hspace{0.3cm}
        \begin{prooftree}[separation=1em]
            \hypo{t_1 \cbvEq[\indexI+1] u_1}
            \hypo{t_2 \cbvEq[\indexI] u_2}
            \inferCbnEqApp{\app{t_1}{t_2} \cbvEq[\indexI+1] \app{u_1}{u_2}}
        \end{prooftree}
    \hspace{0.3cm}
        \begin{prooftree}[separation=1em]
            \hypo{t_1 \cbvEq[\indexOmega] u_1}
            \hypo{t_2 \cbvEq[\indexOmega] u_2}
            \inferCbnEqApp{\app{t_1}{t_2} \cbvEq[\indexOmega] \app{u_1}{u_2}}
        \end{prooftree}
    \\[0.2cm]
	\mbox{%
	\begin{prooftree}
		\hypo{t \cbvEq[0] u}
		\inferCbnEqEs[1]{t\esub{x}{t'} \cbvEq[0] u\esub{x}{u'}}
	\end{prooftree}
	\hspace{0.1cm}
	\begin{prooftree}
		\hypo{t \cbvEq[\indexI+1] u}
		\hypo{t' \cbvEq[\indexI] u'}
		\inferCbnEqEs{t\esub{x}{t'} \cbvEq[\indexI+1] u\esub{x}{u'}}
	\end{prooftree}
	\hspace{0.1cm}
	\begin{prooftree}
		\hypo{t \cbvEq[\indexOmega] u}
		\hypo{t' \cbvEq[\indexOmega] u'}
		\inferCbnEqEs{t\esub{x}{t'} \cbvEq[\indexOmega] u\esub{x}{u'}}
	\end{prooftree}
}
\end{center}

\begin{example}
    Let $t_0 := \app{x}{\Id}\esub{x}{\app{y}{\Omega}}$ and $t_1 :=
    \app{x}{\Id}\esub{x}{\app{y}{\Id}}$. We have $t_0 \cbvEq[0] t_1$
    and $t_0 \cbvEq[1] t_1$ but not $t_0 \cbvEq[\indexK] t_1$ for any
    $\indexK \geq 2$.
\end{example}

While the same intuition behind stratification is shared between
\CBVSymb and \CBNSymb, their respective implementations turn out to be
quite different, depending on the underlying notion of reduction.

\paragraph*{\textbf{(Stratified) Operational Semantics.}} The notions
of contexts already established, we can now equip the \CBNSymb with
an operational semantics acting \emph{at a distance}~\cite{AK10,AK12}. The
following rewriting rules are the base components of such operational
semantics. 
\begin{equation*}
    \begin{array}{c c c}
        \app{\cbnLCtxt<\abs{x}{s}>}{t}
            \;\mapstoR[\cbnSymbBeta]\;
        \cbnLCtxt<s\esub{x}{t}>
    &\hspace{0.25cm}&
        t\esub{x}{u}
            \;\mapstoR[\cbnSymbSubs]\;
        t\isub{x}{u}
    \end{array}
\end{equation*}
Rule $\cbnSymbBeta$ assumed to be \emph{capture free}, so no free
variable of $t$ is captured by the context $\cbnLCtxt$. The main
difference between \CBNSymb and \CBVSymb is that rule $\cbvSymbSubs$
of the latter is replaced by rule $\cbnSymbSubs$ in the former.

The \defn{stratified reduction} (or
\defn{$\cbnSymbSurface_\indexK$-reduction})
$\cbnArr_{S_\indexK}$ is the $\cbnSCtxt_\indexK$-closure of the two
rewriting rules $\mapstoR[\cbnSymbBeta]$ and $\mapstoR[\cbnSymbSubs]$.
In particular, $\cbnArr_{S_0}$ (\resp $\cbnArr_{S_\indexOmega}$) is
called \defn{surface} (resp.  \defn{full}) \defn{reduction}. 
We write $\cbnArr*_{S_\indexK}$ for the reflexive and
transitive closure of reduction $\cbnArr_{S_\indexK}$.

\begin{lemma}[\cite{AK12}]
    \label{lem:Cbn_Surface_and_Full_Confluence}
    The \CBNSymb \tSurface and \tFull reductions are
    \mbox{confluent}.
\end{lemma}

As in \CBVSymb, \cgiulio{the \CBNSymb  surface reduction is confluent whereas
the stratified reduction one is not}{For $\indexK \notin \{0,\indexOmega\}$, {$\cbnSymbSurface_\indexK$-reduction} is not confluent.
In general, $\cbnArr_{S_\indexI} \,\subsetneq\, \cbnArr_{S_{\indexI+1}} \,\subsetneq\, \cbnArr_{S_\indexOmega}$ for all $i \in \mathbb{N}$}.

\paragraph*{\textbf{Normal Forms.}} Similar to $\CBVSymb$, each level
of the stratified reduction $\cbnArr_{S_\indexK}$ yields a
corresponding notion of \defn{$\cbnSymbSurface_\indexK$-normal form}
and \defn{$\cbnSymbSurface_\indexK$-normalizability}. Normal forms can
be \emph{syntactically} characterized by the following grammar
$\cbnNoS_\indexOmega$, where the subgrammar $\cbnNeS_\indexOmega$
generates special normal forms called \emph{neutral}, which never
create a redex when applied to an argument.

\begin{equation*}
    \begin{array}{c}
        \begin{array}{rcl}
        \cbnNeS_0 &\coloneqq& x
            \vsep \app{\cbnNeS_0}{t}
        \\
        \cbnNeS_{\indexI+1} &\coloneqq& x
            \vsep \app{\cbnNeS_{\indexI+1}}{\cbnNoS_{\indexI}}
                    \qquad (\indexI \in \setIntegers)
        \\
        \cbnNeS_\indexOmega &\coloneqq& x
            \vsep \app{\cbnNeS_\indexOmega}{\cbnNoS_\indexOmega}
    \\[0.2cm]
            \cbnNoS_\indexK &\coloneqq& \abs{x}{\cbnNoS_\indexK}
                \vsep \cbnNeS_\indexK
                \qquad (\indexK \in \setOrdinals)
        \end{array}
    \end{array}
\end{equation*}

\begin{restatable}{lemma}{CbnNoSnCharacterization}
    \LemmaToFromProof{cbn_NnoNn_Characterization}%
Let $\indexK \in \Natinfty$. Then $t \in \setCbnTerms$ is $\cbnSymbSurface_\indexK$-normal iff $t \in
\cbnNoS_\indexK$.
\end{restatable}

\paragraph*{\textbf{Quantitative Typing.}} We now present the
quantitative typing system $\cbnTypeSysBKRV$ based on intersection
types which was introduced in~\cite{KesnerV14}. The typing rule of
system $\cbnTypeSysBKRV$ are defined as follows:
\begin{equation*}
    \begin{array}{ccc}
        \begin{prooftree}
            \inferCbnBKRVVar{x : \mset{\sigma} \vdash x : \sigma}
        \end{prooftree}
    \hspace{1cm}
        \begin{prooftree}
            \hypo{\Pi \cbnTrBKRV \Gamma; x : \M \vdash t : \sigma}
            \inferCbnBKRVAbs{\Gamma \vdash \abs{x}{t} : \M \typeArrow \sigma \quad}
        \end{prooftree}
    \\[0.5cm]
        \begin{prooftree}
            \hypo{\Pi_1 \cbnTrBKRV \Gamma_1 \vdash t : \mset{\tau_i}_{i \in I} \typeArrow \sigma}
            \hypo{\Pi_2^i \cbnTrBKRV \Gamma_2^i \vdash u : \tau_i}
            \delims{\left(}{\right)_{i \in I}}
            \inferCbnBKRVApp{\Gamma_1 +_{i \in I} \Gamma_2^i &\vdash \app{t}{u} : \sigma\hspace{2.8cm}}
        \end{prooftree}
    \\[0.5cm]
        \begin{prooftree}
            \hypo{\Pi_1 \cbnTrBKRV \Gamma_1, x : \mset{\tau_i}_{i \in I} \vdash t : \sigma}
            \hypo{\Pi_2^i \cbnTrBKRV \Gamma_2^i \vdash u : \tau_i}
            \delims{\left(}{\right)_{i \in I}}
            \inferCbnBKRVEs{\Gamma_1 +_{i \in I} \Gamma_2^i &\vdash t\esub{x}{u} : \sigma\hspace{2.5cm}}
        \end{prooftree}
    \end{array}
\end{equation*}
with $I$ finite in rules $(\cbnBKRVAppRuleName)$ and
$(\cbnBKRVEsRuleName)$. As with system $\cbvTypeSysBKRV$, a
\textbf{(type) derivation} is a tree obtained by applying the
(inductive) typing rules of system $\cbnTypeSysBKRV$. Similarly, the
notation $\Pi \cbnTrBKRV \Gamma \vdash t : \sigma$ means there is a
derivation of the judgment $\Gamma \vdash t : \sigma$ in system
$\cbnTypeSysBKRV$ and we say that $t$ is
\textbf{$\cbnTypeSysBKRV$-typable} if $\Pi \cbnTrBKRV \Gamma \vdash t
: \sigma$ holds for some $\Gamma, \sigma$.

System $\cbnTypeSysBKRV$ \emph{characterizes} surface normalization
(see \Cref{lem:cbnBKRV_characterizes_meaningfulness}).


\paragraph*{\textbf{Meaningfulness.}} We start by defining 
\CBNSymb-meaningfulness.
\begin{definition}\label{def:cbn-meaningful}
    A term $t\in \setCbnTerms$ is \CBNSymb \defn{meaningful} if there
    is a testing context $\cbnTCtxt$ such that $\cbnTCtxt<t>
    \cbnArr*_{S_0} \Id$ (with $\cbvTCtxt \Coloneqq \Hole \vsep
    \app[\,]{\cbvTCtxt}{u} \vsep
    \app[\,]{(\abs{x}{\cbvTCtxt})}{u}$).\footnotemark
    \footnotetext{Usually, \CBNSymb-meaningfulness (aka \emph{solvability})
    is defined using contexts of the form $(\abs{x_1\dots
    x_m}\Hole)N_1\dots N_n$ ($m,n \geq 0$)
    \cite{Barendregt75,barendregt84nh,RoccaP04}, instead of testing
    contexts. It is easy to check that the two definitions are
    equivalent in \CBNSymb. The benefit of our definition is that the
    same testing contexts are also used to define
    \CBVSymb-meaningfulness~(\Cref{sec:cbv}).}
\end{definition}
As for the  \CBVSymb-case, \CBNSymb-meaningfulness can be characterized \emph{operationally},
through the notion of surface-normalization, and \emph{logically},
through typability in type system $\cbnTypeSysBKRV$. 

\begin{lemma}[\CBNSymb Meaningful Characterizations~\cite{BucciarelliKR21,BKRV20}]
    \label{lem:cbnBKRV_characterizes_meaningfulness}%
    Let $t \in \setCbvTerms$. Then, the following characterizations hold:
    \begin{itemize}[leftmargin=6em]
    \item[(Operational)] $t$ is \CBNSymb-meaningful iff $t$ is
        \CBNSymb surface-normalizing.

    \item[(Logical)] $t$ is \CBNSymb-meaningful iff $t$ is
        $\cbnTypeSysBKRV$-typable.
  \end{itemize}
\end{lemma}

\paragraph*{\textbf{Surface Genericity.}} As in $\CBVSymb$, since
the typing system characterizes meaningfulness, we can
prove a typed genericity result.

\begin{restatable}[\CBNSymb Typed Genericity]{lemma}{CbnTypedSurfaceGenericity}
    \LemmaToFromProof{Cbn_t_meaningless_and_Pi_|>_F<t>_==>_Pi'_|>_F<u>}%
    Let $t \in \setCbnTerms$ be \CBNSymb-meaningless. If $\Pi
    \cbnTrBKRV \Gamma \vdash \cbnCCtxt<t> : \sigma$, then there is
    $\Pi' \cbnTrBKRV \Gamma \vdash \cbnCCtxt<u> : \sigma$ for all $u
    \in \setCbnTerms$.%
\end{restatable}

Likewise, it can be transformed into a qualitative surface genericity
using the logical characterization
(\Cref{lem:cbnBKRV_characterizes_meaningfulness}). 

\begin{restatable}[\CBNSymb Qualitative Surface Genericity]{theorem}{CbnQualitativeSurfaceGenericity}
    \LemmaToFromProof{Cbn_Qualitative_Surface_Genericity}%
    Let  $t \in \setCbnTerms$ be \CBNSymb-meaningless. If
    $\cbnCCtxt<t>$ is \CBNSymb-meaningful then $\cbnCCtxt<u>$ is also
    \CBNSymb-meaningful for all $u \in \setCbnTerms$.
\end{restatable}

\subsection{Call-by-Name Genericity}
\label{sec:genericity-cbn}

\paragraph*{\textbf{Partial \CBNSymb-Calculus}}
\label{sec:Cbn_Partial_Terms}

The set $\setCbnAprxTerms$ of \textbf{\tPartTerm+} of the \tPartial\
\CBNSymb-calculus, denoted by the same symbols used for terms but in a
bold font, is given by the following inductive definition:
\begin{equation*}
    \begin{array}{r rcl}
        \textbf{(\tPartial^ terms)} \quad
        &\aprxT, \aprxU&\Coloneqq& x \in \setCbnVariables
            \vsep \abs{x}{\aprxT}
            \vsep \app{\aprxT}{\aprxU}
            \vsep \aprxT\esub{x}{\aprxU}
            \vsep \cbnAprxBot
    \end{array}
\end{equation*}
As in the \CBNSymb-calculus, there is a stratified operational, but
now defined by means of \tPartial contexts. The set of
\textbf{\tPartLCtxt+} $(\cbnAprxLCtxt)$ and \textbf{\tPartSnCtxt+}
$(\cbnAprxSCtxt_\indexK)$ for some $\indexK \in \setOrdinals$ are
inductively defined as follows (where $\indexI \in \setIntegers$):
\begin{equation*}
        \begin{array}{rcl}
            \cbnAprxLCtxt &\;\Coloneqq\;& \Hole
                \vsep \cbnAprxLCtxt\esub{x}{t}
        \\[0.2cm]
            \cbnAprxSCtxt_0 &\;\Coloneqq\;& \Hole
                \vsep \abs{x}{\cbnAprxSCtxt_0}
                \vsep \app[\,]{\cbnAprxSCtxt_0}{t}
                \vsep \cbnAprxSCtxt_0\esub{x}{t}
        \\
            \cbnAprxSCtxt_{\indexI+1} &\;\Coloneqq\;& \Hole
                \vsep \abs{x}{\cbnAprxSCtxt_{\indexI+1}}
                \vsep \app[\,]{\cbnAprxSCtxt_{\indexI+1}}{t}
                \vsep \app[\,]{t}{\cbnAprxSCtxt_\indexI}
                \vsep \cbnAprxSCtxt_{\indexI+1}\esub{x}{t}
                \vsep t\esub{x}{\cbnAprxSCtxt_\indexI}
        \\
            \cbnAprxSCtxt_\indexOmega &\;\Coloneqq\;& \Hole
                \vsep \abs{x}{\cbnAprxSCtxt_\indexOmega}
                \vsep \app[\,]{\cbnAprxSCtxt_\indexOmega}{t}
                \vsep \app[\,]{t}{\cbnAprxSCtxt_\indexOmega}
                \vsep \cbnAprxSCtxt_\indexOmega\esub{x}{t}
                \vsep t\esub{x}{\cbnAprxSCtxt_\indexOmega}
        \end{array}
\end{equation*}
The following rewriting rules are the base components of our reduction
relations. Any \tPartTerm having the shape of the left-hand side of
one of these three rules is called a \textbf{redex}.
\begin{equation*}
    \begin{array}{c c c}
        \app{\cbnAprxLCtxt<\abs{x}{\aprxT}>}{\aprxU}
            \;\mapstoR[\cbnAprxSymbBeta]\;
        \cbnAprxLCtxt\cbnCtxtPlug{\aprxT\esub{x}{\aprxU}}
    &\hspace{0.25cm}&
        \aprxT\esub{x}{\aprxU}
            \;\mapstoR[\cbnAprxSymbSubs]\;
        \aprxT\isub{x}{\aprxU}
    \end{array}
\end{equation*}

Building upon the established pattern used for terms, the
\defn{\tPartSnRed} (or $\cbnSymbSurface_\indexK$-reduction)
$\cbnAprxArr_{S_\indexK}$ is defined  as the \tPartial
$\indexK$-stratified closure ($\cbnAprxSCtxt_\indexK$-closure) of the
two rewriting rules $\mapstoR[\cbnAprxSymbBeta]$ and
$\mapstoR[\cbnAprxSymbSubs]$. Its reflexive-transitive closure is
denoted by $\cbnAprxArr*_{S_\indexK}$.
\begin{example} \label{example:cbn_partial_reduction}  $\aprxT_0 =
    \big(\abs{y}{\app{\cbnAprxBot}{w}}\big)\esub{w}{\abs{z}{\cbnAprxBot}}
    \cbnAprxArr_{S_\indexK}
    \abs{y}{\big(\app[\,]{\cbnAprxBot}{\abs{z}{\cbnAprxBot}}\big)}$ $=
    \aprxT_1$ for any $\indexK \in \setOrdinals$. In particular,
    $\aprxT_1$ is $\cbnAprxSymbSurface_0$-normal form.
\end{example}

\paragraph*{\textbf{Meaningful Approximation.}} The 
definition is the same as in \CBVSymb, just adapting the notion of meaningfulness
to \CBNSymb. Indeed, the \defn{\CBNSymb-meaningful
approximant} $\cbnTApprox{t}$ of $t \in \setCbnTerms$ is
$\cbnAprxBot$ if $t$ is \CBNSymb-meaningless, and is defined
inductively on $t$ \mbox{otherwise}:
\begin{align*}\hspace{-0.3cm}
    \hspace{-0.25cm}
        \cbnTApprox{x}
            &\coloneqq x
    &&&
    \hspace{-0.25cm}
        \cbnTApprox{\app{t}{u}}
            &\coloneqq \app{\cbnTApprox{t}}{\cbnTApprox{u}}
    \\
    \hspace{-0.25cm}
        \cbnTApprox{\abs{x}{t}}
            &\coloneqq \abs{x}{\cbnTApprox{t}}
    &&&
    \hspace{-0.25cm}
        \cbnTApprox{t\esub{x}{u}}
            &\coloneqq \cbnTApprox{t}\esub{x}{\cbnTApprox{u}}
\end{align*}

The meaningful approximation of a term $t$ approximates $t$:%
\begin{restatable}{lemma}{CbnMeaningfulApproximationLeqTerm}
    \LemmaToFromProof{Cbn_MA(t)_<=_t}%
    Let $t \in \setCbnTerms$, then $\cbnTApprox{t} \cbnAprxLeq t$.
\end{restatable}


\begin{restatable}[\CBNSymb Approximant Genericity]{lemma}{CbnApproximantGenericity}
    \LemmaToFromProof{cbn_Approximant_Genericity}%
    Let $\cbnCCtxt$ be a full  context. If $t \in \setCbnTerms$
    is \CBNSymb-meaningless, then $\cbnTApprox{\cbnCCtxt<t>}
    \cbnAprxLeq \cbnTApprox{\cbnCCtxt<u>}$ for all $u \in
    \setCbnTerms$.
\end{restatable}

\paragraph*{\textbf{Dynamic Approximation and Lifting.}} To replicate
the proof of quantitative stratified genericity on the
\CBNSymb-calculus, we now need to be able to approximate term
reductions and lift \tPartial reductions. Fortunately, both instances
of
\Cref{ax:Gen_t_->Sn_u_and_t_<=_t'_==>_t'_->Sn_u'_and_u_<=_u',ax:Gen_t_->Sn_u_=>_OA(t)_->*Sn_|u_OA(u)}
for \CBNSymb hold. Consequently, the \tPartial\ \CBNSymb-calculus can
effectively be used to approximate reductions of the
\CBNSymb-calculus.

\paragraph*{\textbf{Observability.}} We now check that the
\tPartial\ \CBNSymb-calculus correctly captures the meaningful and
observable outcomes.

For $\indexK \in
\setOrdinals$, the set $\cbnAprxBnoS_\indexK$ of \tPartial terms introduced below aims to provide
perfect approximation of them, being
$\cbnAprxSymbSurface_\indexK$-normal forms without any occurrence of
$\cbnAprxBot$ at depth $\indexK$.
\begin{equation*}
    \begin{array}{rcl}
        \cbnAprxBneS_0 &\coloneqq& x
            \vsep \app{\cbnAprxBneS_0}{\aprxT}
    \\
        \cbnAprxBneS_{\indexI+1} &\coloneqq& x
            \vsep \app{\cbnAprxBneS_{\indexI+1}}{\cbnAprxBnoS_\indexI}
                    \qquad (\indexI \in \setIntegers)
    \\
        \cbnAprxBneS_\indexOmega &\coloneqq& x
            \vsep \app{\cbnAprxBneS_\indexOmega}{\cbnAprxBnoS_\indexOmega}
    \\[0.2cm]
        \cbnAprxBnoS_\indexK &\coloneqq& \abs{x}{\cbnAprxBnoS_\indexK}
            \vsep \cbnAprxBneS_\indexK
            \qquad (\indexK \in \setOrdinals)
    \end{array}
\end{equation*}

Remarkably, the \tPartial terms in
$\cbnAprxBnoS_\indexK$ are precise enough to perfectly observe the
$\cbnSymbSurface_\indexK$-normal forms and entirely determine them up
to depth $\indexK$, so that the \CBNSymb instances of
\Cref{ax:Gen_t_in_BnoSn_==>_OA(t)_in_noSn,ax:Gen_|t_in_bnoSn_and_t_<=_u_==>_u_in_BnoSn}
hold.

\paragraph*{\textbf{Quantitative Stratified Genericity.}} We now have
all the material to enhance our \CBNSymb qualitative surface
genericity (\Cref{lem:Cbn_Qualitative_Surface_Genericity}) into a
stronger quantitative stratified genericity result.

\begin{restatable}[\CBNSymb Quantitative Stratified Genericity]{theorem}{CbnQuantitativeStratifiedGenericity}
    \LemmaToFromProof{cbn_Quantitative_Stratified_Genericity}%
    \label{t:quant-stratified-genericity}
    Let $k \in \Natinfty$. If $t \in \setCbnTerms$ is
    \CBNSymb-meaningless and $\cbnCCtxt<t>$ is
    $\cbnSymbSurface_\indexK$-normalizing, then there is $i \in
    \mathbb{N}$ such that, for every $u \in \setCbnTerms$, there are
    $t', u' \in \cbnNoS_\indexK$ such that $\cbnCCtxt<t>
    \cbnArr^{i}_{S_\indexK} t' \;\cbvEq[\indexK]\; u'
    \;\sideset{^{i}_{\cbnSymbSurface_\indexK\!\!\!\!}}{}\revArr\;
    \cbnCCtxt<u>$. 
\end{restatable}

The proof of this theorem follows the exact same schematic as
\Cref{lem:Cbv_Quantitative_Stratified_Genericity}. A particular case
is given by the following statement since $\indexK$-equality for
$\indexK = \indexOmega$ is just the equality between terms.

\begin{corollary}[\CBNSymb Quantitative Full Genericity]
    \label{lem:Cbn_Quantitative_Full_Genericity}%
    If $t \in \setCbnTerms$ is \CBNSymb-meaningless and $\cbnCCtxt<t>
    \cbnArr*_{S_\indexOmega} s$ where $s$ is
    $\cbnSymbSurface_\indexOmega$-normal, then there is $i \in
    \mathbb{N}$ such that, for every $u \in \setCbnTerms$,
    $\cbnCCtxt<t> \cbnArr^{i}_{S_\indexOmega} s
    \;\sideset{^{i}_{\mathrm{n}_\indexOmega\!\!\!\!}}{}\revArr\;
    \cbnCCtxt<u>$.
\end{corollary}

\section{Theories}
\label{sec:theories}

In this section, we prove that
the meaningful (resp. observational)
theory $\genThH$ (resp. $\genThH*$) generated by \CBVSymb
and \CBNSymb is \emph{consistent} (resp. \emph{consistent}
  and \emph{sensible}). This is when  genericity comes into play. To streamline
the presentation of our approach and avoid unnecessary repetitions, we
employ a \emph{unified reasoning} methodology for both frameworks. To achieve
this, we begin by presenting the proofs of consistency and sensibility
for a \emph{generic} calculus (\Cref{sec:generic-theories}), for which
we assume a set of axioms, mirroring our approach when describing the
methodology used to obtain genericity (\Cref{sec:methodology}).
Then, we instantiate the generic calculus with \CBVSymb\ (\Cref{sec:consistent-sensible-cbv})
and \CBNSymb\ (\Cref{sec:consistent-sensible-cbn}),
showing that each of these instantiations enjoys the axioms stated
before for the generic calculus.

\subsection{Generic Theories}
\label{sec:generic-theories}
We introduce here a generic calculus and three associated
theories: the computational theory $\genThLambda$, the meaningful
theory $\genThH$ and the observational theory $\genThH*$. We provide an
axiomatized proof of the consistency of $\genThH$, and of the
consistency and sensibility~of~$\genThH*$.%

Let us consider a \defn{generic calculus}, based on the set of terms
$\setGenTerms$ equipped with a notion of stratified contexts for all
$\indexK \in \setOrdinals$. In particular, \emph{full contexts} ($\indexK =
\indexOmega$) are denoted by $\genCCtxt$. We also assume the
generic calculus to be equipped with a set $\setrules$ of rewrite
rules. 
The following definitions should
be subscripted by $\setrules$ but we omit it for better readability.
As in \Cref{sec:methodology}, we also assume the generic calculus to
be equipped with a notion of \emph{stratified reduction} $\genArr_\indexK$
generated by the closure of the rules in $\setrules$ by stratified
contexts of level $\indexK \in \setOrdinals$. This gives rise to
corresponding notions of $\indexK$-normal form and
$\indexK$-normalization. In particular the generic calculus  is
equipped with a notion of surface reduction (level $\indexK = 0$) as
well as a notion of full reduction (level $\indexK = \indexOmega$). We
also assume that the generic calculus is equipped with an appropriate
notion~of~\emph{meaningfulness}.

A \defn{theory} is a set of equations of the form $t \genThHEq u$ with
$t,u \in \setGenTerms$. A \defn{congruence} $\genThCongr$ is the
smallest theory containing a set $\genThBaseEquations$ of \emph{base
equations}, and closed by reflexivity, symmetry, transitivity and full
contexts. Alternatively, $\genThCongr$ can be specified by the rules
below:
\begin{equation*}
   \begin{array}{c}
        \begin{prooftree}
            \hypo{\phantom{t \doteq u}}
            \inferGenThCongrRefl[1]{t}{t}
        \end{prooftree}
    \hspace{0.5cm}
        \begin{prooftree}
            \hypoGenThCongr{t}{u}
            \inferGenThCongrSym{u}{t}
        \end{prooftree}
    \hspace{0.5cm}
        \begin{prooftree}[separation = 1em]
            \hypoGenThCongr{t}{s}
            \hypoGenThCongr{s}{u}
            \inferGenThCongrTrans{t}{u}
        \end{prooftree}
   \\[0.3cm]
        \begin{prooftree}
            \hypoGenThCongr{t}{u}
            \inferGenThCongrCtxt{\genCCtxt<t>}{\genCCtxt<u>}
        \end{prooftree}
    \hspace{0.5cm}
        \begin{prooftree}
            \inferGenThCongrStep{t}{u}
        \end{prooftree}
   \end{array}
\end{equation*}
Given a theory $\genThBaseEquations$, we write
  $\genThCongrTail{t}{u}$ if $t \genThCongrEq u \in \genThCongr$, and
  $\genThBaseEquations \not\vdash t \genThCongrEq u$ if $t
  \genThCongrEq u \notin \genThCongr$. A congruence $\genThCongr$ is
\defn{consistent}, if it does not equate all terms (\ie there are $t,
u \in \setGenTerms$ such that $\genThBaseEquations \not\vdash t
\genThCongrEq u$).

The \textbf{computational theory} $\genThLambda$ is the congruence
$\genThCongr_\genThLambdaBaseEquations$ (recall that
$\genThLambdaBaseEquations$ is the set of rewrite rules of the generic calculus). 
We write $t \genSymArr_{\setrules} u$ for the reflexive,
symmetric, transitive closure of the reduction $\genArr_\indexOmega$.
As expected, the congruence $\genThLambda$ axiomatizes the equivalence
$\genSymArr_{\setrules}$: 

\begin{restatable}{lemma}{GenThTCharacterization}
    \label{lem:Theories_lambda_|-_t_=_u_iff_t_=F_u}%
    Let $t, u \in \setGenTerms$, then $\genThLambdaTail{t}{u}$ if and
    only if $t \genSymArr_{\setrules} u$.
\end{restatable}

A \defn{\genLambdaTheoryTxt} $\genThT$ is any congruence
containing the computational theory $\genThLambda$, that is,  $\genThT = \genThCongr$ for some $\genThBaseEquations \supseteq \setrules$. 
A \genLambdaTheoryTxt\ $\genThCongr$ is
\defn{sensible} if it equates all meaningless terms, that is, $\genThBaseEquations \vdash t \genThCongrEq u$ if $t$ and $u$ are meaningless. The
\defn{meaningful theory $\genThH$} is the smallest sensible
\genLambdaTheoryTxt (\ie $\genThH \coloneqq \genThCongr_\genThHBaseEquations$
where $\genThHBaseEquations \coloneqq \genThLambdaBaseEquations \cup \{(t, u)
\vsep t, u \text{ meaningless}\}$).

Working directly on $\genThH$ is not easy, but it is possible to
refine its specification following \cite{barendregt84nh} (see
\Cref{sec:ProofTheories} for details).

In order to abstractly prove consistency of the \genLambdaTheoryTxt
$\genThH$, we first assume consistency of the \genLambdaTheoryTxt
$\genThLambda$ as an \emph{axiom}.

\begin{assumption}[\textbf{Consistency}]
    \label{ax:thLambda_Consistent}%
    The \genLambdaTheoryTxt $\genThLambda$ is consistent.
\end{assumption}

Why  this assumption is reasonable? The  \genLambdaTheoryTxt $\genThH$
is built as an extension of $\genThLambda$ to ensure its invariance
under reduction. If $\genThLambda$ is inconsistent then all its
extensions (including $\genThH$) are inconsistent as well, regardless
of the definition of meaningfulness. Taking $\genThLambda$
to be consistent as an axiom
is therefore reasonable  in our framework. 

We abstractly prove that consistency of $\genThH$ can be reduced to
that of $\genThLambda$, following the same ideas
in~\cite{barendregt84nh}. For that, we start by showing in particular
that $\genThH$ and $\genThLambda$ both agree on $\indexOmega$-normal
forms, a forthcoming property
(\Cref{lem:Theories_T_|-_t_=_u_and_u_FNF_==>_lambda_|-_t_=_u}) which
uses the fact that the generic calculus satisfies the following
\cref{ax:Confluence,ax:Full_Genericity}.

\begin{restatable}[\textbf{Confluence and Normal Form}]{assumption}{AssumptionTheoryFullConfluence}
    \label{ax:Confluence}%
    Reduction $\genArr_\indexOmega$ is confluent and
    there exists an $\indexOmega$-normal form.
\end{restatable}

Indeed, normal forms represent the results of computations, thus being
without $\indexOmega$-normal forms would mean that the calculus never
yields a result (also, variables are $\indexOmega$-normal in any
reasonable rewriting system). Moreover, confluence of
$\genArr_\indexOmega$ implies that two evaluations from the same term
cannot yield different outcomes. 
Hence, \Cref{ax:Confluence} is a perfectly reasonable assumption.

\begin{assumption}[\textbf{Full Genericity}]
    \label{ax:Full_Genericity}%
    If $t \in \setGenTerms$ is meaningless and $\genCCtxt<t>
    \cbvArr*_\indexOmega s$ where $s$ is $\indexOmega$-normal, then
    for every $u \in \setGenTerms$, 
    $\genCCtxt<t> \cbvArr*_\indexOmega s
    \;\sideset{^*_\indexOmega}{}\revArr\; \genCCtxt<u>$.
\end{assumption}


Since $\genThH$ is the smallest sensible \genLambdaTheoryTxt, proving
$\genCCtxt<t> \genThLambdaEq \genCCtxt<u>$ in $\genThH$ can be done either by using the rules for
$\genThLambda$ or --thanks to sensibility-- by replacing $t$ with $u$ when they are both meaningless.
Full genericity (\Cref{ax:Full_Genericity}) allows one to do that directly in
$\genThLambda$, in the special case where $\genCCtxt<t>$ or $\genCCtxt<u>$ is $\indexOmega$-normalizing.
Therefore,

\begin{restatable}{lemma}{GenTheoryHAndLambdaAgreeFNF}
    \LemmaToFromProof{Theories_T_|-_t_=_u_and_u_FNF_==>_lambda_|-_t_=_u}%
    If $u \in \setGenTerms$ has a $\indexOmega$-normal form, then for
    any $t \in \setGenTerms$, if $\genThHTail{t}{u}$ then
    $\genThLambdaTail{t}{u}$.
\end{restatable}

\begin{theorem}
    \label{lem:Theories_H_Consistent}%
    The \genLambdaTheoryTxt $\genThH$ is consistent.
\end{theorem}
\begin{proof}
By \Cref{ax:thLambda_Consistent}, there exist two terms $t_1, t_2 \in
\setGenTerms$ such that $t_1 \genThLambdaEq t_2$ is not provable in
$\genThLambda$. Moreover, using \Cref{ax:Confluence}, there exists an
$\indexOmega$-normal form $u \in \setGenTerms$. Suppose by absurd that
$\genThH$ is inconsistent, then $\genThHTail{t_1}{u}$ and
$\genThHTail{u}{t_2}$. Thus $\genThLambdaTail{t_1}{u}$ and
$\genThLambdaTail{u}{t_2}$ using
\Cref{lem:Theories_T_|-_t_=_u_and_u_FNF_==>_lambda_|-_t_=_u}, hence
$\genThLambdaTail{t_1}{t_2}$ by transitivity which
contradicts~\Cref{ax:thLambda_Consistent}. One then concludes that $\genThH$ is consistent.
\end{proof}

Now we study the \defn{observational theory $\genThH*$} associated
with a generic calculus, which is defined by the following set:
\begin{equation*}
    \genThH* \coloneqq \{t \genThHEq u \vsep \forall\; \genCCtxt, \; \genCCtxt<t> \text{ meaningful } \Leftrightarrow \genCCtxt<u> \text{ meaningful}\}
\end{equation*}

To prove that the theory $\genThH*$ is a consistent and
sensible \genLambdaTheoryTxt, we use the following complementary
axioms, which are natural and self-explained, as well as consistency of $\genThH$ (\Cref{lem:Theories_H_Consistent}).
    
\begin{assumption}[\textbf{Meaningfulness Stability}]
    \label{ax:Meaningfulness_Stable_Reduction}%
    Let $t, u \in \setGenTerms$ such that $t \genArr_\indexOmega u$,
    then $t$ is meaningful if and only if $u$ is meaningful.
\end{assumption}

\begin{assumption}[\textbf{Surface Genericity}]
    \label{ax:Surface_Genericity}%
    If $t \in \setGenTerms$ is meaningless and $\genCCtxt<t>$
    is meaningful, then for any $u \in \setGenTerms$, $\genCCtxt<u>$
    is meaningful.
\end{assumption}

\begin{assumption}[\textbf{Meaningful Existence}]
    \label{ax:Exists_Meaningless}%
    There is a meaningless term in the generic calculus.
\end{assumption}


\begin{restatable}{theorem}{GenTheoryHStarConsistent}
    \label{thm:Theory_H*}%
    $\genThH*$ is a sensible and consistent \genLambdaTheoryTxt. 

\end{restatable}
\begin{proof}
    (\genLambdaTheoryTxt) \cgiulio{By induction on $\genThLambdaTail{t}{u}$
    using \Cref{ax:Meaningfulness_Stable_Reduction}}{It suffices to show $\genThH* =
    \genThCongr_{\genThH*}$ and $\genThLambdaBaseEquations \subseteq \genThH*$.
	If $(t,u) \in \genThLambdaBaseEquations$ then, for every full context $\genCCtxt$, we have $\genCCtxt<t> \genArr_\indexOmega \genCCtxt<u>$ and so $t \genThLambdaEq u \in \genThH*$ by \Cref{ax:Meaningfulness_Stable_Reduction}.
	By definition of $\genThH*$\!, it is easy to \mbox{check that $t \!\genThLambdaEq\! u \in \genThH*$ iff $t \!\genThLambdaEq\! u \in \genThCongr_{\genThH*}$}}.

    (Sensibility) \giulio{As $\genThH*$ is a \genLambdaTheoryTxt, it suffices to show that $\genThH*$ equates all meaningless terms}.
    Let $t, u \in \setGenTerms$ be meaningless. By
    surface genericity (\Cref{ax:Surface_Genericity}), $\genCCtxt<t>$
    is meaningful iff so is $\genCCtxt<u>$, thus $t \doteq u \in
    \genThH^*$.

    (Consistency) As $\genThH$ is consistent
    (\Cref{lem:Theories_H_Consistent}), $\mathsf{H} \not\vdash t \genThLambdaEq
    u$ for some $t, u \in
    \setGenTerms$. This means that $t$ or $u$ --say $t$-- is meaningful
    (otherwise $\genThHTail{t}{u}$). By \Cref{ax:Exists_Meaningless},
    there is a meaningless $s \in \setGenTerms$. For $\genCCtxt =
    \Hole$, $\genCCtxt<t> = t$ is meaningful and $\genCCtxt<s> = s$ is
    meaningless, then $t \genThLambdaEq s \notin \genThH* =
    \genThCongr_{\genThH*}$ thus $\genThH* \not\vdash t
    \doteq s$, that is, $\genThH*$~is~consistent.
\end{proof}

\subsection{\CBVSymb Theories}
\label{sec:consistent-sensible-cbv}

We now instantiate the generic theories
above to the \CBVSymb-calculus\giulio{: let $\cbvThH$ (\resp $\cbvThH*$) be the instantiation of $\genThH$ (resp. $\genThH*$) obtained by replacing ``meaningful'' with ``\CBVSymb-meaningful'' (see \Cref{def:cbv-meaningful}) in its definiton.}
\cgiulio{, thus proving}{Thus, we prove} consistency (\resp and
sensibility) of $\cbvThH$ (\resp $\cbvThH*$). We also prove
that $\cbvThH*$ is the \emph{unique} \emph{maximal} consistent
\cbvLambdaTheoryTxt extending $\cbvThH$.

\begin{theorem}
    The \CBVSymb-calculus verifies
    \Cref{ax:Confluence,ax:Full_Genericity,ax:thLambda_Consistent,ax:Meaningfulness_Stable_Reduction,ax:Surface_Genericity,ax:Exists_Meaningless}.
\end{theorem}

\begin{proof}
    For \Cref{ax:Confluence} see \cite{AP12}.
    \Cref{ax:Full_Genericity}
    is a special case of \Cref{lem:Cbv_Quantitative_Full_Genericity}.
    \Cref{ax:Meaningfulness_Stable_Reduction} is proved using Full
    Subject Reduction and Expansion in system $\cbvTypeSysBKRV$~\cite{BKRV20}  and
    \Cref{lem:cbvBKRV_characterizes_meaningfulness}. For
    \Cref{ax:Surface_Genericity} we use
    \Cref{lem:Cbv_Qualitative_Surface_Genericity}. For
    \Cref{ax:Exists_Meaningless}: the term $\Omega$ is $\cbvArr_{S_0}$
    divergent thus by \Cref{lem:cbvBKRV_characterizes_meaningfulness},
    one concludes that it is \CBVSymb-meaningless.
\end{proof}

\begin{corollary}\label{cor:consistency}
    \begin{enumerate}
    \item The \cbvLambdaTheoryTxt $\cbvThH$ is consistent. 

    \item The theory $\cbvThH*$ is a sensible and
    consistent \cbvLambdaTheoryTxt.
    \end{enumerate}
\end{corollary}


Sensibility means that $\cbvThH*$ is an extension of $\cbvThH$.
We prove that $\cbvThH*$ is actually the unique consistent
maximal \cbvLambdaTheoryTxt extending $\cbvThH$,
where \defn{maximal} means that, for every $t, u \in \setGenTerms$, either $\cbvThH* \vdash t \doteq u$ or $\cbvThH*$ extended with the base equation $t \genThHEq u$ is inconsistent.
For that we need the following property.

\begin{restatable}{lemma}{CbvTheoryHExtended}
    \LemmaToFromProof{Cbv_Theory_H_Extended_Single}%
    Let $t, u \in \setCbvTerms$. If $\cbvThH$ is extended with the
    base equation $t \genThHEq u$ is consistent, then for all full
    context $\cbvCCtxt$, one has that $\cbvCCtxt<t>$ is
    \CBVSymb-meaningful iff $\cbvCCtxt<u>$ is \CBVSymb-meaningful.
\end{restatable}

By \Cref{lem:Cbv_Theory_H_Extended_Single} and sensibility of
$\cbvThH*$ (\Cref{thm:Theory_H*}), we have:

\begin{restatable}{corollary}{CbvTheoryHStarHPComplete}
    \LemmaToFromProof{cbvTheories_HStar_HP_Complete}%
    \label{cor:maximality}%
    The theory $\cbvThH*$ is the unique maximal consistent and sensible
    \cbvLambdaTheoryTxt.
\end{restatable}

The \cbvLambdaTheoryTxt $\cbvThH*$ equates \emph{more} then
$\cbvThH$. For example, $\Id \cbvThHEq \abs{x}{\abs{y}{\app{x}{y}}}$
is provable in $\cbvThH*$ but not in $\cbvThH$.

Finally, we can now prove our last contribution: the
\cbvLambdaTheoryTxt $\cbvThH*$ coincides with the observational
equivalence $\cong$  defined in \Cref{sec:cbv}. To prove this, we
define the auxiliary notion of \defn{open-observational equivalence}
$\cong^o$ in \CBVSymb: given $t, u \in \setCbvTerms$, $t \cong^o u$ if
for \emph{every} full context $\cbvCCtxt$, $\cbvCCtxt<t>
\cbvArr*_{S_\indexOmega} v_1$ iff $\cbvCCtxt<u>
\cbvArr*_{S_\indexOmega} v_2$ for some values $v_1, v_2 \in
\setCbvValues$. Note that, differently from $\cong$, $\cong^o$
quantifies over all full contexts and not only on \emph{closing} full
contexts, hence $\cong^o \,\subseteq\, \cong$.
 
%

\begin{corollary}\label{cor:operational-equivalence}
	Let $t, u \in \setCbvTerms$, then $t \cong u$
	iff $\cbvThH* \vdash t  \genThTEq u$.
\end{corollary}

\begin{proof}
	$\cbvThH* \vdash t  \genThTEq
	u$ means that for every full context $\cbvCCtxt$,
	$\cbvCCtxt<t>$ is meaningless iff $\cbvCCtxt<u>$
	is. By quantitative stratified genericity with
	$\indexK = 0$
	(\Cref{lem:Cbv_Quantitative_Stratified_Genericity}), this is
	equivalent to say that, for every full context $\cbvCCtxt$,
	$\cbvCCtxt<t> \cbvArr*_{S_0} v$ for some value $v$ iff
	$\cbvCCtxt<u> \cbvArr*_{S_0} v'$ for some value $v'$, that is, $t
	\cong^o u$. Thus, $\cbvThH* \vdash t  \genThTEq u$ if and only if
	$t \cong^o u$. 
	
	Therefore, as $\cong^o \,\subseteq\, \cong$, if $\cbvThH*
	\vdash t \genThTEq u$ then $t \cong u$, in particular $\cong$ is sensible (since so is $\cbvThH*$, \Cref{cor:consistency}).
	As $\cong$ is a consistent (clearly,
	$\Id \not\cong \Omega$) $\lambda$-theory (trivial), one
	has that $\cbvThH* \vdash t  \genThTEq u$ if and only if $t
	\cong u$, by maximality of $\cbvThH*$
	(\Cref{cor:maximality}).
\end{proof}

\Cref{cor:consistency,cor:maximality,cor:operational-equivalence} are among our \mbox{main contributions}.
\giulio{In particular, the fact that $\cbvThH* = \ \cong$ means that two different approaches to define a semantics in \CBVSymb actually coincide. 
	This further backs up the idea that what we call \CBVSymb-meaningfulness approppriately represents meaningfuness in \CBVSymb.}

\subsection{\CBNSymb Theories}
\label{sec:consistent-sensible-cbn}


We now instantiate the generic theories of
\Cref{sec:generic-theories} to the
\CBNSymb-calculus\cgiulio{,}{: let $\cbnThH$ (\resp $\cbnThH*$) be the instantiation of $\genThH$ (\resp $\genThH*$) obtained by replacing ``meaningful'' with \CBNSymb-meaningful (see \Cref{def:cbn-meaningful})} \cgiulio{thus showing}{Thus, we prove} that $\cbnThH$ (\resp $\cbnThH*$) is
consistent (\resp consistent and sensible). 

\begin{theorem}
    The \CBNSymb-calculus verifies
    \Cref{ax:Confluence,ax:Full_Genericity,ax:thLambda_Consistent,ax:Meaningfulness_Stable_Reduction,ax:Surface_Genericity,ax:Exists_Meaningless}.
\end{theorem}
\begin{proof} 
    For \Cref{ax:Confluence} see~\cite{BKRV23}. 
    For
    \Cref{ax:Full_Genericity}, it is a special case of~\Cref{t:quant-stratified-genericity}. 
    For \Cref{ax:Meaningfulness_Stable_Reduction}  is proved using
    Full Subject Reduction and Expansion in system $\cbnTypeSysBKRV$~\cite{BKRV20} and
    \Cref{lem:cbnBKRV_characterizes_meaningfulness}. For
    \Cref{ax:Surface_Genericity} we use
    \Cref{lem:Cbn_Qualitative_Surface_Genericity}. For
    \Cref{ax:Exists_Meaningless}: the term $\Omega$ is $\cbnArr_{S_0}$
    divergent thus by \Cref{lem:cbnBKRV_characterizes_meaningfulness},
    one concludes that it is \CBNSymb-meaningless.
\end{proof}

\begin{corollary}     \begin{enumerate}
    \item The \cbnLambdaTheoryTxt $\cbnThH$ is consistent. 

    \item The theory $\cbnThH*$ is a sensible and
        consistent \cbnLambdaTheoryTxt.
    \end{enumerate}
\end{corollary}


Using arguments similar to \CBVSymb, we can establish that the theory
$\cbnThH*$ is the unique maximal \giulio{consistent} \cbnLambdaTheoryTxt extending
$\cbnThH$ and coincides with the observational equivalence where the observable is being $\cbnSymbSurface_0$-normalizing, analogously to what is well-known for the pure
$\lambda$-calculus \cite{Hyland75,Wadsworth76,barendregt84nh}.

\section{Related Works and Conclusion}
\label{sec:conclusion}

\paragraph{Related Works}
There are at least six different methods to prove (full) \emph{genericity}~in~CbN.
\begin{itemize}[leftmargin=*,labelindent=0pt]
 \item Barendregt~\cite[Ch. 14.3]{barendregt84nh} uses a
tree-topology on 
	terms based  on Böhm trees, and continuity of term application with respect~to~it.
      \item Takahashi's~\cite{Takahashi94} operational proof  exploits the fact that
        the solvable terms are precisely the
       ones having a head normal form.
	\item Kuper's~\cite{Kuper95} operational proof is based on leftmost reduction and
	can also be extended to other calculi (\eg with types, PCF).
	\item Kennaway~\emph{et al.} \cite{KennawayOostromVries99} abstractly provide three axioms in order for a set of $\lambda$-terms to be coherently equated and  satisfy genericity.
	\item Ghilezan~\cite{Ghilezan01} exploits a topology induced by an intersection type system and, as in~\cite{barendregt84nh}, continuity of application with respect~to~it.
	\item Barbarossa and Manzonetto~\cite{BarbarossaM2019} prove that genericity follows from the commutation of normalization and Taylor~expansion.
\end{itemize}


Let us see why the existing methods to prove CbN genericity cannot be
easily adapted to CbV. Barendregt's method does not easily lift to CbV
because CbV \emph{Böhm trees} and their properties are still to be
explored; a first attempt is in~\cite{KerinecManzonettoPagani20,KerinecMR21}. Takahashi's method is
based on \emph{lifting} genericity with respect to substitutions to
genericity with respect to \emph{full contexts}: such a method breaks
in CbV, because it relies on adding abstractions, which in CbV (and
not in CbN) can turn a meaningless term into a meaningful one.
Kuper's approach relies on \emph{leftmost} reduction, which is a
normalizing strategy for CbN; the equivalent normalizing strategy for
CbV~\cite{AccattoliCondoluciCoen21,AccattoliGuerrieriLeberle23} is
involuted and not really leftmost in a syntactic sense, so this method
would not be handy. Kennaway et al.'s \emph{axiomatics} targets
first-order term rewriting systems and CbN $\lambda$-calculus, but it
is unclear how to adapt the technique to \emph{higher-order}
reduction, in particular to reduction at a distance. 
Ghilezan's approach can only work with \emph{idempotent} intersection types, and this makes impossible to extract any quantitative information as we do.
Finally, it is
tedious to adapt Barbarossa and Manzonetto's method to \emph{CbV
Taylor expansion}, whose underlying properties are still not well
explored: some preliminary attempts appear in~\cite{Ehrhard12,CarraroGuerrieri14,KerinecManzonettoPagani20,KerinecMR21}.
Moreover the Taylor expansion exploits some kind of
non-determinism that in fact is not needed to prove genericity.
Indeed, our approach still uses some kind of
approximations of $\lambda$-terms, but contrary to the Taylor
expansion, its reduction rules \mbox{are not non-deterministic}.

Genericity and $\lambda$-theories of CbV was scarcely approached in
the literature on CbV.  The non-adequacy problem mentioned
in~\Cref{sec:cbv} for Plotkin's CbV was already noticed by~Paolini,
Ronchi Della Rocca \emph{et
al.}~\cite{paolini99tia,RoccaP04,AP12,CarraroGuerrieri14,GuerrieriPaoliniRonchi17},
who studied different properties of CbV solvability.
As mentioned before, this notion does not yield a consistent
associated theory. 

Garcia-Perez and Nogueira~\cite{Garcia-PerezN16} proposed a \emph{different}
notion of solvability for Plotkin's CbV that takes into account
  the \emph{order} of a term (number of arguments it can take),  for which they prove
genericity. The $\lambda$-theory generated by their notion of
solvability is different from $\cbvThH$ and it can also be proved to be
consistent. However, the original calculus still suffers from
\emph{non-adequacy}. 

Independently, Accattoli and Lancelot~\cite{AL24} have recently
developed an alternative study of what we call \emph{surface}
genericity in this paper, for both CbV and CbN, via different tools
(\eg bisimulations). As far as we know, the challenging proof of
\emph{full} genericity for \CBVSymb-meaningless terms has never been
proved before.

The consistency of the theory $\cbvThH$ is already proved in the
literature ---or can be easily derived--- from the consistency of some
related
$\lambda_v$-theories~\cite{EgidiHonsellRonchi92,paolini99tia,CarraroGuerrieri14,AccattoliGuerrieri22bis,AL24}.
However, this paper provides the first proof of consistency of the
theory $\cbvThH$ that is \emph{direct} ---not based on consistency of some
related theories--- and \emph{operational} ---not based on semantical
methods.

It would also be reasonable to study $\lambda$-theories that are consistent but non sensible, as it is done for CbN, see \eg \cite{Barendregt92a}.

\paragraph{Conclusion}
We introduce a novel technique to prove \emph{stratified} genericity,
yielding both \emph{surface} and \emph{full} genericity as special
cases. Remarkably, this technique is specified by an \emph{axiomatic}
approach which applies to both \emph{call-by-value} and
\emph{call-by-name}. Moreover, our genericity theorems are not only
\emph{qualitative} (as usual in the literature) but also
\emph{quantitative}, capturing the fact that replacing a meaningless
subterm in a normalizing term $t$ does not affect the number of steps
to reach the normal form of $t$.

Furthermore, we establish the \emph{first direct} and
\emph{operational} proof that the smallest congruence relation
resulting from equating all meaningless $\lambda$-terms is
\emph{consistent}. We substantiate this proof by introducing an
\emph{axiomatic} approach, which applies to both \emph{call-by-value}
and \emph{call-by-name} paradigms. Moreover, we show that these two
theories admit a unique \emph{maximal} consistent extension which
coincides with well-known \emph{observational equivalences}.

Alternatively, our methodology based on meaningful approximants can be
replaced by a non-idempotent intersection type system characterizing
full normalization, as~\cite{BKV17,KesnerV14} for \CBNSymb
and~\cite{AccattoliGuerrieriLeberle23} for \CBVSymb. The obtained
result would however be weaker, since such typing systems only
characterize normalizability of terms, but do not provide any
information about the form of the resulting normal forms. Some ideas
coming from~\cite{BL11} will be studied in this direction.

Through the tools developed in this paper, we aim to extend our
approach to establish genericity within a more encompassing unifying
framework, such as call-by-push-value~\cite{Levy99} or the bang
calculus~\cite{Ehrhard2016,EhrardGuerrieri16,GuerrieriManzonetto18,BKRV20,BKRV23,ArrialGuerrieriKesner23}.
These frameworks, by explicitly distinguishing between values and
computations, have the capability to subsume both call-by-value and
call-by-name at the same time.

An equally challenging model of computation, still unexplored in terms
of meaningfulness and genericity, is
call-by-need~\cite{AriolaF97,AriolaFMOW95} We believe that by
stratifying the notion of call-by-need contexts and reduction, we can
construct a well-defined notion of meaningfulness and then achieve a
corresponding genericity result.

\giulio{We are also interested in understanding the relation between the notion of CbV meaningless studied here and the one of \emph{unsolvability of order $0$} introduced by \cite{Garcia-PerezN16} for Plotkin's CbV.}

\giulio{The notion of meaningfulness in CbN (solvability) is strictly related to the one of Böhm tree, which represents a term as a possible infinite tree that is the sup of its approximations. 
Our notion of meaningful approximation yields a notion of Böhm tree in CbV, which  slightly differs from the one introduced in \cite{KerinecManzonettoPagani20,KerinecMR21} (because of their different notion of approximation). It would be interesting to compare the two notions of CbV Böhm trees and characterize the $\lambda$-theories they generate in a denotational or operational way.}

{We believe that our approach to approximation and stratification are powerful tools that can be used not only to prove genericity but also other important results in $\lambda$-calculus such as Scott's continuity, Berry's stability or Kahn and Plotkin's sequentiality theory, simplifing the proofs provided in \cite{BarbarossaM2019} via the notion of Taylor expansion.

\giulio{
	An open question in CbV was the existence of an appropriate notion of meaningfulness and its associated  $\lambda$-theory fulfilling \emph{all} the basic properties and criteria we discussed in \Cref{sec:intro}, which are well-known for the theory $\genThH$ in CbN. 
	This paper fills the gap by answering positively to this fundamental question, through a new approach that is robust as it works uniformly for CbN~and~CbV.}
%


\bibliographystyle{ACM-Reference-Format}
\nocite{AccattoliGK20} 
\bibliography{biblio}

\newpage
\onecolumn
\appendix
\booltrue{inAppendix}

\section*{Technical Appendix}

\section{Proofs of Section \ref{sec:cbv}}

Let us recall the mutual inductive definitions of $\cbvVrS_\indexK$,
$\cbvNeS_\indexK$ and $\cbvNoS_\indexK$ for $\indexK \in
\setOrdinals$:
\begin{equation*}
    \begin{array}{c}
        \begin{array}{rcl}
            \cbvVrS_\indexK &\coloneqq& x
                \vsep \cbvVrS_\indexK\esub{x}{\cbvNeS_\indexK}
        \\[0.2cm]
            \cbvNeS_\indexK &\coloneqq& \app[\,]{\cbvVrS_\indexK}{\cbvNoS_\indexK}
                \vsep \app[\,]{\cbvNeS_\indexK}{\cbvNoS_\indexK}
                \vsep \cbvNeS_\indexK\esub{x}{\cbvNeS_\indexK}
        \\[0.2cm]
            \cbvNoS_0 &\coloneqq& \abs{x}{t}
                \vsep \cbvNeS_0
                \vsep \cbvVrS_0
                \vsep \cbvNoS_0\esub{x}{\cbvNeS_0}
        \\
            \cbvNoS_{\indexI+1} &\coloneqq& \abs{x}{\cbvNoS_\indexI}
                \vsep \cbvNeS_{\indexI+1}
                \vsep \cbvVrS_{\indexI+1}
                \vsep \cbvNoS_{\indexI+1}\esub{x}{\cbvNeS_{\indexI+1}}\ (\indexI \in \Nat)
        \\
            \cbvNoS_\indexOmega &\coloneqq& \abs{x}{\cbvNoS_\indexOmega}
                \vsep \cbvVrS_\indexOmega
                \vsep \cbvNeS_\indexOmega
                \vsep \cbvNoS_\indexOmega\esub{x}{\cbvNeS_\indexOmega}
        \end{array}
    \end{array}
\end{equation*}

From now on we write $\cbvVarPred{t}$ iff $t= \cbvLCtxt<x>$ for some
list context $\cbvLCtxt$ and some $x$, and $\cbvAbsPred{t}$ iff
$t=\cbvLCtxt<\abs{x}{u}>$ for some list context $\cbvLCtxt$ and some
$u$.

\CbvNoSnCharacterization*
\label{prf:cbv_NoSn_Characterization}%
\stableProof{
     \begin{proof}
Let $\indexK \in \setOrdinals$ and $t \in \setCbvTerms$. We strengthen
the property as follows:
\begin{itemize}
\item[\bltI] $t \in \cbvVrS_\indexK$ if and only if $t$ is a
    $\cbvSymbSurface_\indexK$-normal form and $\cbvVarPred{t}$.

\item[\bltI] $t \in \cbvNeS_\indexK$ if and only if $t$ is a
    $\cbvSymbSurface_\indexK$-normal form and $\neg\cbvAbsPred{t}$ and
    $\neg\cbvVarPred{t}$.

\item[\bltI] $t \in \cbvNoS_\indexK$ if and only if $t$ is a
    $\cbvSymbSurface_\indexK$-normal form.
\end{itemize}
By double implication:
\begin{itemize}
\item[$(\Rightarrow)$] By mutual induction on $t \in \cbvVrS_\indexK$,
    $t \in \cbvNeS_\indexK$ and $t \in \cbvNoS_\indexK$:
    \begin{itemize}
    \item[\bltI] $t \in \cbvVrS_\indexK$: We distinguish two cases:
        \begin{itemize}
        \item[\bltII] $t = x$: Then $t$ is a
            $\cbvSymbSurface_\indexK$-normal form and in particular
            $\cbvVarPred{t}$.

        \item[\bltII] $t = t_1\esub{x}{t_2}$ with $t_1 \in
            \cbvVrS_\indexK$ and $t_2 \in \cbvNeS_\indexK$: By \ih on
            $t_1$ and $t_2$, one has that both are
            $\cbvSymbSurface_\indexK$-normal forms and in particular
            $\cbvVarPred{t_1}$, $\neg\cbvAbsPred{t_2}$ and
            $\neg\cbvVarPred{t_2}$. Thus $t = t_1\esub{x}{t_2}$ is a
            $\cbvSymbSurface_\indexK$-normal form and
            $\cbvVarPred{t}$.
        \end{itemize}

    \item[\bltI] $t \in \cbvNeS_\indexK$: We distinguish three cases: 
        \begin{itemize}
        \item[\bltII] $t = \app{t_1}{t_2}$ with $t_1
            \in\cbvVrS_\indexK$ and $t_2 \in \cbvNoS_\indexK$: Then by
            \ih on $t_1$ and $t_2$, one has that both are
            $\cbvSymbSurface_\indexK$-normal forms and
            $\cbvVarPred{t_1}$. Thus, $t = \app{t_1}{t_2}$ is also a
            $\cbvSymbSurface_\indexK$-normal form and in particular
            $\neg\cbvAbsPred{t}$ and $\neg\cbvVarPred{t}$.

        \item[\bltII] $t = \app{t_1}{t_2}$ with $t_1 \in
            \cbvNeS_\indexK$ and $t_2 \in \cbvNoS_\indexK$: Then by
            \ih on $t_1$ and $t_2$, one has that both are
            $\cbvSymbSurface_\indexK$-normal forms and
            $\neg\cbvAbsPred{t_1}$. Thus, $t = \app{t_1}{t_2}$ is also
            a $\cbvSymbSurface_\indexK$-normal form and in particular
            $\neg\cbvAbsPred{t}$ and $\neg\cbvVarPred{t}$.

        \item[\bltII] $t = t_1\esub{x}{t_2}$ with $t_1, t_2 \in
            \cbvNeS_\indexK$: Then by \ih on $t_1$ and $t_2$, one has
            that both are $\cbvSymbSurface_\indexK$-normal forms and
            $\neg\cbvAbsPred{t_1}$, $\neg\cbvVarPred{t_1}$,
            $\neg\cbvVarPred{t_2}$ and $\neg\cbvVarPred{t_2}$. Thus,
            $t = t_1\esub{x}{t_2}$ is also a
            $\cbvSymbSurface_\indexK$-normal form and in particular
            $\neg\cbvAbsPred{t}$ and $\neg\cbvVarPred{t}$.
        \end{itemize}

    \item[\bltI] $t \in \cbvNoS_\indexK$: We distinguish four cases:
        \begin{itemize}
        \item[\bltII] $t = \abs{x}{t'}$: We distinguish three cases:
            \begin{itemize}
            \item[\bltIII] $\indexK = 0$: Then $t = \abs{x}{t'}$ is a
                $\cbvSymbSurface_0$-normal form.

            \item[\bltIII] $\indexK = \indexI + 1$ and $t' \in
                \cbvNoS_\indexI$ for some $\indexI \in \mathbb{N}$: By
                \ih on $t'$, one has that $t'$ is a
                $\cbvSymbSurface_\indexI$-normal form thus $t = \abs{x}{t'}$
                is a $\cbvSymbSurface_\indexK$-normal form.

            \item[\bltIII] $\indexK = \indexOmega$ and $t' \in
                \cbvNoS_\indexOmega$: By \ih on $t'$, one has that
                $t'$ is a $\cbvSymbSurface_\indexOmega$-normal form
                thus $t = \abs{x}{t'}$ is also a
                $\cbvSymbSurface_\indexOmega$-normal form.
            \end{itemize}

        \item[\bltII] $t \in \cbvVrS_\indexK$: Then by \ih on $t$, it
            is a $\cbvSymbSurface_\indexK$-normal form.

        \item[\bltII] $t \in \cbvNeS_\indexK$: Then by \ih on $t$, it
            is a $\cbvSymbSurface_\indexK$-normal form.

        \item[\bltII] $t = t_1\esub{x}{t_2}$ with $t_1 \in
            \cbvNoS_\indexK$ and $t_2 \in \cbvNeS_\indexK$: Then by
            \ih on $t_1$ and $t_2$, one has that both are
            $\cbvSymbSurface_\indexK$-normal forms and,
            $\neg\cbvAbsPred{t_2}$ and $\neg\cbvVarPred{t_2}$. Thus,
            $t = t_1\esub{x}{t_2}$ is also a
            $\cbvSymbSurface_\indexK$-normal form.
        \end{itemize}
    \end{itemize}

\item[$(\Leftarrow)$] By induction on $t$:
    \begin{itemize}
    \item[\bltI] $t = x$: Then $\cbvVarPred{t}$ and $t \in
    \cbvVrS_\indexK$.

    \item[\bltI] $t = \abs{x}{t'}$: Then $\cbvAbsPred{t}$. We
        distinguish three cases:
        \begin{itemize}
        \item[\bltII] $\indexK = 0$: Then $t = \abs{x}{t'} \in
        \cbvNoS_0$.

        \item[\bltII] $\indexK = \indexI + 1$ for some $\indexI \in
            \mathbb{N}$: Then, $t'$ is necessarily a
            $\cbvSymbSurface_\indexI$-normal form. By \ih on $t'$, one
            has that $t' \in \cbvNoS_\indexI$ thus $t = \abs{x}{t'}
            \in \abs{x}{\cbvNoS_i} \subseteq \cbvNoS_\indexK$.

        \item[\bltII] $\indexK = \indexOmega$: By \ih on $t'$, one has
            that $t' \in \cbvNoS_\indexOmega$ thus $t = \abs{x}{t'}
            \in \abs{x}{\cbvNoS_\indexOmega} \subseteq
            \cbvNoS_\indexOmega$.
        \end{itemize}

    \item[\bltI] $t = \app{t_1}{t_2}$: Then $\neg\cbvAbsPred{t}$ and
        $\neg\cbvVarPred{t}$. Moreover, $t_1$ and $t_2$ are
        necessarily $\cbvSymbSurface_\indexK$-normal forms and
        $\neg\cbvAbsPred{t_1}$. By \ih on $t_2$, one has that $t_2 \in
        \cbvNoS_\indexK$. We distinguish two cases:
        \begin{itemize}
        \item[\bltII] $\cbvVarPred{t_1}$: Then, by \ih on $t_1$, one
            has that $t_1 \in \cbvVrS_\indexK$, thus $t =
            \app{t_1}{t_2} \in \app{\cbvVrS_\indexK}{\cbvNoS_\indexK}
            \subseteq \cbvNeS_\indexK$.

        \item[\bltII] $\neg\cbvVarPred{t_1}$: Then, by \ih on $t_1$,
            one has that $t_1 \in \cbvNeS_\indexK$, thus $t =
            \app{t_1}{t_2} \in \app{\cbvNeS_\indexK}{\cbvNoS_\indexK}
            \subseteq \cbvNeS_\indexK$.
        \end{itemize}

    \item[\bltI] $t = t_1\esub{x}{t_2}$: Then both $t_1$ and $t_2$ are
        necessarily $\cbvSymbSurface_\indexK$-normal forms and,
        $\neg\cbvAbsPred{t_2}$ and $\neg\cbvVarPred{t_2}$. By \ih on
        $t_2$, one has that $t_2 \in \cbvNeS_\indexK$. We distinguish
        three cases:
        \begin{itemize}
        \item[\bltII] $\cbvVarPred{t_1}$: Then, by \ih on $t_1$, one
            has that $t_1 \in \cbvVrS_\indexK$ thus $t =
            t_1\esub{x}{t_2} \in
            \cbvVrS_\indexK\esub{x}{\cbvNeS_\indexK} \subseteq
            \cbvVrS_\indexK$.

        \item[\bltII] $\neg\cbvVarPred{t_1}$ and
            $\neg\cbvAbsPred{t_1}$: Then, by \ih on $t_1$, one has
            that $t_1 \in \cbvNeS_\indexK$ thus $t = t_1\esub{x}{t_2}
            \in \cbvNeS_\indexK\esub{x}{\cbvNeS_\indexK} \subseteq
            \cbvNeS_\indexK$.

        \item[\bltII] otherwise: Then, by \ih on $t_1$, one has that
            $t_1 \in \cbvNoS_\indexK$ thus $t = t_1\esub{x}{t_2} \in
            \cbvNoS_\indexK\esub{x}{\cbvNeS_\indexK} \subseteq
            \cbvNoS_\indexK$.
        \end{itemize}
    \end{itemize}
\end{itemize}
\end{proof}
}

\CbvSurfaceTypedGenericity*
\label{prf:t_meaningless_and_Pi_|>_F<t>_==>_Pi'_|>_F<u>}%
\stableProof{
    \begin{proof}
By induction on $\cbvCCtxt$:
\begin{itemize}
\item[\bltI] $\cbvCCtxt = \Hole$: Then $\cbvCCtxt<t> = t$ thus
    $\cbvCCtxt<t>$ is meaningless which is impossible since it
    contradicts \Cref{lem:cbvBKRV_characterizes_meaningfulness}.

\item[\bltI] $\cbvCCtxt = \abs{x}{\cbvCCtxt'}$: Then $\Pi$ is
    necessarily of the following form, for some finite set $I$:
    \begin{equation*}
        \begin{prooftree}
            \hypo{\Pi_i \cbvTrBKRV \Gamma_i, x : \M_i \vdash \cbvCCtxt'<t> : \tau_i}
            \delims{\left(}{\right)_{i \in I}}
            \inferCbvBKRVAbs{\quad +_{i \in I} \Gamma_i \vdash \abs{x}{\cbvCCtxt'<t>} : \mset{\M_i \typeArrow \tau_i}_{i \in  I} \quad}
        \end{prooftree}
    \end{equation*}
    with $\Gamma = +_{i \in I} \Gamma_i$ and $\sigma = \mset{\M_i
    \typeArrow \tau_i}_{i \in I}$. For each $i \in I$, by \ih on
    $\Pi_i$, one has $\Pi'_i \cbvTrBKRV \Gamma_i, x : \M_i \cbvTrBKRV
    \cbvCCtxt'<u> : \tau_i$ and we set $\Pi'$ as the following
    derivation:
    \begin{equation*}
        \begin{prooftree}
            \hypo{\Pi'_i \cbvTrBKRV \Gamma_i, x : \M_i \vdash \cbvCCtxt'<u> : \tau_i}
            \delims{\left(}{\right)_{i \in I}}
            \inferCbvBKRVAbs{\quad +_{i \in I} \Gamma_i \vdash \abs{x}{\cbvCCtxt'<u>} : \mset{\M_i \typeArrow \tau_i}_{i \in  I} \quad}
        \end{prooftree}
    \end{equation*}

\item[\bltI] $\cbvCCtxt = \app{\cbvCCtxt'}{s}$: Then $\Pi$ is
    necessarily of the following form:
    \begin{equation*}
        \begin{prooftree}
            \hypo{\Pi_1 \cbvTrBKRV \Gamma_1 \vdash \cbvCCtxt'<t> : \mset{\M \typeArrow \sigma}}
            \hypo{\Pi_2 \cbvTrBKRV \Gamma_2 \vdash s : \M}
            \inferCbvBKRVApp{\Gamma_1 + \Gamma_2 \vdash \app{\cbvCCtxt'<t>}{s} : \sigma}
        \end{prooftree}
    \end{equation*}
    with $\Gamma = \Gamma_1 + \Gamma_2$. By \ih on $\Pi_1$, one has
    $\Pi'_1 \cbvTrBKRV \Gamma_1 \vdash \cbvCCtxt'<u> : \mset{\M
    \typeArrow \sigma}$ and we set $\Pi'$ as the following derivation:
    \begin{equation*}
        \begin{prooftree}
            \hypo{\Pi'_1 \cbvTrBKRV \Gamma_1 \vdash \cbvCCtxt'<u> : \mset{\M \typeArrow \sigma}}
            \hypo{\Pi_2 \cbvTrBKRV \Gamma_2 \vdash s : \M}
            \inferCbvBKRVApp{\Gamma_1 + \Gamma_2 \vdash \app{\cbvCCtxt'<u>}{s} : \sigma}
        \end{prooftree}
    \end{equation*}

\item[\bltI] $\cbvCCtxt = \app[\,]{s}{\cbvCCtxt'}$: Then $\Pi$ is
    necessarily of the following form:
    \begin{equation*}
        \begin{prooftree}
            \hypo{\Pi_1 \cbvTrBKRV \Gamma_1 \vdash s : \mset{\M \typeArrow \sigma}}
            \hypo{\Pi_2 \cbvTrBKRV \Gamma_2 \vdash \cbvCCtxt'<t> : \M}
            \inferCbvBKRVApp{\Gamma_1 + \Gamma_2 \vdash \app{s}{\cbvCCtxt'<t>} : \sigma}
        \end{prooftree}
    \end{equation*}
    with $\Gamma = \Gamma_1 + \Gamma_2$. By \ih on $\Pi_2$, one has
    $\Pi'_2 \cbvTrBKRV \Gamma_2 \vdash \cbvCCtxt'<u> : \M$ and we set
    $\Pi'$ as the following derivation:
    \begin{equation*}
        \begin{prooftree}
            \hypo{\Pi_1 \cbvTrBKRV \Gamma_1 \vdash s : \mset{\M \typeArrow \sigma}}
            \hypo{\Pi'_2 \cbvTrBKRV \Gamma_2 \vdash \cbvCCtxt'<u> : \M}
            \inferCbvBKRVApp{\Gamma_1 + \Gamma_2 \vdash \app{s}{\cbvCCtxt'<u>} : \sigma}
        \end{prooftree}
    \end{equation*}

\item[\bltI] $\cbvCCtxt = \cbvCCtxt'\esub{x}{s}$: Then $\Pi$ is
    necessarily of the following form:
    \begin{equation*}
        \begin{prooftree}
            \hypo{\Pi_1 \cbvTrBKRV \Gamma_1, x : \M \vdash \cbvCCtxt'<t> : \sigma}
            \hypo{\Pi_2 \cbvTrBKRV \Gamma_2 \vdash s : \M}
            \inferCbvBKRVEs{\Gamma_1 + \Gamma_2 \vdash \cbvCCtxt'<t>\esub{x}{s} : \sigma}
        \end{prooftree}
    \end{equation*}
    with $\Gamma = \Gamma_1 + \Gamma_2$. By \ih on $\Pi_1$, one has
    $\Pi'_1 \cbvTrBKRV \Gamma_1, x : \M \vdash \cbvCCtxt'<u> : \sigma$
    and we set $\Pi'$ as the following derivation:
    \begin{equation*}
        \begin{prooftree}
            \hypo{\Pi'_1 \cbvTrBKRV \Gamma_1, x : \M \vdash \cbvCCtxt'<u> : \sigma}
            \hypo{\Pi_2 \cbvTrBKRV \Gamma_2 \vdash s : \M}
            \inferCbvBKRVEs{\Gamma_1 + \Gamma_2 \vdash \cbvCCtxt'<u>\esub{x}{s} : \sigma}
        \end{prooftree}
    \end{equation*}

\item[\bltI] $\cbvCCtxt = s\esub{x}{\cbvCCtxt'}$: Then $\Pi$ is
    necessarily of the following form:
    \begin{equation*}
        \begin{prooftree}
            \hypo{\Pi_1 \cbvTrBKRV \Gamma_1, x : \M \vdash s : \sigma}
            \hypo{\Pi_2 \cbvTrBKRV \Gamma_2 \vdash \cbvCCtxt'<t> : \M}
            \inferCbvBKRVEs{\Gamma_1 + \Gamma_2 \vdash s\esub{x}{\cbvCCtxt'<t>} : \sigma}
        \end{prooftree}
    \end{equation*}
    with $\Gamma = \Gamma_1 + \Gamma_2$. By \ih on $\Pi_2$, one has
    $\Pi_1 \cbvTrBKRV \Gamma_1, x : \M \vdash \cbvCCtxt'<u> : \sigma$
    and we set $\Pi'$ as the following derivation:
    \begin{equation*}
        \begin{prooftree}
            \hypo{\Pi_1 \cbvTrBKRV \Gamma_1, x : \M \vdash s : \sigma}
            \hypo{\Pi'_2 \cbvTrBKRV \Gamma_2 \vdash \cbvCCtxt'<u> : \M}
            \inferCbvBKRVEs{\Gamma_1 + \Gamma_2 \vdash s\esub{x}{\cbvCCtxt'<u>} : \sigma}
        \end{prooftree}
    \end{equation*}
\end{itemize}
\end{proof}
}%

\section{Proofs of Section \ref{sec:Cbv_Stratified_Genericity}}

\subsection{Technical Lemmas}

\begin{lemma}
    \label{lem:cbv_t1_<=_t2_and_v1_<=_v2_==>_t1{x:=v1}_<=_t2{x:=v2}}%
    Let $\aprxT_1, \aprxT_2 \in \setCbvAprxTerms$ and $\aprxV_1,
    \aprxV_2 \in \setCbvAprxValues$ such that $\aprxT_1 \cbvAprxLeq
    \aprxT_2$ and $\aprxV_1 \cbvAprxLeq \aprxV_2$, then
    $\aprxT_1\isub{x}{\aprxV_1} \cbvAprxLeq
    \aprxT_2\isub{x}{\aprxV_2}$.
\end{lemma}
\stableProof{
    \begin{proof}
By induction on $\aprxT_1$:
\begin{itemize}
\item[\bltI] $\aprxT_1 = \cbvAprxBot$: Then
	$\aprxT_1\isub{x}{\aprxV_1} = \cbvAprxBot \cbvAprxLeq
	\aprxT_2\isub{x}{\aprxV_2}$.

\item[\bltI] $\aprxT_1 = y$: Then $\aprxT_2 = y$ and we distinguish
    two cases:
    \begin{itemize}
    \item[\bltII] $x = y$: Thus $\aprxT_1\isub{x}{\aprxV_1} = \aprxV_1
        \cbvAprxLeq \aprxV_2 = \aprxT_2\isub{x}{\aprxV_2}$.

    \item[\bltII] $x \neq y$: Thus $\aprxT_1\isub{x}{\aprxV_1} = y =
        \aprxT_2\isub{x}{\aprxV_2}$, thus $\aprxT_1\isub{x}{\aprxV_1}
        \cbvAprxLeq \aprxT_2\isub{x}{\aprxV_2}$ by reflexivity of
        $\cbvAprxLeq$.
    \end{itemize}

\item[\bltI] $\aprxT_1 = \abs{y}{\aprxT'_1}$: Then
	$\aprxT_1\isub{x}{\aprxV_1} =
	\abs{y}{\aprxT'_1\isub{x}{\aprxV_1}}$ (by supposing $y \notin
	\freeVar{\aprxV_1} \cup \freeVar{\aprxV_2}$ without loss of
	generality) and necessarily $\aprxT_2 = \abs{y}{\aprxT'_2}$ with
	$\aprxT'_1 \cbvAprxLeq \aprxT'_2$ thus $\aprxT_2\isub{x}{\aprxV_2}
	= \abs{y}{\aprxT'_2\isub{y}{\aprxV_2}}$. By \ih on $\aprxT'_1$,
	one has that $\aprxT'_1\isub{x}{\aprxV_1} \cbvAprxLeq
	\aprxT'_2\isub{x}{\aprxV_2}$ hence $\aprxT_1\isub{x}{\aprxV_1} =
	\abs{y}{\aprxT'_1\isub{x}{\aprxV_1}} \cbvAprxLeq
	\abs{y}{\aprxT'_2\isub{x}{\aprxV_2}} = \aprxT_2\isub{x}{\aprxV_2}$
	by contextual closure.

\item[\bltI] $\aprxT_1 = \app{\aprxT_1^1}{\aprxT_1^2}$: Then
	$\aprxT_1\isub{x}{\aprxV_1} =
	\app{(\aprxT_1^1\isub{x}{\aprxV_1})}{(\aprxT_1^2\isub{x}{\aprxV_1})}$
	and necessarily $\aprxT_2 = \app[\,]{\aprxT_2^1}{\aprxT_2^2}$ with
	$\aprxT_1^1 \cbvAprxLeq \aprxT_2^1$ and $\aprxT_1^2 \cbvAprxLeq
	\aprxT_2^2$, thus $\aprxT_2\isub{y}{\aprxV_2} =
	\app{(\aprxT_2^1\isub{x}{\aprxV_2})}{(\aprxT_2^2\isub{x}{\aprxV_2})}$.
	By \ih on $\aprxT_1^1$ and $\aprxT_1^2$, one obtains that
	$\aprxT_1^1\isub{x}{\aprxV_1} \cbvAprxLeq
	\aprxT_2^1\isub{x}{\aprxV_2}$ and $\aprxT_1^2\isub{x}{\aprxV_1}
	\cbvAprxLeq \aprxT_2^2\isub{x}{\aprxV_2}$, hence
	$\aprxT_1\isub{x}{\aprxV_1} =
	\app{(\aprxT_1^1\isub{x}{\aprxV_1})}{(\aprxT_1^2\isub{x}{\aprxV_1})}
	\cbvAprxLeq
	\app{(\aprxT_2^1\isub{x}{\aprxV_2})}{(\aprxT_2^2\isub{x}{\aprxV_2})}
	= \aprxT_2\isub{x}{\aprxV_2}$ by contextual closure.

\item[\bltI] $\aprxT_1 = \aprxT_1^1\esub{y}{\aprxT_1^2}$: Then
	$\aprxT_1\isub{x}{\aprxV_1} =
	\aprxT_1^1\isub{x}{\aprxV_1}\esub{y}{\aprxT_1^2\isub{x}{\aprxV_1}}$
	and necessarily $\aprxT_2 = \aprxT_2^1\esub{y}{\aprxT_2^2}$ with
	$\aprxT_1^1 \cbvAprxLeq \aprxT_2^1$ and $\aprxT_1^2 \cbvAprxLeq
	\aprxT_2^2$, thus $\aprxT_2\isub{y}{\aprxV_2} =
	\aprxT_2^1\isub{x}{\aprxV_2}\esub{y}{\aprxT_2^2\isub{x}{\aprxV_2}}$.
	By \ih on $\aprxT_1^1$ and $\aprxT_1^2$, one has
	$\aprxT_1^1\isub{x}{\aprxV_1} \cbvAprxLeq
	\aprxT_2^1\isub{x}{\aprxV_2}$ and $\aprxT_1^2\isub{x}{\aprxV_1}
	\cbvAprxLeq \aprxT_2^2\isub{x}{\aprxV_2}$, so by contextual
	closure $\aprxT_1\isub{x}{\aprxV_1} =
	\aprxT_1^1\isub{x}{\aprxV_1}\esub{y}{\aprxT_1^2\isub{x}{\aprxV_1}}
	\cbvAprxLeq
	\aprxT_2^1\isub{x}{\aprxV_2}\esub{y}{\aprxT_2^2\isub{x}{\aprxV_2}}
	= \aprxT_2\isub{x}{\aprxV_2}$.
	\qedhere
\end{itemize}
\end{proof}
}

\begin{definition}
    The partial order $\cbvAprxLeq$ is extended to \CBVSymb contexts
    by setting $\Hole \cbvAprxLeq \Hole$.
\end{definition}

For example,
$\app{\Hole\esub{x}{\abs{y}{\cbvAprxBot}}}{(\app{z}{\cbvAprxBot})}
\cbvAprxLeq
\app{\Hole\esub{x}{\abs{y}{(\app{z}{z})}}}{(\app{z}{\cbvAprxBot})}$.

\begin{lemma}
    \label{lem:cbv_L1_<=_L2_and_t1_<=_t2_==>_L1<t1>_<=_L2<t2>}%
    Let $\cbvAprxLCtxt_1 \cbvAprxLeq \cbvAprxLCtxt_2$ and $\aprxT_1
    \cbvAprxLeq \aprxT_2$ for some list contexts $\cbvAprxLCtxt_1,
    \cbvAprxLCtxt_2$ and some $\aprxT_1, \aprxT_2 \in
    \setCbvAprxTerms$, then $\cbvAprxLCtxt_1<\aprxT_1> \cbvAprxLeq
    \cbvAprxLCtxt_2<\aprxT_2>$.
\end{lemma}
\stableProof{
    \begin{proof}
By induction on $\cbvAprxLCtxt_1$:
\begin{itemize}
\item[\bltI] $\cbvAprxLCtxt_1 = \Hole$: Then necessarily
    $\cbvAprxLCtxt_2 = \Hole$ thus $\cbvAprxLCtxt_1<\aprxT_1> =
    \aprxT_1 \cbvAprxLeq \aprxT_2 = \cbvAprxLCtxt_2<\aprxT_2>$.

\item[\bltI] $\cbvAprxLCtxt_1 = \cbvAprxLCtxt'_1\esub{x}{\aprxU_1}$:
    Then necessarily $\cbvAprxLCtxt_2 =
    \cbvAprxLCtxt'_2\esub{x}{\aprxU_2}$ with $\cbvAprxLCtxt'_1
    \cbvAprxLeq \cbvAprxLCtxt'_2$ and $\aprxU_1 \cbvAprxLeq \aprxU_2$.
    By \ih on $\cbvAprxLCtxt'_1$, one has that
    $\cbvAprxLCtxt'_1<\aprxT_1> \cbvAprxLeq
    \cbvAprxLCtxt'_2<\aprxT_2>$ thus $\cbvAprxLCtxt_1<\aprxT_1> =
    \cbvAprxLCtxt'_1<\aprxT_1>\esub{x}{\aprxU_1} \cbvAprxLeq
    \cbvAprxLCtxt'_2<\aprxT_2>\esub{x}{\aprxU_2} =
    \cbvAprxLCtxt_2<\aprxT_2>$ by contextual closure.
    \qedhere
\end{itemize}
\end{proof}
}

\begin{lemma}
    \label{lem:cbv_t_meaningful_and_t_->Sn_u_=>_u_meaningful}%
    Let $t, u \in \setCbvTerms$ such that $t \cbvArr_{S_\indexK} u$.
    Then, $t$ is \CBVSymb-meaningful if and only if $u$ is
    \CBVSymb-meaningful.
\end{lemma}
\stableProof{
\begin{proof}
According to qualitative subject reduction and expansion \cite[Prop.
7.4]{AccattoliGuerrieri22bis}, given $t \cbvArr_{S_\indexOmega} u$,
there is a derivation $\Pi \cbvTrBKRV \Gamma \vdash t : \sigma$ if and
only if there is a derivation $\Pi' \cbvTrBKRV \Gamma \vdash u :
\sigma$. As $\cbvArr_{S_\indexK} \subseteq \cbvArr_{S_\indexOmega}$
for all $\indexK \in \setIntegers$, the same equivalence also holds if
we replace the hypothesis $t \cbvArr_{S_\indexOmega} u$ with $t
\cbvArr_{S_\indexK} u$, for all $\indexK \in \setIntegers$. We
conclude thanks to the type-theoretical characterization of
\CBVSymb-meaningfulness
(\Cref{lem:cbvBKRV_characterizes_meaningfulness}).
\end{proof}
}

\begin{lemma}
    \label{lem:Cbv_solvability_on_app_and_es}%
    Let $t, u \in \setCbvTerms$ and $v \in \setCbvValues$. Then
    \begin{enumerate}
    \item If $\app{t}{u}$ is \CBVSymb-meaningful then $t$ and $u$ are
        \CBVSymb-meaningful.%
        \label{lem:cbv_meaningful_on_app}

    \item If $t\esub{x}{u}$ is \CBVSymb-meaningful then $t$ and $u$
        are \CBVSymb-meaningful.%
        \label{lem:cbv_meaningful_on_esub}

    \item If $t\isub{x}{v}$ is \CBVSymb-meaningful then $t$ is
        \CBVSymb-meaningful.%
        \label{lem:cbv_meaningful_on_isub}
    \end{enumerate}
\end{lemma}
\stableProof{
\begin{proof}
\begin{enumerate}
\item Let $\app{t}{u}$ be \CBVSymb-meaningful and suppose by absurd
    that $t$ is \CBVSymb-meaningless. Using
    \Cref{lem:cbvBKRV_characterizes_meaningfulness}, we know that
    $\app{t}{u}$ \tSurface normalizes. Since $t$ is
    \CBVSymb-meaningless, there exists a $(t_i)_{i \in \mathbb{N}}$
    with $t_0 = t$ and $t_i \cbvArr_{S_0} t_{i+1}$ yielding the
    infinite path $\app{t}{u} \cbvArr_{S_0} \app{t_1}{u} \cbvArr_{S_0}
    \app{t_2}{u} \cbvArr_{S_0} \dots$ obtained by always reducing $t$.
    Since surface reduction is diamond, this contradicts the fact that
    $\app{t}{u}$ is surface normalizing. And similarly for
    $\app{u}{t}$.

\item Same principle as the previous item.

\item Let $t\isub{x}{v}$ be \CBVSymb-meaningful. By induction on $t$,
    one can show that if $t \cbvArr_{S_0} u$ then $t\isub{x}{v}
    \cbvArr_{S_0} u\isub{x}{v}$. Using this, one applied the same
    principles as the two previous points.
\end{enumerate}
\end{proof}
}

\begin{definition}
    Meaningful approximation is extended to list contexts $\cbvLCtxt$
    inductively as follows:
    \begin{equation*}
        \cbvTApprox{\Hole}
            \coloneqq \Hole
    \hspace{2cm}
        \cbvTApprox{\cbvLCtxt\esub{x}{s}}
            \coloneqq
                \cbvTApprox{\cbvLCtxt}\esub{x}{\cbvTApprox{s}}
    \end{equation*}
\end{definition}

\begin{corollary} 
    \label{lem:cbv_MA(L<t>)_=_MA(L)<MA(t)>}%
    Let $\cbvLCtxt<t>$ be \CBVSymb-meaningful, then
    $\cbvTApprox{\cbvLCtxt<t>} =
    \cbvTApprox{\cbvLCtxt}\cbvCtxtPlug{\cbvTApprox{t}}$.
\end{corollary}
\stableProof{
    \begin{proof}
By induction on $\cbvLCtxt$:
\begin{itemize}
\item[\bltI] $\cbvLCtxt = \Hole$: Then $\cbvTApprox{\cbvLCtxt<t>} =
    \cbvTApprox{t} = \Hole\cbvCtxtPlug{\cbvTApprox{t}} =
    \cbvTApprox{\Hole}\cbvCtxtPlug{\cbvTApprox{t}}$.

\item[\bltI] $\cbvLCtxt = \cbvLCtxt'\esub{x}{u}$: Then $\cbvLCtxt<t> =
    \cbvLCtxt'<t>\esub{x}{u}$ is \CBVSymb-meaningful so that
    $\cbvTApprox{\cbvLCtxt<t>} =
    \cbvTApprox{\cbvLCtxt'<t>}\esub{x}{\cbvTApprox{u}}$. Moreover,
    using
    \Cref{lem:Cbv_solvability_on_app_and_es}.(\ref{lem:cbv_meaningful_on_esub}),
    one has that $\cbvLCtxt'<t>$ is \CBVSymb-meaningful. By \ih on
    $\cbvLCtxt'$, one obtains that $\cbvTApprox{\cbvLCtxt'<t>} =
    \cbvTApprox{\cbvLCtxt'}\cbvCtxtPlug{\cbvTApprox{t}}$. Then:
    \begin{equation*}
        \begin{array}{r cl cl}
            \cbvTApprox{\cbvLCtxt<t>}
                &=& \cbvTApprox{\cbvLCtxt'<t>}\esub{x}{\cbvTApprox{u}}
        \\
                &\eqih& \cbvTApprox{\cbvLCtxt'}\cbvCtxtPlug{\cbvTApprox{t}}\esub{x}{\cbvTApprox{u}}
        \\
                &=& \cbvTApprox{\cbvLCtxt'}\esub{x}{\cbvTApprox{u}}\cbvCtxtPlug{\cbvTApprox{t}}
        \\
                &=& \cbvTApprox{\cbvLCtxt'\esub{x}{u}}\cbvCtxtPlug{\cbvTApprox{t}}
                &=& \cbvTApprox{\cbvLCtxt}\cbvCtxtPlug{\cbvTApprox{t}}
        \end{array}
    \end{equation*}
    \qedhere
\end{itemize}
\end{proof}
}

\begin{lemma} 
    \label{lem:cbv_MA(t{x:=v})_<=_MA(t){x:=MA(v)}}%
    Let $t \in \setCbvTerms$ and $v \in \setCbvValues$, then
    $\cbvTApprox{t\isub{x}{v}} \cbvAprxLeq
    \cbvTApprox{t}\isub{x}{\cbvTApprox{v}}$.
\end{lemma}
\stableProof{
    \begin{proof}
Let $t \in \setCbvTerms$ and $v \in \setCbvValues$. Cases:
\begin{itemize}
\item[\bltI] $t\isub{x}{v}$ is \CBVSymb-meaningless: Then
    $\cbvTApprox{t\isub{x}{v}} = \cbvAprxBot \cbvAprxLeq
    \cbvTApprox{t}\isub{x}{\cbvTApprox{v}}$.

\item[\bltI] $t\isub{x}{v}$ is \CBVSymb-meaningful:
    Using~\Cref{lem:Cbv_solvability_on_app_and_es}.(\ref{lem:cbv_meaningful_on_isub}),
    one deduces that $t$ is \CBVSymb-meaningful. We proceed by
    induction on $t$. Cases:
    \begin{itemize}
    \item[\bltII] $t = y$: Then $\cbvTApprox{t} = y$ and we
        distinguish two cases:
        \begin{itemize}
        \item[\bltIII] $y = x$: Then $\cbvTApprox{t\isub{x}{v}} =
            \cbvTApprox{v} = x\isub{x}{\cbvTApprox{v}} =
            \cbvTApprox{t}\isub{x}{\cbvTApprox{v}}$.

        \item[\bltIII] $y \neq x$:  Then $\cbvTApprox{t\isub{x}{v}} =
            \cbvTApprox{y} = y = y\isub{x}{\cbvTApprox{v}} =
            \cbvTApprox{t}\isub{x}{\cbvTApprox{v}}$.
        \end{itemize}

    \item[\bltII] $t = \abs{y}{t'}$: Then $\cbvTApprox{t} =
        \abs{y}{\cbvTApprox{t'}}$. By \ih on $t'$, one has that
        $\cbvTApprox{t'\isub{x}{v}} \cbvAprxLeq
        \cbvTApprox{t'}\isub{x}{\cbvTApprox{v}}$. Therefore:
            \begin{equation*}
                \begin{array}{r cl cl}
                    \cbvTApprox{t\isub{x}{v}}
                    &=& \cbvTApprox{\abs{y}{(t'\isub{x}{v})}}
                \\
                    &=& \abs{y}{\cbvTApprox{t'\isub{x}{v}}}
                \\
                    & \cbvAprxLeqIh& \abs{y}{\cbvTApprox{t'}\isub{x}{\cbvTApprox{v}}}
                \\
                    &=& (\abs{y}{\cbvTApprox{t'}})\isub{x}{\cbvTApprox{v}}
                    &=& \cbvTApprox{t}\isub{x}{\cbvTApprox{v}}
                \end{array}
            \end{equation*}

    \item[\bltII] $t = \app{t_1}{t_2}$: Thus $\cbvTApprox{t} =
        \app{\cbvTApprox{t_1}}{\cbvTApprox{t_2}}$. By \ih on $t_1$ and
        $t_2$, one has that $\cbvTApprox{t_1\isub{x}{v}} \cbvAprxLeq
        \cbvTApprox{t_1}\isub{x}{\cbvTApprox{v}}$ and
        $\cbvTApprox{t_2\isub{x}{v}} \cbvAprxLeq
        \cbvTApprox{t_2}\isub{x}{\cbvTApprox{v}}$. Therefore:
        \begin{equation*}
            \begin{array}{r cl cl}
                \cbvTApprox{t\isub{x}{v}}
                &=& \cbvTApprox{\app{t_1\isub{x}{v}}{t_2\isub{x}{v}}}
            \\
                &=& \app{\cbvTApprox{t_1\isub{x}{v}}}{\cbvTApprox{t_2\isub{x}{v}}}
            \\
                & \cbvAprxLeqIh & \app{(\cbvTApprox{t_1}\isub{x}{\cbvTApprox{v}})}{(\cbvTApprox{t_2}\isub{x}{\cbvTApprox{v}})}
            \\
                &=& (\app{\cbvTApprox{t_1}}{\cbvTApprox{t_2}})\isub{x}{\cbvTApprox{v}}
                &=& \cbvTApprox{t}\isub{x}{\cbvTApprox{v}}
            \end{array}
        \end{equation*}

    \item[\bltII] $t = t_1\esub{y}{t_2}$: Thus $\cbvTApprox{t} =
        \cbvTApprox{t_1}\esub{y}{\cbvTApprox{t_2}}$. By \ih on $t_1$
        and $t_2$, one has that $\cbvTApprox{t_1\isub{x}{v}}
        \cbvAprxLeq \cbvTApprox{t_1}\isub{x}{\cbvTApprox{v}}$ and
        $\cbvTApprox{t_2\isub{x}{v}} \cbvAprxLeq
        \cbvTApprox{t_2}\isub{x}{\cbvTApprox{v}}$. Therefore:
        \begin{equation*}
            \begin{array}{r cl cl}
                \cbvTApprox{t\isub{x}{v}}
                &=& \cbvTApprox{t_1\isub{x}{v}\esub{y}{t_2\isub{x}{v}}}
            \\
                &=& \cbvTApprox{t_1\isub{x}{v}}\esub{y}{\cbvTApprox{t_2\isub{x}{v}}}
            \\
                & \cbvAprxLeqIh & \cbvTApprox{t_1}\isub{x}{\cbvTApprox{v}}\esub{y}{\cbvTApprox{t_2}\isub{x}{\cbvTApprox{v}}}
            \\
                &=& (\cbvTApprox{t_1}\esub{y}{\cbvTApprox{t_2}})\isub{x}{\cbvTApprox{v}}
                &=& \cbvTApprox{t}\isub{x}{\cbvTApprox{v}}
            \end{array}
        \end{equation*}
        \qedhere
    \end{itemize}
\end{itemize}
\end{proof}
}

\subsection{Proofs of Subsection
\ref{subsec:Meaningful_Approximation}}

\CbvMeaningfulApproximationLeqTerm* \label{prf:Cbv_MA(t)_<=_t}%
\stableProof{
    \begin{proof}
We distinguish two cases:
\begin{itemize}
\item[\bltI] $t$ is meaningless: Then $\cbvTApprox{t} = \cbvAprxBot
    \cbvAprxLeq t$.

\item[\bltI] $t$ is meaningful: We proceed by induction on $t$.
    Subcases:
    \begin{itemize}
    \item[\bltII] $t = x$: Then $\cbvTApprox{t} = x \cbvAprxLeq x =
        t$.

    \item[\bltII] $t = \abs{x}{t'}$: By \ih on $t'$, one has that
        $\cbvTApprox{t'} \cbvAprxLeq t'$ thus $\cbvTApprox{t} =
        \abs{x}{\cbvTApprox{t'}} \cbvAprxLeq \abs{x}{t'} = t$.

    \item[\bltII] $t = \app{t_1}{t_2}$: By \ih on $t_1$ and $t_2$, one
        has that $\cbvTApprox{t_1} \cbvAprxLeq t_1$ and
        $\cbvTApprox{t_2} \cbvAprxLeq t_2$ so that $\cbvTApprox{t} =
        \app{\cbvTApprox{t_1}}{\cbvTApprox{t_2}} \cbvAprxLeq
        \app{t_1}{t_2} = t$.

    \item[\bltII] $t = t_1\esub{x}{t_2}$: By \ih on $t_1$ and $t_2$,
        $\cbvTApprox{t_1} \cbvAprxLeq t_1$ and $\cbvTApprox{t_2}
        \cbvAprxLeq t_2$,  hence $\cbvTApprox{t} =
        \cbvTApprox{t_1}\esub{x}{\cbvTApprox{t_2}} \cbvAprxLeq
        t_1\esub{x}{t_2} = t$.
        \qedhere
    \end{itemize}
\end{itemize}
\end{proof}
}

\CbvApproximantGenericity*
\label{prf:Cbv_t_meaningless_=>_MA(F<t>)_<=_MA(F<u>)}%
\stableProof{
    \begin{proof}
Let $\cbvCCtxt$ be a full context and $t, u \in \setCbvTerms$ be such
that $t$ is \CBVSymb-meaningless. We distinguish two cases:
\begin{itemize}
\item[\bltI] $\cbvCCtxt<t>$ is \CBVSymb-meaningless: Then
    $\cbvTApprox{\cbvCCtxt<t>} = \cbvAprxBot \cbvAprxLeq
    \cbvTApprox{\cbvCCtxt<u>}$.

\item[\bltI] $\cbvCCtxt<t>$ is meaningful: Then by surface genericity
    (\Cref{lem:Cbv_Qualitative_Surface_Genericity}), one has that
    $\cbvCCtxt<u>$ is also meaningful. We proceed by induction on
    $\cbvCCtxt$:
    \begin{itemize}
    \item[\bltII] $\cbvCCtxt = \Hole$: This case is impossible because
        it would mean $\cbvCCtxt<t> = t$ is \CBVSymb-meaningless (by
        hypothesis) and meaningful at the same time.

    \item[\bltII] $\cbvCCtxt = \abs{x}{\cbvCCtxt'}$: By \ih on
        $\cbvCCtxt'$, one has $\cbvTApprox{\cbvCCtxt'<t>} \cbvAprxLeq
        \cbvTApprox{\cbvCCtxt'<u>}$, thus $\cbvTApprox{\cbvCCtxt<t>} =
        \abs{x}{\cbvTApprox{\cbvCCtxt'<t>}} \cbvAprxLeqIh
        \abs{x}{\cbvTApprox{\cbvCCtxt'<u>}} =
        \cbvTApprox{\cbvCCtxt<u'>}$ by contextual closure.

    \item[\bltII] $\cbvCCtxt = \app{\cbvCCtxt'}{s}$: By \ih on
        $\cbvCCtxt'$, one has $\cbvTApprox{\cbvCCtxt'<t>} \cbvAprxLeq
        \cbvTApprox{\cbvCCtxt'<u>}$, thus $\cbvTApprox{\cbvCCtxt<t>} =
        \app{\cbvTApprox{\cbvCCtxt'<t>}}{\cbvTApprox{s}} \cbvAprxLeqIh
        \app{\cbvTApprox{\cbvCCtxt'<u>}}{\cbvTApprox{s}} \allowbreak=
        \cbvTApprox{\cbvCCtxt<u'>}$ by contextual closure.

    \item[\bltII] $\cbvCCtxt = \app[\,]{s}{\cbvCCtxt'}$: By \ih on
        $\cbvCCtxt'$, one has $\cbvTApprox{\cbvCCtxt'<t>} \cbvAprxLeq
        \cbvTApprox{\cbvCCtxt'<u>}$, thus $\cbvTApprox{\cbvCCtxt<t>} =
        \app{\cbvTApprox{s}}{\cbvTApprox{\cbvCCtxt'<t>}} \cbvAprxLeqIh
        \app{\cbvTApprox{s}}{\cbvTApprox{\cbvCCtxt'<u>}} \allowbreak=
        \cbvTApprox{\cbvCCtxt<u'>}$ by contextual closure.

    \item[\bltII] $\cbvCCtxt = \cbvCCtxt'\esub{x}{s}$: By \ih on
        $\cbvCCtxt'$, one has $\cbvTApprox{\cbvCCtxt'<t>} \cbvAprxLeq
        \cbvTApprox{\cbvCCtxt'<u>}$, and hence
        $\cbvTApprox{\cbvCCtxt<t>} =
        \cbvTApprox{\cbvCCtxt'<t>}\esub{x}{\cbvTApprox{s}}
        \cbvAprxLeqIh
        \cbvTApprox{\cbvCCtxt'<u>}\esub{x}{\cbvTApprox{s}} =
        \cbvTApprox{\cbvCCtxt<u'>}$ by contextual closure.

    \item[\bltII] $\cbvCCtxt = s\esub{x}{\cbvCCtxt'}$: By \ih on
        $\cbvCCtxt'$, one has $\cbvTApprox{\cbvCCtxt'<t>} \cbvAprxLeq
        \cbvTApprox{\cbvCCtxt'<u>}$, and hence
        $\cbvTApprox{\cbvCCtxt<t>} =
        \cbvTApprox{s}\esub{x}{\cbvTApprox{\cbvCCtxt'<t>}}
        \cbvAprxLeqIh
        \cbvTApprox{s}\esub{x}{\cbvTApprox{\cbvCCtxt'<u>}} =
        \cbvTApprox{\cbvCCtxt<u'>}$ by contextual closure.
        \qedhere
    \end{itemize}
\end{itemize}
\end{proof}

}

\CbvDynamicApproximation*
\label{prf:Cbv_t_->Sn_u_=>_MA(t)_->*Sn_MA(u)}%
\stableProof{
    \begin{proof}
We distinguish two cases:
\begin{itemize}
\item[\bltI] $t$ is \CBVSymb-meaningless: Then $u$ is necessarily
    \CBVSymb-meaningless by
    \Cref{lem:cbv_t_meaningful_and_t_->Sn_u_=>_u_meaningful}, thus
    $\cbvTApprox{t} = \cbvAprxBot$ and $\cbvTApprox{u} = \cbvAprxBot$.
    We set $\aprxU \coloneqq \cbvTApprox{u}$, therefore by reflexivity
    of both relations $\cbvTApprox{t} \cbvAprxArr^=_{S_\indexK} \aprxU
    \cbvAprxGeq \cbvTApprox{u}$.

\item[\bltI] $t$ is \CBVSymb-meaningful: Then $u$ is necessarily
    \CBVSymb-meaningful by
    \Cref{lem:cbv_t_meaningful_and_t_->Sn_u_=>_u_meaningful}. Let $t,
    u \in \setCbvTerms$ such that $t \cbvArr_{S_\indexK} u$. Then
    there exists a context $\cbvSCtxt_\indexK$ and $t', u' \in
    \setCbvTerms$ such that $t = \cbvSCtxt_\indexK<t'>$, $u =
    \cbvSCtxt_\indexK<u'>$ and $t' \mapstoR u'$ for some $\rel \in
    \{\cbvSymbBeta, \cbvSymbSubs\}$. We reason by induction on
    $\cbvSCtxt_\indexK$:
    \begin{itemize}
    \item[\bltII] $\cbvSCtxt_\indexK = \Hole$: Then $t = t'$ and $u =
        u'$. We distinguish two cases:
        \begin{itemize}
        \item[\bltIII] $\rel = \cbvSymbBeta$: Then $t' =
            \app{\cbvLCtxt<\abs{x}{s_1}>}{s_2}$ and $u' =
            \cbvLCtxt<s_1\esub{x}{s_2}>$ for some list context
            $\cbvLCtxt$ and some $s_1, s_2 \in \setCbvTerms$. Since
            $t'$ is \CBVSymb-meaningful, then using
            \Cref{lem:Cbv_solvability_on_app_and_es}.(\ref{lem:cbv_meaningful_on_app}),
            one deduces in particular that $\cbvLCtxt<\abs{x}{s_1}>$
            is \CBVSymb-meaningful. Since $u'$ is \CBVSymb-meaningful,
            then using
            \Cref{lem:Cbv_solvability_on_app_and_es}.(\ref{lem:cbv_meaningful_on_esub}),
            one deduces in particular that $s_1\esub{x}{s_2}$ is
            \CBVSymb-meaningful. Thus, using
            \Cref{lem:cbv_MA(L<t>)_=_MA(L)<MA(t)>}, one has that
            $\cbvTApprox{t'} =
            \app{\cbvTApprox{\cbvLCtxt<\abs{x}{s_1}>}}{\cbvTApprox{s_2}}
            =
            \app{\cbvTApprox{\cbvLCtxt}\cbvCtxtPlug{\abs{x}{\cbvTApprox{s_1}}}}{\cbvTApprox{s_2}}$
            and $\cbvTApprox{u'} =
            \cbvTApprox{\cbvLCtxt<s_1\esub{x}{s_2}>} =
            \cbvTApprox{\cbvLCtxt}\cbvCtxtPlug{
            \cbvTApprox{s_1\esub{x}{s_2}}} =
            \cbvTApprox{\cbvLCtxt}\cbvCtxtPlug{\cbvTApprox{s_1}\esub{x}{\cbvTApprox{s_2}}}$.
            Thus $\cbvTApprox{t} = \cbvTApprox{t'}
            \mapstoR[\cbvSymbBeta] \cbvTApprox{u'} = \cbvTApprox{u}$.
            We set $\aprxU := \cbvTApprox{u}$, therefore by
            reflexivity $\cbvTApprox{t} \cbvAprxArr_{S_\indexK} \aprxU
            \cbvAprxGeq \cbvTApprox{u}$.

        \item[\bltIII] $\rel = \cbvSymbSubs$: Then $t' =
            s\esub{x}{\cbvLCtxt<v>}$ and $u' =
            \cbvLCtxt<s\isub{x}{v}>$ for some list context
            $\cbvLCtxt$, and some $s \in \setCbvTerms$ and $v \in
            \setCbvValues$. Since $t'$ is \CBVSymb-meaningful, then
            $\cbvLCtxt<v>$ is \CBVSymb-meaningful by
            \Cref{lem:Cbv_solvability_on_app_and_es}.(\ref{lem:cbv_meaningful_on_esub}).
            Since $u'$ is \CBVSymb-meaningful, then using
            \Cref{lem:Cbv_solvability_on_app_and_es}.(\ref{lem:cbv_meaningful_on_esub}),
            one deduces that $s\isub{x}{v}$ is \CBVSymb-meaningful.
            Using
            \Cref{lem:cbv_MA(L<t>)_=_MA(L)<MA(t)>,lem:cbv_MA(t{x:=v})_<=_MA(t){x:=MA(v)}},
            one has that $\cbvTApprox{t'} =
            \cbvTApprox{s}\esub{x}{\cbvTApprox{\cbvLCtxt}\cbvCtxtPlug{\cbvTApprox{v}}}$
            and $\cbvTApprox{u'} \cbvAprxLeq
            \cbvTApprox{\cbvLCtxt}\cbvCtxtPlug{\cbvTApprox{s}\isub{x}{\cbvTApprox{v}}}$.
            Thus $\cbvTApprox{t} = \cbvTApprox{t'}
            \mapstoR[\cbvSymbSubs]
            \cbvTApprox{\cbvLCtxt}\cbvCtxtPlug{\cbvTApprox{s}\isub{x}{\cbvTApprox{v}}}
            \cbvAprxGeq \cbvTApprox{u'} = \cbvTApprox{u}$. We set
            $\aprxU \coloneqq
            \cbvTApprox{\cbvLCtxt}\cbvCtxtPlug{\cbvTApprox{s}\isub{x}{\cbvTApprox{v}}}$,
            therefore $\cbvTApprox{t} \cbvAprxArr_{S_\indexK} \aprxU
            \cbvAprxGeq \cbvTApprox{u}$.
        \end{itemize}

    \item[\bltII] $\cbvSCtxt_\indexK = \abs{x}{\cbvSCtxt_{\indexK-1}}$
        with $k > 0$ or $\indexK = \indexOmega$: By \ih on
        $\cbvSCtxt_{\indexK-1}$, there is $\aprxU'$ such that
        $\cbvTApprox{\cbvSCtxt_{\indexK-1}<t'>}
        \cbvAprxArr^=_{S_{\indexK-1}} \aprxU' \cbvAprxGeq
        \cbvTApprox{\cbvSCtxt_{\indexK-1}<u'>}$. We set $\aprxU
        \coloneqq \abs{x}{\aprxU'}$, therefore
        $\cbvTApprox{\cbvSCtxt_\indexK<t'>} =
        \abs{x}{\cbvTApprox{\cbvSCtxt_{\indexK-1}<t'>}}
        \cbvAprxArr^=_{S_\indexK} \aprxU \cbvAprxGeq
        \abs{x}{\cbvTApprox{\cbvSCtxt_{\indexK-1}<u'>}} =
        \cbvTApprox{\cbvSCtxt_\indexK<u'>}$.

    \item[\bltII] $\cbvSCtxt_\indexK = \app{\cbvSCtxt'_\indexK}{s}$:
        Thus $\cbvTApprox{t} =
        \app{\cbvTApprox{\cbvSCtxt_\indexK<t'>}}{\cbvTApprox{s}}$ and
        $\cbvTApprox{u} =
        \app{\cbvTApprox{\cbvSCtxt_\indexK<u'>}}{\cbvTApprox{s}}$. By
        \ih on $\cbvSCtxt'_\indexK$, there is $\aprxU' \in
        \setCbvAprxTerms$ such that
        $\cbvTApprox{\cbvSCtxt'_\indexK<t'>} \cbvAprxArr^=_{S_\indexK}
        \aprxU' \cbvAprxGeq \cbvTApprox{\cbvSCtxt'_\indexK<u'>}$. We
        set $\aprxU \coloneqq \app{\aprxU'}{\cbvTApprox{s}}$,
        therefore $\cbvTApprox{t} \cbvAprxArr^=_{S_\indexK} \aprxU
        \cbvAprxGeq \cbvTApprox{u}$.

    \item[\bltII] $\cbvSCtxt_\indexK =
        \app[\,]{s}{\cbvSCtxt'_\indexK}$: Thus $\cbvTApprox{t} =
        \app{\cbvTApprox{s}}{\cbvTApprox{\cbvSCtxt_\indexK<t'>}}$ and
        $\cbvTApprox{u} =
        \app{\cbvTApprox{s}}{\cbvTApprox{\cbvSCtxt_\indexK<u'>}}$. By
        \ih on $\cbvSCtxt'_\indexK$, there is $\aprxU' \in
        \setCbvAprxTerms$ such that
        $\cbvTApprox{\cbvSCtxt'_\indexK<t'>} \cbvAprxArr^=_{S_\indexK}
        \aprxU' \cbvAprxGeq \cbvTApprox{\cbvSCtxt'_\indexK<u'>}$. We
        set $\aprxU \coloneqq \app{\cbvTApprox{s}}{\aprxU'}$,
        therefore $\cbvTApprox{t} \cbvAprxArr^=_{S_\indexK} \aprxU
        \cbvAprxGeq \cbvTApprox{u}$.

    \item[\bltII] $\cbvSCtxt_\indexK = \cbvSCtxt'_\indexK\esub{x}{s}$:
        Thus $\cbvTApprox{t} =
        \cbvTApprox{\cbvSCtxt_\indexK<t'>}\esub{x}{\cbvTApprox{s}}$
        and $\cbvTApprox{u} =
        \cbvTApprox{\cbvSCtxt_\indexK<u'>}\esub{x}{\cbvTApprox{s}}$.
        By \ih on $\cbvSCtxt'_\indexK$, there is $\aprxU' \in
        \setCbvAprxTerms$ such that
        $\cbvTApprox{\cbvSCtxt'_\indexK<t'>} \cbvAprxArr^=_{S_\indexK}
        \aprxU' \cbvAprxGeq \cbvTApprox{\cbvSCtxt'_\indexK<u'>}$. We
        set $\aprxU \coloneqq \aprxU'\esub{x}{\cbvTApprox{s}}$, hence
        $\cbvTApprox{t} \cbvAprxArr^=_{S_\indexK} \aprxU \cbvAprxGeq
        \cbvTApprox{u}$.

    \item[\bltII] $\cbvSCtxt_\indexK = s\esub{x}{\cbvSCtxt'_\indexK}$:
        Thus $\cbvTApprox{t} =
        \cbvTApprox{s}\esub{x}{\cbvTApprox{\cbvSCtxt_\indexK<t'>}}$
        and $\cbvTApprox{u} =
        \cbvTApprox{s}\esub{x}{\cbvTApprox{\cbvSCtxt_\indexK<u'>}}$.
        By \ih on $\cbvSCtxt'_\indexK$, there is $\aprxU' \in
        \setCbvAprxTerms$ such that
        $\cbvTApprox{\cbvSCtxt'_\indexK<t'>} \cbvAprxArr^=_{S_\indexK}
        \aprxU' \cbvAprxGeq \cbvTApprox{\cbvSCtxt'_\indexK<u'>}$. We
        set $\aprxU \coloneqq \cbvTApprox{s}\esub{x}{\aprxU'}$, hence
        $\cbvTApprox{t} \cbvAprxArr^=_{S_\indexK} \aprxU \cbvAprxGeq
        \cbvTApprox{u}$.
        \qedhere
    \end{itemize}
\end{itemize}
\end{proof}

}

\CbvDynamicLift* \label{prf:cbvAGK_Aprx_Simulation_MultiStep}%
\stableProof{
    \begin{proof}
Let $\aprxT, \aprxU, \aprxT' \in \setCbvAprxTerms$ such that $\aprxT
\cbvAprxLeq \aprxT'$ and $\aprxT \cbvAprxArr_{S_\indexK} \aprxU$. Then
$\aprxT = \cbvAprxSCtxt_\indexK<\aprxT_r>$ and $\aprxU =
\cbvAprxSCtxt_\indexK<\aprxU_r>$ with $\aprxT_r \mapstoR \aprxU_r$ for
some $\aprxT_r, \aprxU_r \in \setCbvAprxTerms$ and $\rel \in
\{\cbvSymbBeta, \cbvSymbSubs\}$. Let us proceed by induction on
$\cbvAprxSCtxt_\indexK$:
\begin{itemize}
\item[\bltI] $\cbvAprxSCtxt_\indexK = \Hole$: Then $\aprxT = \aprxT_r$
    and $\aprxU = \aprxU_r$. We distinguish two cases:
    \begin{itemize}
    \item[\bltII] $\rel = \cbvSymbBeta$: Then $\aprxT_r =
        \app{\cbvAprxLCtxt<\abs{x}{\aprxS_1}>}{\aprxS_2}$ and
        $\aprxU_r = \cbvAprxLCtxt<\aprxS_1\esub{x}{\aprxS_2}>$ for
        some list context $\cbvAprxLCtxt$ and $\aprxS_1, \aprxS_2 \in
        \setCbvAprxTerms$. Since $\aprxT \cbvAprxLeq \aprxT'$, by
        induction on $\cbvAprxLCtxt$, one can prove that $\aprxT' =
        \app{\cbvAprxLCtxt'<\abs{x}{\aprxS'_1}>}{\aprxS'_2}$ with
        $\cbvAprxLCtxt \cbvAprxLeq \cbvAprxLCtxt'$, $\aprxS_1
        \cbvAprxLeq \aprxS'_1$ and $\aprxS_2 \cbvAprxLeq \aprxS'_2$
        for some list context $\cbvAprxLCtxt'$ and $\aprxS'_1,
        \aprxS'_2 \in \setCbvAprxTerms$. Taking $\aprxU' =
        \cbvAprxLCtxt'<\aprxS'_1\esub{x}{\aprxS'_2}>$ concludes this
        case since $\aprxT' =
        \app{\cbvAprxLCtxt'<\abs{x}{\aprxS'_1}>}{\aprxS'_2}
        \cbvAprxArr_{S_\indexK}
        \cbvAprxLCtxt'<\aprxS'_1\esub{x}{\aprxS'_2}> = \aprxU'$ and by
        transitivity $\aprxS_1\esub{x}{\aprxS_2} \cbvAprxLeq
        \aprxS'_1\esub{x}{\aprxS'_2}$ thus
        $\cbvAprxLCtxt<\aprxS_1\esub{x}{\aprxS_2}> \cbvAprxLeq
        \cbvAprxLCtxt'<\aprxS'_1\esub{x}{\aprxS'_2}> = \aprxU'$ using
        \Cref{lem:cbv_L1_<=_L2_and_t1_<=_t2_==>_L1<t1>_<=_L2<t2>}.

    \item[\bltII] $\rel = \cbvSymbSubs$: Then $\aprxT_r =
        \aprxS\esub{x}{\cbvAprxLCtxt<\aprxV>}$ and $\aprxU_r =
        \cbvAprxLCtxt<\aprxS\isub{x}{\aprxV}>$ for some list context
        $\cbvAprxLCtxt$, $\aprxS \in \setCbvAprxTerms$ and $\aprxV \in
        \setCbvAprxValues$. Since $\aprxT \cbvAprxLeq \aprxT'$, by
        induction on $\cbvAprxLCtxt$, one can prove that $\aprxT' =
        \aprxS'\esub{x}{\cbvAprxLCtxt'<\aprxV'>}$ with $\cbvAprxLCtxt
        \cbvAprxLeq \cbvAprxLCtxt'$, $\aprxS \cbvAprxLeq \aprxS'$ and
        $\aprxV \cbvAprxLeq \aprxV'$  for some list context
        $\cbvAprxLCtxt'$,  $\aprxS' \in \setCbvAprxTerms$ and $\aprxV'
        \in \setCbvAprxValues$. Taking $\aprxU' =
        \cbvAprxLCtxt'<\aprxS'\isub{x}{\aprxV'}>$ concludes this case
        since $\aprxT' = \aprxS'\esub{x}{\cbvAprxLCtxt'<\aprxV'>}
        \cbvAprxArr_{S_\indexK}
        \cbvAprxLCtxt'<\aprxS'\isub{x}{\aprxV'}>$ and by
        \Cref{lem:cbv_t1_<=_t2_and_v1_<=_v2_==>_t1{x:=v1}_<=_t2{x:=v2}}
        $\aprxS\isub{x}{\aprxV} \cbvAprxLeq \aprxS'\isub{x}{\aprxV'}$,
        thus $\cbvAprxLCtxt<\aprxS\isub{x}{\aprxV}> \cbvAprxLeq
        \cbvAprxLCtxt'<\aprxS'\isub{x}{\aprxV'}> = \aprxU'$ using
        \Cref{lem:cbv_L1_<=_L2_and_t1_<=_t2_==>_L1<t1>_<=_L2<t2>}.
    \end{itemize}

\item[\bltI] $\cbvAprxSCtxt_\indexK = \abs{x}{\cbvAprxSCtxt'_m}$ with
    either $k > 0$ and $m = k - 1$, or $k = m = \omega$: Then $\aprxT
    = \abs{x}{\cbvAprxSCtxt'_m<\aprxT_r>}$ and $\aprxU =
    \abs{x}{\cbvAprxSCtxt'_m<\aprxU_r>}$ thus $\aprxT' =
    \abs{x}{\aprxT'_1}$ for some $\aprxT'_1  \in \setCbvAprxTerms$
    such that $\cbvAprxSCtxt'_m<\aprxT_r> \cbvAprxLeq \aprxT'_1$.
    Since $\cbvAprxSCtxt'_m<\aprxT_r> \cbvAprxArr_{S_m}
    \cbvAprxSCtxt'_m<\aprxU_r>$, by \ih on $\cbvAprxSCtxt'_m$, one
    obtains $\aprxU'_1 \in \setCbvAprxTerms$ such that $\aprxT'_1
    \cbvAprxArr_{S_m} \aprxU'_1$ and $\cbvAprxSCtxt'_m<\aprxU_r>
    \cbvAprxLeq \aprxU'_1$. Taking $\aprxU' = \abs{x}{\aprxU'_1}$
    concludes this case since $\aprxT' = \abs{x}{\aprxT'_1}
    \cbvAprxArr_{S_\indexK} \abs{x}{\aprxU'_1} = \aprxU'$ and $\aprxU
    = \abs{x}{\cbvAprxSCtxt'_\indexK<\aprxU_r>} \cbvAprxLeq
    \abs{x}{\aprxU'_1} = \aprxU'$.

\item[\bltI] $\cbvAprxSCtxt_\indexK =
    \app[\,]{\cbvAprxSCtxt'_\indexK}{\aprxS}$: Then $\aprxT =
    \app{\cbvAprxSCtxt'_\indexK<\aprxT_r>}{\aprxS}$ and $\aprxU =
    \app{\cbvAprxSCtxt'_\indexK<\aprxU_r>}{\aprxS}$ thus $\aprxT' =
    \app[\,]{\aprxT'_1}{\aprxS'}$ for some $\aprxT'_1, \aprxS \in
    \setCbvAprxTerms$ such that $\cbvAprxSCtxt'_\indexK<\aprxT_r>
    \cbvAprxLeq \aprxT'_1$ and $\aprxS \cbvAprxLeq \aprxS'$. Since
    $\cbvAprxSCtxt'_\indexK<\aprxT_r> \cbvAprxArr_{S_\indexK}
    \cbvAprxSCtxt'_\indexK<\aprxU_r>$, by \ih on
    $\cbvAprxSCtxt'_\indexK$, one obtains $\aprxU'_1 \in
    \setCbvAprxTerms$ such that $\aprxT'_1 \cbvAprxArr_{S_\indexK}
    \aprxU'_1$ and $\cbvAprxSCtxt'_\indexK<\aprxU_r> \cbvAprxLeq
    \aprxU'_1$. Taking $\aprxU' = \app{\aprxU'_1}{\aprxS'}$ concludes
    this case since $\aprxT' = \app{\aprxT'_1}{\aprxS'}
    \cbvAprxArr_{S_\indexK} \app{\aprxU'_1}{\aprxS'} = \aprxU'$ and
    $\aprxU = \app{\cbvAprxSCtxt'_\indexK<\aprxU_r>}{\aprxS}
    \cbvAprxLeq \app{\aprxU'_1}{\aprxS'} = \aprxU'$ by transitivity.

\item[\bltI] $\cbvAprxSCtxt_\indexK =
    \app[\,]{\aprxS}{\cbvAprxSCtxt'_\indexK}$: Then $\aprxT =
    \app[\,]{\aprxS}{\cbvAprxSCtxt'_\indexK<\aprxT_r>}$ and $\aprxU =
    \app[\,]{\aprxS}{\cbvAprxSCtxt'_\indexK<\aprxU_r>}$ thus $\aprxT'
    = \app[\,]{\aprxS'}{\aprxT'_1}$ for some $\aprxS, \aprxT'_1 \in
    \setCbvAprxTerms$ such that $\aprxS \cbvAprxLeq \aprxS'$ and
    $\cbvAprxSCtxt'_\indexK<\aprxT_r> \cbvAprxLeq \aprxT'_1$. Since
    $\cbvAprxSCtxt'_\indexK<\aprxT_r> \cbvAprxArr_{S_\indexK}
    \cbvAprxSCtxt'_\indexK<\aprxU_r>$, by \ih on
    $\cbvAprxSCtxt'_\indexK$, one obtains $\aprxU'_1 \in
    \setCbvAprxTerms$ such that $\aprxT'_1 \cbvAprxArr_{S_\indexK}
    \aprxU'_1$ and $\cbvAprxSCtxt'_\indexK<\aprxU_r> \cbvAprxLeq
    \aprxU'_1$. Taking $\aprxU' = \app{\aprxS'}{\aprxU'_1}$ concludes
    this case since $\aprxT' = \app{\aprxS'}{\aprxT'_1}
    \cbvAprxArr_{S_\indexK} \app{\aprxS'}{\aprxU'_1} = \aprxU'$ and
    $\aprxU = \app[\,]{\aprxS}{\cbvAprxSCtxt'_\indexK<\aprxU_r>}
    \cbvAprxLeq \app{\aprxS'}{\aprxU'_1} = \aprxU'$ by transitivity.

\item[\bltI] $\cbvAprxSCtxt_\indexK =
    \cbvAprxSCtxt'_\indexK\esub{x}{\aprxS}$: Then $\aprxT =
    \cbvAprxSCtxt'_\indexK<\aprxT_r>\esub{x}{\aprxS}$ and $\aprxU =
    \cbvAprxSCtxt'_\indexK<\aprxU_r>\esub{x}{\aprxS}$ thus $\aprxT' =
    \aprxT'_1\esub{x}{\aprxS'}$ for some $\aprxT'_1, \aprxS \in
    \setCbvAprxTerms$ such that $\cbvAprxSCtxt'_\indexK<\aprxT_r>
    \cbvAprxLeq \aprxT'_1$ and $\aprxS \cbvAprxLeq \aprxS'$. Since
    $\cbvAprxSCtxt'_\indexK<\aprxT_r> \cbvAprxArr_{S_\indexK}
    \cbvAprxSCtxt'_\indexK<\aprxU_r>$, by \ih on
    $\cbvAprxSCtxt'_\indexK$, one obtains $\aprxU'_1 \in
    \setCbvAprxTerms$ such that $\aprxT'_1 \cbvAprxArr_{S_\indexK}
    \aprxU'_1$ and $\cbvAprxSCtxt'_\indexK<\aprxU_r> \cbvAprxLeq
    \aprxU'_1$. Taking $\aprxU' = \aprxU'_1\esub{x}{\aprxS'}$
    concludes this case since $\aprxT' = \aprxT'_1\esub{x}{\aprxS'}
    \cbvAprxArr_{S_\indexK} \aprxU'_1\esub{x}{\aprxS'} = \aprxU'$ and
    $\aprxU = \cbvAprxSCtxt'_\indexK<\aprxU_r>\esub{x}{\aprxS}
    \cbvAprxLeq \aprxU'_1\esub{x}{\aprxS'} = \aprxU'$ by transitivity.

\item[\bltI] $\cbvAprxSCtxt_\indexK =
    \aprxS\esub{x}{\cbvAprxSCtxt'_\indexK}$: Then $\aprxT =
    \aprxS\esub{x}{\cbvAprxSCtxt'_\indexK<\aprxT_r>}$ and $\aprxU =
    \aprxS\esub{x}{\cbvAprxSCtxt'_\indexK<\aprxU_r>}$ thus $\aprxT' =
    \aprxS'\esub{x}{\aprxT'_1}$ for some $\aprxS, \aprxT'_1 \in
    \setCbvAprxTerms$ such that $\aprxS \cbvAprxLeq \aprxS'$ and
    $\cbvAprxSCtxt'_\indexK<\aprxT_r> \cbvAprxLeq \aprxT'_1$. Since
    $\cbvAprxSCtxt'_\indexK<\aprxT_r> \cbvAprxArr_{S_\indexK}
    \cbvAprxSCtxt'_\indexK<\aprxU_r>$, by \ih on
    $\cbvAprxSCtxt'_\indexK$, one obtains $\aprxU'_1 \in
    \setCbvAprxTerms$ such that $\aprxT'_1 \cbvAprxArr_{S_\indexK}
    \aprxU'_1$ and $\cbvAprxSCtxt'_\indexK<\aprxU_r> \cbvAprxLeq
    \aprxU'_1$. Taking $\aprxU' = \aprxS'\esub{x}{\aprxU'_1}$
    concludes this case since $\aprxT' = \aprxS'\esub{x}{\aprxT'_1}
    \cbvAprxArr_{S_\indexK} \aprxS'\esub{x}{\aprxU'_1} = \aprxU'$ and
    $\aprxU = \aprxS\esub{x}{\cbvAprxSCtxt'_\indexK<\aprxU_r>}
    \cbvAprxLeq \aprxS'\esub{x}{\aprxU'_1} = \aprxU'$ by transitivity.
    \qedhere
\end{itemize}
\end{proof}

}

Let us recall mutual inductive definitions of $\cbvAprxBvrS_\indexK$,
$\cbvAprxBneS_\indexK$ and $\cbvAprxBnoS_\indexK$ for $\indexK \in
\setOrdinals$:
\begin{equation*}
    \begin{array}{rcl}
        \cbvAprxBvrS_\indexK &\coloneqq& x
            \vsep \cbvAprxBvrS_\indexK\esub{x}{\cbvAprxBneS_\indexK}
    \\[0.2cm]
        \cbvAprxBneS_\indexK &\coloneqq& \app[\,]{\cbvAprxBvrS_\indexK}{\cbvAprxBnoS_\indexK}
            \vsep \app[\,]{\cbvAprxBneS_\indexK}{\cbvAprxBnoS_\indexK}
            \vsep \cbvAprxBneS_\indexK\esub{x}{\cbvAprxBneS_\indexK}
            \qquad
    \\[0.2cm]
        \cbvAprxBnoS_0 &\coloneqq& \abs{x}{\aprxT}
            \vsep \cbvAprxBvrS_0
            \vsep \cbvAprxBneS_0
            \vsep \cbvAprxBnoS_0\esub{x}{\cbvAprxBneS_0}
    \\
        \cbvAprxBnoS_{\indexI+1} &\coloneqq& \abs{x}{\cbvAprxBnoS_\indexI}
            \vsep \cbvAprxBvrS_{\indexI+1}
            \vsep \cbvAprxBneS_{\indexI+1}
            \vsep \cbvAprxBnoS_{\indexI+1}\esub{x}{\cbvAprxBneS_{\indexI+1}}
            \qquad (\indexI \in \Nat)
     \\
        \cbvAprxBnoS_\indexOmega &\coloneqq& \abs{x}{\cbvAprxBnoS_\indexOmega}
            \vsep \cbvAprxBvrS_\indexOmega
            \vsep \cbvAprxBneS_\indexOmega
            \vsep \cbvAprxBnoS_\indexOmega\esub{x}{\cbvAprxBneS_\indexOmega}
   \end{array}
\end{equation*}

\CbvObservabilityNormalFormApproximant*
\label{prf:Cbv_t_in_BnoSn_==>_OA(t)_in_noSn}%
\stableProof{
    \begin{proof}

We strengthen the induction hypothesis:
\begin{itemize}
\item[\bltI] If $t \in \cbvVrS_\indexK$ then $\cbvTApprox{t} \in
    \cbvAprxBvrS_\indexK$.
\item[\bltI] If $t \in \cbvNeS_\indexK$ then $\cbvTApprox{t} \in
    \cbvAprxBneS_\indexK$.
\item[\bltI] If $t \in \cbvNoS_\indexK$ then $\cbvTApprox{t} \in
    \cbvAprxBnoS_\indexK$.
\end{itemize}
Since $t \in \cbnNoS_\indexK$ then in particular $t \in \cbnNoS_0$
thus $t$ is \CBVSymb-meaningful. By mutual induction on $t \in
\cbvVrS_\indexK$, $t \in \cbvNeS_\indexK$ and $t \in \cbvNoS_\indexK$.
Cases:
\begin{itemize}
\item[\bltI] $t \in \cbvVrS_\indexK$: We distinguish two cases:
    \begin{itemize}
    \item[\bltII] $t = x$: Then $\cbvTApprox{t} = x \in
        \cbvAprxBvrS_\indexK$.

    \item[\bltII] $t = t_1\esub{x}{t_2}$ with $t_1 \in
        \cbvVrS_\indexK$ and $t_2 \in \cbvNeS_\indexK$: By \ih on
        $t_1$ and $t_2$, one has that $\cbvTApprox{t_1} \in
        \cbvAprxBvrS_\indexK$ and $\cbvTApprox{t_2} \in
        \cbvAprxBneS_\indexK$ thus $\cbvTApprox{t} =
        \cbvTApprox{t_1}\esub{x}{\cbvTApprox{t_2}} \subseteq
        \cbvAprxBvrS_\indexK\esub{x}{\cbvAprxBneS_\indexK} \subseteq
        \cbvAprxBvrS_\indexK$.
    \end{itemize}

\item[\bltI] $t \in \cbvNeS_\indexK$: We distinguish three cases: 
    \begin{itemize}
    \item[\bltII] $t = \app{t_1}{t_2}$ with $t_1 \in\cbvVrS_\indexK$
        and $t_2 \in \cbvNoS_\indexK$: Then by \ih on $t_1$ and $t_2$,
        one has that $\cbvTApprox{t_1} \in \cbvAprxBvrS_\indexK$ and
        $\cbvTApprox{t_2} \in \cbvAprxBnoS_\indexK$ thus
        $\cbvTApprox{t} = \app{\cbvTApprox{t_1}}{\cbvTApprox{t_2}} \in
        \app{\cbvAprxBvrS_\indexK}{\cbvAprxBnoS_\indexK} \subseteq
        \cbvAprxBneS_\indexK$.

    \item[\bltII] $t = \app{t_1}{t_2}$ with $t_1 \in \cbvNeS_\indexK$
        and $t_2 \in \cbvNoS_\indexK$: Then by \ih on $t_1$ and $t_2$,
        one has that $\cbvTApprox{t_1} \in \cbvAprxBneS_\indexK$ and
        $\cbvTApprox{t_2} \in \cbvAprxBnoS_\indexK$ thus
        $\cbvTApprox{t} = \app{\cbvTApprox{t_1}}{\cbvTApprox{t_2}} \in
        \app{\cbvAprxBneS_\indexK}{\cbvAprxBnoS_\indexK} \subseteq
        \cbvAprxBneS_\indexK$.

    \item[\bltII] $t = t_1\esub{x}{t_2}$ with $t_1, t_2 \in
        \cbvNeS_\indexK$: By \ih on $t_1$ and $t_2$, one has that
        $\cbvTApprox{t_1} \in \cbvAprxBneS_\indexK$ and
        $\cbvTApprox{t_2} \in \cbvAprxBneS_\indexK$ thus
        $\cbvTApprox{t} = \cbvTApprox{t_1}\esub{x}{\cbvTApprox{t_2}}
        \subseteq \cbvAprxBneS_\indexK\esub{x}{\cbvAprxBneS_\indexK}
        \subseteq \cbvAprxBneS_\indexK$.
    \end{itemize}

\item[\bltI] $t \in \cbvNoS_\indexK$: We distinguish four cases:
    \begin{itemize}
    \item[\bltII] $t = \abs{x}{t'}$: We distinguish three cases:
        \begin{itemize}
        \item[\bltIII] $\indexK = 0$: Then $t = \abs{x}{t'}$ and
            $\cbvTApprox{t} = \abs{x}{\cbvTApprox{t'}} \in
            \cbvAprxBnoS_0$.

        \item[\bltIII] $\indexK = \indexI + 1$ and $t' \in
            \cbvNoS_\indexI$ for some $\indexI \in \mathbb{N}$: By \ih
            on $t'$, one has that $\cbvTApprox{t'} \in
            \cbvAprxBnoS_\indexI$ thus $\cbvTApprox{t} =
            \abs{x}{\cbvTApprox{t'}} \in \abs{x}{\cbvAprxBnoS_\indexI}
            \subseteq \cbvAprxBnoS_\indexK$.

        \item[\bltIII] $\indexK = \indexOmega$ and $t' \in
            \cbvNoS_\indexOmega$: By \ih on $t'$, one has that
            $\cbvTApprox{t'} \in \cbvAprxBnoS_\indexOmega$ thus
            $\cbvTApprox{t} = \abs{x}{\cbvTApprox{t'}} \in
            \abs{x}{\cbvAprxBnoS_\indexOmega} \subseteq
            \cbvAprxBnoS_\indexOmega$.
        \end{itemize}

    \item[\bltII] $t \in \cbvVrS_\indexK$: Then by \ih on $t$, one has
        that $\cbvTApprox{t} \in \cbvAprxBvrS_\indexK \subseteq
        \cbvAprxBnoS_\indexK$.

    \item[\bltII] $t \in \cbvNeS_\indexK$: Then by \ih on $t$, one has
        that $\cbvTApprox{t} \in \cbvAprxBneS_\indexK \subseteq
        \cbvAprxBnoS_\indexK$.

    \item[\bltII] $t = t_1\esub{x}{t_2}$ with $t_1 \in
        \cbvNoS_\indexK$ and $t_2 \in \cbvNeS_\indexK$: By \ih on
        $t_1$ and $t_2$, one has that $\cbvTApprox{t_1} \in
        \cbvAprxBnoS_\indexK$ and $\cbvTApprox{t_2} \in
        \cbvAprxBneS_\indexK$ thus $\cbvTApprox{t} =
        \cbvTApprox{t_1}\esub{x}{\cbvTApprox{t_2}} \subseteq
        \cbvAprxBnoS_\indexK\esub{x}{\cbvAprxBneS_\indexK} \subseteq
        \cbvAprxBnoS_\indexK$.
    \end{itemize}
\end{itemize}
\end{proof}

}

Let us recall inductive definitions of stratified equality, where
$\indexK \in \setOrdinals$ and $\indexI \in \Nat$:
\begin{equation*}
    \begin{array}{c}
    \hspace{-0.2cm}
        \begin{prooftree}
            \hypo{\phantom{\cbvEq[\indexI]}}
            \inferCbvEqVar[1]{x \cbvEq[\indexK] x}
        \end{prooftree}
    \hspace{0.4cm}
        \begin{prooftree}
            \hypo{t_1 \cbvEq[\indexK] u_1}
            \hypo{t_2 \cbvEq[\indexK] u_2}
            \inferCbvEqApp{\app{t_1}{t_2} \cbvEq[\indexK] \app{u_1}{u_2}}
        \end{prooftree}
    \hspace{0.4cm}
        \begin{prooftree}
            \hypo{t_1 \cbvEq[\indexK] u_1}
            \hypo{t_2 \cbvEq[\indexK] u_2}
            \inferCbvEqEs{t_1\esub{x}{t_2} \cbvEq[\indexK] u_1\esub{x}{u_2}}
        \end{prooftree}
    \\[0.4cm]
       \begin{prooftree}
            \hypo{\phantom{t \cbvEq[\indexK] u}}
            \inferCbvEqAbs{\abs{x}{t} \cbvEq[0] \abs{y}{u}}
        \end{prooftree}
    \hspace{0.8cm}
    	\begin{prooftree}
    	\hypo{t \cbvEq[\indexI] u }
    	\inferCbvEqAbs{\abs{x}{t} \cbvEq[\indexI+1] \abs{x}{u}}
    	\end{prooftree}
    \hspace{0.8cm}
         \begin{prooftree}
            \hypo{t \cbvEq[\indexOmega] u }
            \inferCbvEqAbs{\abs{x}{t} \cbvEq[\indexOmega] \abs{x}{u}}
        \end{prooftree}
    \end{array}
\end{equation*}

\CbvStabilityMeaningfulObservables*
\label{prf:Cbv_|t_in_bnoSn_and_t_<=_u_==>_u_in_BnoSn}%
\stableProof{
    \begin{proof}

Let $\aprxT, \aprxU \in \setCbvAprxTerms$ such that $\aprxT
\cbvAprxLeq \aprxU$ We strengthen the induction hypothesis:
\begin{itemize}
\item[\bltI] If $\aprxT \in \cbvAprxBvrS_\indexK$ then $\aprxU \in
    \cbvAprxBvrS_\indexK$.

\item[\bltI] If $\aprxT \in \cbvAprxBneS_\indexK$ then $\aprxU \in
    \cbvAprxBneS_\indexK$.

\item[\bltI] If $\aprxT \in \cbvAprxBnoS_\indexK$ then $\aprxU \in
    \cbvAprxBnoS_\indexK$.
\end{itemize}
Moreover, $\aprxT \cbvEq[\indexK] \aprxU$.

We prove it by mutual induction on $\aprxT \in \cbvAprxBvrS_\indexK$,
$\aprxT \in \cbvAprxBneS_\indexK$ and $\aprxT \in
\cbvAprxBnoS_\indexK$:
\begin{itemize}
\item[\bltI] $\aprxT \in \cbvAprxBvrS_\indexK$:
    \begin{itemize}
    \item[\bltII] $\aprxT = x$: Then necessarily $\aprxU = x$ thus
        $\aprxU \in \cbvAprxBvrS_\indexK$ and $\aprxT \cbvEq[\indexK]
        \aprxU$ using $(\cbvEqVarRuleName)$.

    \item[\bltII] $\aprxT = \aprxT_1\esub{x}{\aprxT_2}$ with $\aprxT_1
        \in \cbvAprxBvrS_\indexK$ and $\aprxT_2 \in
        \cbvAprxBneS_\indexK$. Then necessarilyn $\aprxU =
        \aprxU_1\esub{x}{\aprxU_2}$ with $\aprxT_1 \cbvAprxLeq
        \aprxU_1$ and $\aprxT_2 \cbvAprxLeq \aprxU_2$ for some
        $\aprxU_1, \aprxU_2 \in \setCbvAprxTerms$. By \ih on
        $\aprxT_1$ and $\aprxT_2$, one has that $\aprxU_1 \in
        \cbvAprxBvrS_\indexK$, $\aprxU_2 \in \cbvAprxBneS_\indexK$,
        $\aprxT_1 \cbvEq[\indexK] \aprxU_1$ and $\aprxT_2
        \cbvEq[\indexK] \aprxU_2$. Thus $\aprxU =
        \aprxU_1\esub{x}{\aprxU_2} \in
        \cbvAprxBvrS_\indexK\esub{x}{\cbvAprxBneS_\indexK} \subseteq
        \cbvAprxBvrS_\indexK$ and $\aprxT \cbvEq[\indexK] \aprxU$
        using $(\cbvEqEsRuleName)$.
    \end{itemize}

\item[\bltI] $\aprxT \in \cbvAprxBneS_\indexK$:
    \begin{itemize}
    \item[\bltII] $\aprxT = \app{\aprxT_1}{\aprxT_2}$ with $\aprxT_1
        \in \cbvAprxBvrS_\indexK$ and $\aprxT_2 \in
        \cbvAprxBnoS_\indexK$: Necessarily $\aprxU =
        \app{\aprxU_1}{\aprxU_2}$ with $\aprxT_1 \cbvAprxLeq \aprxU_1$
        and $\aprxT_2 \cbvAprxLeq \aprxU_2$. By \ih on $\aprxT_1$ and
        $\aprxT_2$, one has $\aprxU_1 \in \cbvAprxBvrS_\indexK$,
        $\aprxU_2 \in \cbvAprxBnoS_\indexK$, $\aprxT_1 \cbvEq[\indexK]
        \aprxU_1$ and $\aprxT_2 \cbvEq[\indexK] \aprxU_2$. Thus
        $\aprxU = \app{\aprxU_1}{\aprxU_2} \in
        \app[\,]{\cbvAprxBvrS_\indexK}{\cbvAprxBnoS_\indexK} \subseteq
        \cbvAprxBneS_\indexK$ and $\aprxT \cbvEq[\indexK] \aprxU$
        using $(\cbvEqAppRuleName)$.

    \item[\bltII] $\aprxT = \app{\aprxT_1}{\aprxT_2}$ with $\aprxT_1
        \in \cbvAprxBneS_\indexK$ and $\aprxT_2 \in
        \cbvAprxBnoS_\indexK$: Then necessarily $\aprxU =
        \app{\aprxU_1}{\aprxU_2}$ with $\aprxT_1 \cbvAprxLeq \aprxU_1$
        and $\aprxT_2 \cbvAprxLeq \aprxU_2$. By \ih on $\aprxT_1$ and
        $\aprxT_2$, one has that $\aprxU_1 \in \cbvAprxBneS_\indexK$,
        $\aprxU_2 \in \cbvAprxBnoS_\indexK$, $\aprxT_1 \cbvEq[\indexK]
        \aprxU_1$ and $\aprxT_2 \cbvEq[\indexK] \aprxU_2$. Thus
        $\aprxU = \app{\aprxU_1}{\aprxU_2} \in
        \app{\cbvAprxBneS_\indexK}{\cbvAprxBnoS_\indexK} \subseteq
        \cbvAprxBneS_\indexK$ and $\aprxT \cbvEq[\indexK] \aprxU$
        using $(\cbvEqAppRuleName)$.

    \item[\bltII] $\aprxT = \aprxT_1\esub{x}{\aprxT_2}$ with
        $\aprxT_1, \aprxT_2 \in \cbvAprxBneS_\indexK$: Then
        necessarily $\aprxU = \aprxU_1\esub{x}{\aprxU_2}$ with
        $\aprxT_1 \cbvAprxLeq \aprxU_1$ and $\aprxT_2 \cbvAprxLeq
        \aprxU_2$. By \ih on $\aprxT_1$ and $\aprxT_2$, one has that
        $\aprxU_1 \in \cbvAprxBneS_\indexK$, $\aprxU_2 \in
        \cbvAprxBneS_\indexK$, $\aprxT_1 \cbvEq[\indexK] \aprxU_1$ and
        $\aprxT_2 \cbvEq[\indexK] \aprxU_2$. Thus $\aprxU =
        \aprxU_1\esub{x}{\aprxU_2} \in
        \cbvAprxBneS_\indexK\esub{x}{\cbvAprxBneS_\indexK} \subseteq
        \cbvAprxBneS_\indexK$ and $\aprxT \cbvEq[\indexK] \aprxU$
        using $(\cbvEqEsRuleName)$.
    \end{itemize}

\item[\bltI] $\aprxT \in \cbvAprxBnoS_\indexK$:
    \begin{itemize}
    \item[\bltII] $\aprxT = \abs{x}{\aprxT'}$: Then necessarily
        $\aprxU = \abs{x}{\aprxU'}$ with $\aprxT' \cbvAprxLeq
        \aprxU'$. We distinguish three cases:
        \begin{itemize}
        \item[\bltIII] $\indexK = 0$: Then $\aprxU = \abs{x}{\aprxU'}
            \in \cbvAprxBnoS_0$.

        \item[\bltIII] $\indexK = \indexI + 1$ for some $\indexI \in
            \mathbb{N}$: Then $\aprxT' \in \cbvAprxBnoS_\indexI$ and
            by \ih on $\aprxT'$ one has $\aprxU'\! \in
            \cbvAprxBnoS_\indexI$ and $\aprxT' \cbvEq[\indexI]
            \aprxU'$. Thus $\aprxU = \abs{x}{\aprxU'\!} \in
            \abs{x}{\cbvAprxBnoS_\indexI} \subseteq
            \cbvAprxBnoS_\indexK$ and $\aprxT \cbvEq[\indexK] \aprxU$
            using $(\cbvEqAbsRuleName)$.

        \item[\bltIII] $\indexK = \indexOmega$: Then $\aprxT' \in
            \cbvAprxBnoS_\indexOmega$ and by \ih on $\aprxT'$ one has
            $\aprxU'\! \in \cbvAprxBnoS_\indexOmega$, so $\aprxU =
            \abs{x}{\aprxU'\!} \in \abs{x}{\cbvAprxBnoS_\indexOmega}
            \subseteq \cbvAprxBnoS_\indexK$ and $\aprxT
            \cbvEq[\indexOmega] \aprxU$ using $(\cbvEqAbsRuleName)$.
        \end{itemize}

    \item[\bltII] $\aprxT \in \cbvAprxBvrS_\indexK$: By \ih on
        $\aprxT$, one has that $\aprxU \in \cbvAprxBvrS_\indexK
        \subseteq \cbvAprxBnoS_\indexK$ and $\aprxT \cbvEq[\indexK]
        \aprxU$.

    \item[\bltII] $\aprxT \in \cbvAprxBneS_\indexK$: By \ih on
        $\aprxT$, one has that $\aprxU \in \cbvAprxBneS_\indexK
        \subseteq \cbvAprxBnoS_\indexK$ and $\aprxT \cbvEq[\indexK]
        \aprxU$.

    \item[\bltII] $\aprxT = \aprxT_1\esub{x}{\aprxT_2}$ with $\aprxT_1
        \in \cbvAprxBnoS_\indexK$ and $\aprxT_2 \in
        \cbvAprxBneS_\indexK$: Then necessarily $\aprxU =
        \aprxU_1\esub{x}{\aprxU_2}$ with $\aprxT_1 \cbvAprxLeq
        \aprxU_1$ and $\aprxT_2 \cbvAprxLeq \aprxU_2$. By \ih on
        $\aprxT_1$ and $\aprxT_2$, one has that $\aprxU_1 \in
        \cbvAprxBnoS_\indexK$, $\aprxU_2 \in \cbvAprxBneS_\indexK$,
        $\aprxT_1 \cbvEq[\indexK] \aprxU_1$ and $\aprxT_2
        \cbvEq[\indexK] \aprxU_2$. Thus $\aprxU =
        \aprxU_1\esub{x}{\aprxU_2} \in
        \cbvAprxBnoS_\indexK\esub{x}{\cbvAprxBneS_\indexK} \subseteq
        \cbvAprxBnoS_\indexK$ and $\aprxT \cbvEq[\indexK] \aprxU$
        using $(\cbvEqEsRuleName)$.
        \qedhere
    \end{itemize}
\end{itemize}
\end{proof}

}

\section{Proof of Section \ref{sec:cbn}}

\subsection{Proof of Subsection \ref{sec:syntax-cbn}}

Let us recall the mutual inductive definitions of $\cbnNeS_\indexK$
and $\cbnNoS_\indexK$:
\begin{equation*}
    \begin{array}{c}
        \begin{array}{rcl}
        \cbnNeS_0 &\coloneqq& x
            \vsep \app{\cbnNeS_0}{t}
        \\
        \cbnNeS_{\indexI+1} &\coloneqq& x
            \vsep \app{\cbnNeS_{\indexI+1}}{\cbnNoS_\indexI}
                    \qquad (\indexI \in \setIntegers)
        \\
        \cbnNeS_\indexOmega &\coloneqq& x
            \vsep \app{\cbnNeS_\indexOmega}{\cbnNoS_\indexOmega}
    \\[0.2cm]
            \cbnNoS_\indexK &\coloneqq& \abs{x}{\cbnNoS_\indexK}
                \vsep \cbnNeS_\indexK
                \qquad (\indexK \in \setOrdinals)
        \end{array}
    \end{array}
\end{equation*}

\CbnNoSnCharacterization* \label{prf:cbn_NnoNn_Characterization}%
\stableProof{
    \begin{proof}
Let $\indexK \in \setOrdinals$ and $t \in \setCbnTerms$. We strengthen
the property as follows:
\begin{itemize}
\item[\bltI] $t \in \cbnNeS_\indexK$ if and only if $t$ is a
    $\cbnSymbSurface_\indexK$-normal form and $\neg\cbvAbsPred{t}$.

\item[\bltI] $t \in \cbnNoS_\indexK$ if and only if $t$ is a
    $\cbnSymbSurface_\indexK$-normal form.
\end{itemize}
By double implication:
\begin{itemize}
\item[$(\Rightarrow)$] By mutual induction on $t \in \cbnNeS_\indexK$
    and $t \in \cbnNoS_\indexK$:
    \begin{itemize}
    \item[\bltI] $t \in \cbnNeS_\indexK$: We distinguish two cases: 
        \begin{itemize}
        \item[\bltII] $t = x$: Then $t$ is a
            $\cbnSymbSurface_\indexK$-normal form and
            $\neg\cbnAbsPred{t}$.

        \item[\bltII] $t = \app{t_1}{t_2}$ with $t_1 \in
            \cbnNeS_\indexK$: Then, by \ih on $t_1$, one has that
            $t_1$ is a $\cbnSymbSurface_\indexK$-normal form and
            $\neg\cbnAbsPred{t_1}$ thus $t = \app{t_1}{t_2}$ is not a
            redex. Moreover, $\neg\cbnAbsPred{t}$. We distinguish
            three cases:
            \begin{itemize}
            \item[\bltIII] $\indexK = 0$: Then $t = \app{t_1}{t_2}$ is
                $\cbnSymbSurface_0$-normal form.

            \item[\bltIII] $\indexK = \indexI + 1$ for some $\indexI
                \in \mathbb{N}$: Then $t_2 \in \cbnNoS_\indexI$ thus
                by \ih, one has that $t_2$ is a
                $\cbnSymbSurface_\indexI$-normal form. Since $t$ is
                not a redex then it is a
                $\cbnSymbSurface_\indexK$-normal form.

            \item[\bltIII] $\indexK = \indexOmega$: Then $t_2 \in
                \cbnNoS_\indexOmega$ thus by \ih, one has that $t_2$
                is a $\cbnSymbSurface_\indexOmega$-normal form. Since
                $t$ is not a redex then it is a
                $\cbnSymbSurface_\indexOmega$-normal form.
            \end{itemize}
        \end{itemize}

    \item[\bltI] $t \in \cbnNoS_\indexK$: We distinguish two cases:
        \begin{itemize}
        \item[\bltII] $t = \abs{x}{t'}$ with $t' \in \cbnNoS_\indexK$:
            By \ih on $t'$, one has that $t'$ is a
            $\cbnSymbSurface_\indexK$-normal form thus $\abs{x}{t'}$
            is also a $\cbnSymbSurface_\indexK$-normal form.

        \item[\bltII] $t \in \cbnNeS_\indexK$: Then by \ih on $t$, one
            deduces that $t$ is a $\cbnSymbSurface_\indexK$-normal
            form.
        \end{itemize}
    \end{itemize}

\item[$(\Leftarrow)$] By induction on $t$:
    \begin{itemize}
    \item[\bltI] $t = x$: Then $\neg\cbnAbsPred{t}$ and $t \in
        \cbnNeS_\indexK$.

    \item[\bltI] $t = \abs{x}{t'}$: Then $\cbnAbsPred{t}$ and $t'$ is
        necessarily a $\cbnSymbSurface_\indexK$-normal form. By \ih on
        $t'$, one has that $t' \in \cbnNoS_\indexK$ thus $t =
        \abs{x}{t'} \in \abs{x}{\cbnNoS_\indexK} \subseteq
        \cbnNoS_\indexK$.

    \item[\bltI] $t = \app{t_1}{t_2}$: Then $\neg\cbnAbsPred{t}$ and
        necessarily, $\neg\cbvAbsPred{t}$ and $t_1$ is a
        $\cbnSymbSurface_\indexK$-normal form. By \ih on $t_1$, one
        deduces that $t_1 \in \cbnNeS_\indexK$. We distinguish three
        cases:
        \begin{itemize}
        \item[\bltII] $\indexK = 0$: Then $t = \app{t_1}{t_2} \in
        \app{\cbnNeS_0}{t_2} \subseteq \cbnNeS_0$.

        \item[\bltII] $\indexK = \indexI + 1$ for some $\indexI \in
            \mathbb{N}$: Then $t_2 \in \cbnNoS_\indexI$. By \ih on
            $t_2$, one has that $t_2 \in \cbnNoS_\indexI$ thus $t =
            \app{t_1}{t_2} \in \app{\cbnNeS_\indexK}{\cbnNoS_\indexI}
            \subseteq \cbnNeS_\indexK$.

        \item[\bltII] $\indexK = \indexOmega$: Then $t_2 \in
            \cbnNoS_\indexOmega$. By \ih on $t_2$, one has that $t_2
            \in \cbnNoS_\indexOmega$ thus $t = \app{t_1}{t_2} \in
            \app{\cbnNeS_\indexOmega}{\cbnNoS_\indexOmega} \subseteq
            \cbnNeS_\indexOmega$.
        \end{itemize}

    \item[\bltI] $t = t_1\esub{x}{t_2}$: Impossible since it
        contradict the fact that $t$ is a $\cbnSymbSurface_\indexK$.
    \end{itemize}
\end{itemize}
\end{proof}
}

\CbnTypedSurfaceGenericity*
\label{prf:Cbn_t_meaningless_and_Pi_|>_F<t>_==>_Pi'_|>_F<u>}%
\stableProof{
    \begin{proof}
By induction on $\cbnCCtxt$:
\begin{itemize}
\item[\bltI] $\cbnCCtxt = \Hole$: Then $\cbnCCtxt<t> = t$ is
    meaningless, which is impossible since it contradicts
    \Cref{lem:cbnBKRV_characterizes_meaningfulness}.

\item[\bltI] $\cbnCCtxt = \abs{x}{\cbnCCtxt'}$: Then $\Pi$ is
    necessarily of the following form:
    \begin{equation*}
        \begin{prooftree}
            \hypo{\Phi \cbnTrBKRV \Gamma, x : \M \vdash \cbnCCtxt'<t> : \tau}
            \inferCbnBKRVAbs{\Gamma \vdash \abs{x}{\cbnCCtxt'<t>} : \M \typeArrow \tau\quad}
        \end{prooftree}
    \end{equation*}
    with $\sigma = \M \typeArrow \tau$. By \ih on $\Phi$, one has
    $\Phi' \cbnTrBKRV \Gamma, x : \M \cbnTrBKRV \cbnCCtxt'<u> : \tau$
    and we set $\Pi'$ as the following derivation:
    \begin{equation*}
        \begin{prooftree}
            \hypo{\Phi' \cbnTrBKRV \Gamma, x : \M \vdash \cbnCCtxt'<u> : \tau}
            \inferCbnBKRVAbs{\Gamma \vdash \abs{x}{\cbnCCtxt'<u>} : \M \typeArrow \tau\quad}
        \end{prooftree}
    \end{equation*}

\item[\bltI] $\cbnCCtxt = \app{\cbnCCtxt'}{s}$: Then $\Pi$ is
    necessarily of the following form:
    \begin{equation*}
        \begin{prooftree}
            \hypo{\Pi_1 \cbnTrBKRV \Gamma_1 \vdash \cbnCCtxt'<t> : \mset{\tau_i}_{i \in I} \typeArrow \sigma}
            \hypo{\Pi_2^i \cbnTrBKRV \Gamma_2^i \vdash s : \tau_i}
            \delims{\left(}{\right)_{i \in I}}
            \inferCbnBKRVApp{\Gamma_1 +_{i \in I} \Gamma_2^i &\vdash \app{\cbnCCtxt'<t>}{s} : \sigma\hspace{2.5cm}}
        \end{prooftree}
    \end{equation*}
    with $I$ finite and $\Gamma = \Gamma_1 +_{i \in I} \Gamma_2^i$. By
    \ih on $\Pi_1$, one has $\Pi'_1 \cbnTrBKRV \Gamma_1 \vdash
    \cbnCCtxt'<u> : \mset{\tau_i}_{i \in I} \typeArrow \sigma$ and we
    set $\Pi'$ as the following derivation:
    \begin{equation*}
        \begin{prooftree}
            \hypo{\Pi'_1 \cbnTrBKRV \Gamma_1 \vdash \cbnCCtxt'<u> : \mset{\tau_i}_{i \in I} \typeArrow \sigma}
            \hypo{\Pi_2^i \cbnTrBKRV \Gamma_2^i \vdash s : \tau_i}
            \delims{\left(}{\right)_{i \in I}}
            \inferCbnBKRVApp{\Gamma_1 +_{i \in I} \Gamma_2^i &\vdash \app{\cbnCCtxt'<t>}{s} : \sigma\hspace{2.5cm}}
        \end{prooftree}
    \end{equation*}

\item[\bltI] $\cbnCCtxt = \app[\,]{s}{\cbnCCtxt'}$: Then $\Pi$ is
    necessarily of the following form:
    \begin{equation*}
        \begin{prooftree}
            \hypo{\Pi_1 \cbnTrBKRV \Gamma_1 \vdash s : \mset{\tau_i}_{i \in I} \typeArrow \sigma}
            \hypo{\Pi_2^i \cbnTrBKRV \Gamma_2^i \vdash \cbnCCtxt'<t> : \tau_i}
            \delims{\left(}{\right)_{i \in I}}
            \inferCbnBKRVApp{\Gamma_1 +_{i \in I} \Gamma_2^i &\vdash \app[\,]{s}{\cbnCCtxt'<t>} : \sigma\hspace{2.5cm}}
        \end{prooftree}
    \end{equation*}
    with $I$ finite and $\Gamma = \Gamma_1 +_{i \in I} \Gamma_2^i$.
    For each $i \in I$, by \ih on $\Pi_2^i$, one has ${\Pi_2^i}'
    \cbnTrBKRV \Gamma_2^i \vdash \cbnCCtxt'<u> : \tau_i$ and we set
    $\Pi'$ as the following derivation:
    \begin{equation*}
        \begin{prooftree}
            \hypo{\Pi_1 \cbnTrBKRV \Gamma_1 \vdash s : \mset{\tau_i}_{i \in I} \typeArrow \sigma}
            \hypo{{\Pi_2^i}' \cbnTrBKRV \Gamma_2^i \vdash \cbnCCtxt'<u> : \tau_i}
            \delims{\left(}{\right)_{i \in I}}
            \inferCbnBKRVApp{\Gamma_1 +_{i \in I} \Gamma_2^i &\vdash \app[\,]{s}{\cbnCCtxt'<u>} : \sigma\hspace{2.5cm}}
        \end{prooftree}
    \end{equation*}

\item[\bltI] $\cbnCCtxt = \cbnCCtxt'\esub{x}{s}$: Then $\Pi$ is
    necessarily of the following form:
    \begin{equation*}
        \begin{prooftree}
            \hypo{\Pi_1 \cbnTrBKRV \Gamma_1, x : \mset{\tau_i}_{i \in I} \vdash \cbnCCtxt'<t> : \sigma}
            \hypo{\Pi_2^i \cbnTrBKRV \Gamma_2^i \vdash s : \tau_i}
            \delims{\left(}{\right)_{i \in I}}
            \inferCbnBKRVEs{\Gamma_1 +_{i \in I} \Gamma_2^i &\vdash \cbnCCtxt'<t>\esub{x}{s} : \sigma\hspace{1.9cm}}
        \end{prooftree}
    \end{equation*}
    with $I$ finite and $\Gamma = \Gamma_1 +_{i \in I} \Gamma_2^i$. By
    \ih on $\Pi_1$, one has $\Pi'_1 \cbnTrBKRV \Gamma_1, x :
    \mset{\tau_i}_{i \in I} \vdash \cbnCCtxt'<u> : \sigma$ and we set
    $\Pi'$ as the following derivation:
    \begin{equation*}
        \begin{prooftree}
            \hypo{\Pi'_1 \cbnTrBKRV \Gamma_1, x : \mset{\tau_i}_{i \in I} \vdash \cbnCCtxt'<u> : \sigma}
            \hypo{\Pi_2^i \cbnTrBKRV \Gamma_2^i \vdash s : \tau_i}
            \delims{\left(}{\right)_{i \in I}}
            \inferCbnBKRVEs{\Gamma_1 +_{i \in I} \Gamma_2^i &\vdash \cbnCCtxt'<u>\esub{x}{s} : \sigma\hspace{1.9cm}}
        \end{prooftree}
    \end{equation*}

\item[\bltI] $\cbnCCtxt = s\esub{x}{\cbnCCtxt'}$: Then $\Pi$ is
    necessarily of the following form:
    \begin{equation*}
        \begin{prooftree}
            \hypo{\Pi_1 \cbnTrBKRV \Gamma_1, x : \mset{\tau_i}_{i \in I} \vdash s : \sigma}
            \hypo{\Pi_2^i \cbnTrBKRV \Gamma_2^i \vdash \cbnCCtxt'<t> : \tau_i}
            \delims{\left(}{\right)_{i \in I}}
            \inferCbnBKRVEs{\Gamma_1 +_{i \in I} \Gamma_2^i &\vdash s\esub{x}{\cbnCCtxt'<t>} : \sigma\hspace{1.9cm}}
        \end{prooftree}
    \end{equation*}
    with $I$ finite and $\Gamma = \Gamma_1 +_{i \in I} \Gamma_2^i$.
    For each $i \in I$, by \ih on $\Pi_2^i$, one has $\Pi_2^i
    \cbnTrBKRV \Gamma_2^i \vdash \cbnCCtxt'<u> : \tau_i$ and we set
    $\Pi'$ as the following derivation:
    \begin{equation*}
        \begin{prooftree}
            \hypo{\Pi_1 \cbnTrBKRV \Gamma_1, x : \mset{\tau_i}_{i \in I} \vdash s : \sigma}
            \hypo{{\Pi_2^i}' \cbnTrBKRV \Gamma_2^i \vdash \cbnCCtxt<u> : \tau_i}
            \delims{\left(}{\right)_{i \in I}}
            \inferCbnBKRVEs{\Gamma_1 +_{i \in I} \Gamma_2^i &\vdash s\esub{x}{\cbnCCtxt<u>} : \sigma\hspace{1.9cm}}
        \end{prooftree}
    \end{equation*}
\end{itemize}
\end{proof}

}

\CbnQualitativeSurfaceGenericity*%
\label{prf:Cbn_Qualitative_Surface_Genericity}%
\stableProof{
    \begin{proof}
Let $t \in \setCbnTerms$ such that $t$ is \CBNSymb-meaningless and
$\cbnCCtxt<t>$ is \CBNSymb-meaningful. Using
\Cref{lem:cbnBKRV_characterizes_meaningfulness}, one has $\Pi
\cbnTrBKRV \Gamma \vdash \cbnCCtxt<t> : \sigma$ for some context
$\Gamma$ and some type $\sigma$. By
\Cref{lem:Cbn_t_meaningless_and_Pi_|>_F<t>_==>_Pi'_|>_F<u>}, there is
$\Pi' \cbnTrBKRV \Gamma \vdash \cbnCCtxt<u> : \sigma$, thus
using \Cref{lem:cbnBKRV_characterizes_meaningfulness} again, one
concludes that $\cbnCCtxt<u>$ is \CBNSymb-meaningful.
\end{proof}

}

\subsection{Technical Lemmas}

\begin{lemma}
    \label{lem:Cbn_t1_<=_t2_and_u1_<=_u2_==>_t1{x:=u1}_<=_t2{x:=u2}}%
    Let $\aprxT_1, \aprxT_2, \aprxU_1, \aprxU_2 \in \setCbnAprxTerms$
    such that $\aprxT_1 \cbnAprxLeq \aprxT_2$ and $\aprxU_1
    \cbnAprxLeq \aprxU_2$, then $\aprxT_1\isub{x}{\aprxU_1}
    \cbnAprxLeq \aprxT_2\isub{x}{\aprxU_2}$.
\end{lemma}
\stableProof{
    \begin{proof}
By induction on $\aprxT_1$:
\begin{itemize}
\item[\bltI] $\aprxT_1 = \cbnAprxBot$: Then
    $\aprxT_1\isub{x}{\aprxU_1} = \cbnAprxBot \cbnAprxLeq
    \aprxT_2\isub{x}{\aprxU_2}$.

\item[\bltI] $\aprxT_1 = y$: Then $\aprxT_2 = y$ and we distinguish
    two cases:
    \begin{itemize}
    \item[\bltII] $x = y$: Thus $\aprxT_1\isub{x}{\aprxU_1} = \aprxU_1
        \cbnAprxLeq \aprxU_2 = \aprxT_2\isub{x}{\aprxU_2}$.

    \item[\bltII] $x \neq y$: Thus $\aprxT_1\isub{x}{\aprxU_1} = y =
        \aprxT_2\isub{x}{\aprxU_2}$, thus $\aprxT_1\isub{x}{\aprxU_1}
        \cbnAprxLeq \aprxT_2\isub{x}{\aprxU_2}$ by reflexivity of
        $\cbnAprxLeq$.
    \end{itemize}

\item[\bltI] $\aprxT_1 = \abs{y}{\aprxT'_1}$: Then
    $\aprxT_1\isub{x}{\aprxU_1} =
    \abs{y}{\aprxT'_1\isub{x}{\aprxU_1}}$ (by supposing $y \notin
    \freeVar{\aprxU_1} \cup \freeVar{\aprxU_2}$ without loss of
    generality) and necessarily $\aprxT_2 = \abs{y}{\aprxT'_2}$ with
    $\aprxT'_1 \cbnAprxLeq \aprxT'_2$ thus $\aprxT_2\isub{x}{\aprxU_2}
    = \abs{y}{\aprxT'_2\isub{y}{\aprxU_2}}$. By \ih on $\aprxT'_1$,
    one has that $\aprxT'_1\isub{x}{\aprxU_1} \cbnAprxLeq
    \aprxT'_2\isub{x}{\aprxU_2}$ hence $\aprxT_1\isub{x}{\aprxU_1} =
    \abs{y}{\aprxT'_1\isub{x}{\aprxU_1}} \cbnAprxLeq
    \abs{y}{\aprxT'_2\isub{x}{\aprxU_2}} = \aprxT_2\isub{x}{\aprxU_2}$
    by contextual closure.

\item[\bltI] $\aprxT_1 = \app{\aprxT_1^1}{\aprxT_1^2}$: Then
    $\aprxT_1\isub{x}{\aprxU_1} =
    \app{(\aprxT_1^1\isub{x}{\aprxU_1})}{(\aprxT_1^2\isub{x}{\aprxU_1})}$
    and necessarily $\aprxT_2 = \app{\aprxT_2^1}{\aprxT_2^2}$ with
    $\aprxT_1^1 \cbnAprxLeq \aprxT_2^1$ and $\aprxT_1^2 \cbnAprxLeq
    \aprxT_2^2$, thus $\aprxT_2\isub{y}{\aprxU_2} =
    \app{(\aprxT_2^1\isub{x}{\aprxU_2})}{(\aprxT_2^2\isub{x}{\aprxU_2})}$.
    By \ih on $\aprxT_1^1$ and $\aprxT_1^2$, one obtains that
    $\aprxT_1^1\isub{x}{\aprxU_1} \cbnAprxLeq
    \aprxT_2^1\isub{x}{\aprxU_2}$ and $\aprxT_1^2\isub{x}{\aprxU_1}
    \cbnAprxLeq \aprxT_2^2\isub{x}{\aprxU_2}$, hence
    $\aprxT_1\isub{x}{\aprxU_1} =
    \app{(\aprxT_1^1\isub{x}{\aprxU_1})}{(\aprxT_1^2\isub{x}{\aprxU_1})}
    \cbnAprxLeq
    \app{(\aprxT_2^1\isub{x}{\aprxU_2})}{(\aprxT_2^2\isub{x}{\aprxU_2})}
    = \aprxT_2\isub{x}{\aprxU_2}$ by contextual closure.

\item[\bltI] $\aprxT_1 = \aprxT_1^1\esub{y}{\aprxT_1^2}$: Then
    $\aprxT_1\isub{x}{\aprxU_1} =
    \aprxT_1^1\isub{x}{\aprxU_1}\esub{y}{\aprxT_1^2\isub{x}{\aprxU_1}}$
    and necessarily $\aprxT_2 = \aprxT_2^1\esub{y}{\aprxT_2^2}$ with
    $\aprxT_1^1 \cbnAprxLeq \aprxT_2^1$ and $\aprxT_1^2 \cbnAprxLeq
    \aprxT_2^2$, thus $\aprxT_2\isub{y}{\aprxU_2} =
    \aprxT_2^1\isub{x}{\aprxU_2}\esub{y}{\aprxT_2^2\isub{x}{\aprxU_2}}$.
    By \ih on $\aprxT_1^1$ and $\aprxT_1^2$, one has
    $\aprxT_1^1\isub{x}{\aprxU_1} \cbnAprxLeq
    \aprxT_2^1\isub{x}{\aprxU_2}$ and $\aprxT_1^2\isub{x}{\aprxU_1}
    \cbnAprxLeq \aprxT_2^2\isub{x}{\aprxU_2}$, so by contextual
    closure $\aprxT_1\isub{x}{\aprxU_1} =
    \aprxT_1^1\isub{x}{\aprxU_1}\esub{y}{\aprxT_1^2\isub{x}{\aprxU_1}}
    \cbnAprxLeq
    \aprxT_2^1\isub{x}{\aprxU_2}\esub{y}{\aprxT_2^2\isub{x}{\aprxU_2}}
    = \aprxT_2\isub{x}{\aprxU_2}$.
    \qedhere
\end{itemize}
\end{proof}

}

\begin{definition}
    The partial order $\cbnAprxLeq$ is extended to \CBNSymb contexts
    by setting $\Hole \cbnAprxLeq \Hole$.
\end{definition}

For example,
$\app{\Hole\esub{x}{\abs{y}{\cbnAprxBot}}}{(\app{z}{\cbnAprxBot})}
\cbnAprxLeq
\app{\Hole\esub{x}{\abs{y}{(\app{z}{z})}}}{(\app{z}{\cbnAprxBot})}$.

\begin{lemma}
    \label{lem:cbn_L1_<=_L2_and_t1_<=_t2_==>_L1<t1>_<=_L2<t2>}%
    Let $\cbnAprxLCtxt_1 \cbnAprxLeq \cbnAprxLCtxt_2$ and $\aprxT_1
    \cbnAprxLeq \aprxT_2$ for some list contexts $\cbnAprxLCtxt_1,
    \cbnAprxLCtxt_2$ and some $\aprxT_1, \aprxT_2 \in
    \setCbnAprxTerms$, then $\cbnAprxLCtxt_1<\aprxT_1> \cbnAprxLeq
    \cbnAprxLCtxt_2<\aprxT_2>$.
\end{lemma}
\stableProof{
    \begin{proof}
By induction on $\cbnAprxLCtxt_1$:
\begin{itemize}
\item[\bltI] $\cbnAprxLCtxt_1 = \Hole$: Then necessarily
    $\cbnAprxLCtxt_2 = \Hole$ thus $\cbnAprxLCtxt_1<\aprxT_1> =
    \aprxT_1 \cbnAprxLeq \aprxT_2 = \cbnAprxLCtxt_2<\aprxT_2>$.

\item[\bltI] $\cbnAprxLCtxt_1 = \cbnAprxLCtxt'_1\esub{x}{\aprxU_1}$:
    Then necessarily $\cbnAprxLCtxt_2 =
    \cbnAprxLCtxt'_2\esub{x}{\aprxU_2}$ with $\cbnAprxLCtxt'_1
    \cbnAprxLeq \cbnAprxLCtxt'_2$ and $\aprxU_1 \cbnAprxLeq \aprxU_2$.
    By \ih on $\cbnAprxLCtxt'_1$, one has that
    $\cbnAprxLCtxt'_1<\aprxT_1> \cbnAprxLeq
    \cbnAprxLCtxt'_2<\aprxT_2>$ thus $\cbnAprxLCtxt_1<\aprxT_1> =
    \cbnAprxLCtxt'_1<\aprxT_1>\esub{x}{\aprxU_1} \cbnAprxLeq
    \cbnAprxLCtxt'_2<\aprxT_2>\esub{x}{\aprxU_2} =
    \cbnAprxLCtxt_2<\aprxT_2>$ by contextual closure.
    \qedhere
\end{itemize}
\end{proof}

}%

\begin{lemma}
    \label{lem:cbn_t_meaningful_and_t_->Sn_u_=>_u_meaningful}%
    Let $t, u \in \setCbnTerms$ such that $t \cbnArr_{S_\indexK} u$.
    Then, $t$ is \CBNSymb-meaningful if and only if $u$ is
    \CBNSymb-meaningful.
\end{lemma}
\stableProof{
\begin{proof}
According to qualitative subject reduction and expansion
\cite{deCarvalho18}, given $t \cbnArr_{S_\indexOmega} u$, there is a
derivation $\Pi \cbnTrBKRV \Gamma \vdash t : \sigma$ if and only if
there is a derivation $\Pi' \cbnTrBKRV \Gamma \vdash u : \sigma$. As
$\cbnArr_{S_\indexK} \subseteq \cbnArr_{S_\indexOmega}$ for all
$\indexK \in \setIntegers$, the same equivalence also holds if we
replace the hypothesis $t \cbnArr_{S_\indexOmega} u$ with $t
\cbnArr_{S_\indexK} u$, for all $\indexK \in \setIntegers$. We
conclude thanks to the type-theoretical characterization of
\CBNSymb-meaningfulness
(\Cref{lem:cbnBKRV_characterizes_meaningfulness}).
\end{proof}
}

\begin{lemma}
    \label{lem:cbn_solvability_on_app_and_es}%
    Let $t, u \in \setCbnTerms$. Then
    \begin{enumerate}
    \item If $\app{t}{u}$ is \CBNSymb-meaningful then $t$ is
        \CBNSymb-meaningful.%
        \label{lem:cbn_meaningful_on_app}%

    \item If $t\esub{x}{u}$ is \CBNSymb-meaningful then $t$ is
        \CBNSymb-meaningful.%
        \label{lem:cbn_meaningful_on_esub}%

    \item If $t\isub{x}{u}$ is \CBNSymb-meaningful then $t$ is
        \CBNSymb-meaningful.%
        \label{lem:cbn_meaningful_on_isub}%
    \end{enumerate}
\end{lemma}
\stableProof{
\begin{proof} ~
\begin{enumerate}
\item Let $\app{t}{u}$ be \CBNSymb-meaningful and suppose by absurd
    that $t$ is \CBNSymb-meaningless. Using
    \Cref{lem:cbnBKRV_characterizes_meaningfulness}, we know that
    $\app{t}{u}$ \tSurface normalizes. Since $t$ is
    \CBNSymb-meaningless, there exists a $(t_i)_{i \in \mathbb{N}}$
    with $t_0 = t$ and $t_i \cbvArr_{S_0} t_{i+1}$ yielding the
    infinite path $\app{t}{u} \cbvArr_{S_0} \app{t_1}{u} \cbvArr_{S_0}
    \app{t_2}{u} \cbvArr_{S_0} \dots$ obtained by always reducing $t$.
    Since surface reduction is diamond, this contradicts the fact that
    $\app{t}{u}$ is surface normalizing.

\item Same principle as the previous item.

\item Let $t\isub{x}{v}$ be \CBNSymb-meaningful. By induction on $t$,
    one can show that if $t \cbvArr_{S_0} u$ then $t\isub{x}{v}
    \cbvArr_{S_0} u\isub{x}{v}$. Using this, one applied the same
    principles as the two previous points.
\end{enumerate}
\end{proof}
}

\begin{definition}
    Meaningful approximation is extended to list contexts $\cbnLCtxt$
    inductively as follows:
    \begin{equation*}
        \cbnTApprox{\Hole}
            \coloneqq \Hole
    \hspace{2cm}
        \cbvTApprox{\cbnLCtxt\esub{x}{s}}
            \coloneqq
                \cbnTApprox{\cbnLCtxt}\esub{x}{\cbnTApprox{s}}
    \end{equation*}
\end{definition}

\begin{corollary} 
    \label{lem:cbn_MA(L<t>)_=_MA(L)<MA(t)>}%
    Let $\cbnLCtxt<t>$ be \CBNSymb-meaningful, then
    $\cbnTApprox{\cbnLCtxt<t>} =
    \cbnTApprox{\cbnLCtxt}\cbnCtxtPlug{\cbnTApprox{t}}$.
\end{corollary}
\stableProof{
    \begin{proof}
By induction on $\cbnLCtxt$. Cases:
\begin{itemize}
\item[\bltI] $\cbnLCtxt = \Hole$: Then $\cbnTApprox{\cbnLCtxt<t>} =
    \cbnTApprox{t} = \Hole\cbnCtxtPlug{\cbnTApprox{t}} =
    \cbnTApprox{\Hole}\cbnCtxtPlug{\cbnTApprox{t}}$.

\item[\bltI] $\cbnLCtxt = \cbnLCtxt'\esub{x}{u}$: Then $\cbnLCtxt<t> =
    \cbnLCtxt'<t>\esub{x}{u}$ is \CBNSymb-meaningful, hence by
    \Cref{lem:cbn_solvability_on_app_and_es}.(\ref{lem:cbn_meaningful_on_esub})
    $\cbnLCtxt'<t>$ is \CBNSymb-meaningful. By \ih on $\cbnLCtxt'$,
    one obtains that $\cbnTApprox{\cbnLCtxt'<t>} =
    \cbnTApprox{\cbnLCtxt'}\cbnCtxtPlug{\cbnTApprox{t}}$, thus:
    \begin{equation*}
        \begin{array}{r cl cl}
            \cbnTApprox{\cbnLCtxt<t>}
            &=& \cbnTApprox{\cbnLCtxt'<t>\esub{x}{u}}
        \\
            &=& \cbnTApprox{\cbnLCtxt'<t>}\esub{x}{\cbnTApprox{u}}
        \\
            &\eqih& \cbnTApprox{\cbnLCtxt'}\cbnCtxtPlug{\cbnTApprox{t}}\esub{x}{\cbnTApprox{u}}
        \\
            &=& \big(\cbnTApprox{\cbnLCtxt'}\esub{x}{\cbnTApprox{u}}\big)\cbnCtxtPlug{\cbnTApprox{t}}
        \\
            &=& \cbnTApprox{\cbnLCtxt'\esub{x}{u}}\cbnCtxtPlug{\cbnTApprox{t}}
            &=& \cbnTApprox{\cbnLCtxt}\cbnCtxtPlug{\cbnTApprox{t}}
        \end{array}
    \end{equation*}
\end{itemize}
\end{proof}

}

\begin{lemma} 
    \label{lem:cbn_MA(t{x:=u})_<=_MA(t){x:=MA(u)}}%
    Let $t, u \in \setCbvTerms$, then $\cbnTApprox{t\isub{x}{u}}
    \cbnAprxLeq \cbnTApprox{t}\isub{x}{\cbnTApprox{u}}$.
\end{lemma}
\stableProof{
    \begin{proof}
We distinguish two cases:
\begin{itemize}
\item[\bltI] $t\isub{x}{u}$ is \CBNSymb-meaningless. Then
    $\cbnTApprox{t\isub{x}{u}} = \cbnAprxBot \cbnAprxLeq
    \cbnTApprox{t}\isub{x}{\cbnTApprox{u}}$.

\item[\bltI] $t\isub{x}{u}$ is \CBNSymb-meaningful.
    By~\Cref{lem:cbn_solvability_on_app_and_es}.(\ref{lem:cbn_meaningful_on_isub}),
    $t$ is \CBNSymb-meaningful. We reason by induction on $t$. Cases:
    \begin{itemize}
        \item[\bltII] $t = y$: Then $\cbnTApprox{t} = y$ and we
        distinguish two subcases:
        \begin{itemize}
        \item[\bltIII] $y = x$: Then $t\isub{x}{u} = u$ and so
            $\cbnTApprox{t\isub{x}{u}} = \cbnTApprox{u} =
            x\isub{x}{\cbnTApprox{u}} =
            \cbnTApprox{t}\isub{x}{\cbnTApprox{u}}$.

        \item[\bltIII] $y \neq x$: Then $t\isub{x}{u} = y$ and so
            $\cbnTApprox{t\isub{x}{u}} = \cbnTApprox{y} = y =
            y\isub{x}{\cbnTApprox{u}} =
            \cbnTApprox{t}\isub{x}{\cbnTApprox{u}}$.
        \end{itemize}

    \item[\bltII] $t = \abs{y}{t'}$: Then $\cbnTApprox{t} =
        \abs{y}{\cbnTApprox{t'}}$ and $t\isub{x}{u} =
        \abs{y}(t'\isub{x}{u})$ (we suppose without loss of generality
        that $y \notin \freeVar{u} \cup \{x\}$). By \ih on $t'$, one
        has $\cbnTApprox{t'\isub{x}{u}} \cbnAprxLeq
        \cbnTApprox{t'}\isub{x}{\cbnTApprox{u}}$. Thus:
        \begin{equation*}
            \begin{array}{r cl cl}
                \cbnTApprox{t\isub{x}{u}}
                &=& \cbnTApprox{\abs{y}{(t'\isub{x}{u})}}
                \\
                &=& \abs{y}{\cbnTApprox{t'\isub{x}{u}}}
                \\
                &\cbnAprxLeqIh& \abs{y}{\cbnTApprox{t'}\isub{x}{\cbnTApprox{u}}}
                \\
                &=& (\abs{y}{\cbnTApprox{t'}})\isub{x}{\cbnTApprox{u}}
                \\
                &=& \cbnTApprox{\abs{y}{t'}}\isub{x}{\cbnTApprox{u}}
                &=& \cbnTApprox{t}\isub{x}{\cbnTApprox{u}}
            \end{array}
        \end{equation*}

    \item[\bltII] $t = \app{t_1}{t_2}$: Then  $\cbnTApprox{t} =
        \app{\cbnTApprox{t_1}}{\cbnTApprox{t_2}}$ and $t\isub{x}{u} =
        (t_1\isub{x}{u})\esub{y}{t_2\isub{x}{u}}$ (we suppose without
        loss of generality that $y \notin \freeVar{u} \cup \{x\}$). By
        \ih on $t_1$ and $t_2$, one has $\cbnTApprox{t_1\isub{x}{u}}
        \cbnAprxLeq \cbnTApprox{t_1}\isub{x}{\cbnTApprox{u}}$ and
        $\cbnTApprox{t_2\isub{x}{u}} \cbnAprxLeq
        \cbnTApprox{t_2}\isub{x}{\cbnTApprox{u}}$. Thus:
        \begin{equation*}
            \begin{array}{r cl cl}
                \cbnTApprox{t\isub{x}{u}}
                &=& \cbnTApprox{\app{t_1\isub{x}{u}}{t_2\isub{x}{u}}}
                \\
                &=& \app{\cbnTApprox{t_1\isub{x}{u}}}{\cbnTApprox{t_2\isub{x}{u}}}
                \\
                &\cbnAprxLeqIh& \app{(\cbnTApprox{t_1}\isub{x}{\cbnTApprox{u}})}{(\cbnTApprox{t_2}\isub{x}{\cbnTApprox{u}})}
                \\
                &=& (\app{\cbnTApprox{t_1}}{\cbnTApprox{t_2}})\isub{x}{\cbnTApprox{u}}
                \\
                &=& \cbnTApprox{\app{t_1}{t_2}}\isub{x}{\cbnTApprox{u}}
                &=& \cbnTApprox{t}\isub{x}{\cbnTApprox{u}}
            \end{array}
        \end{equation*}

    \item[\bltII] $t = t_1\esub{y}{t_2}$: Thus $\cbnTApprox{t} =
        \cbnTApprox{t_1}\esub{y}{\cbnTApprox{t_2}}$ and $t\isub{x}{u}
        = (t_1\isub{x}{u})\esub{y}{t_2\isub{x}{u}}$ (we suppose
        without loss of generality that $y \notin \freeVar{u} \cup
        \{x\}$). By \ih on $t_1$ and $t_2$, one has
        $\cbnTApprox{t_1\isub{x}{u}} \cbnAprxLeq
        \cbnTApprox{t_1}\isub{x}{\cbnTApprox{u}}$ and
        $\cbnTApprox{t_2\isub{x}{u}} \cbnAprxLeq
        \cbnTApprox{t_2}\isub{x}{\cbnTApprox{u}}$. Thus:
        \begin{equation*}
            \begin{array}{r cl cl}
                \cbnTApprox{t\isub{x}{u}}
                &=& \cbnTApprox{t_1\isub{x}{u}\esub{y}{t_2\isub{x}{u}}}
                \\
                &=& \cbnTApprox{t_1\isub{x}{u}}\esub{y}{\cbnTApprox{t_2\isub{x}{u}}}
                \\
                &\cbnAprxLeqIh& (\cbnTApprox{t_1}\isub{x}{\cbnTApprox{u}})\esub{y}{\cbnTApprox{t_2}\isub{x}{\cbnTApprox{u}}}
                \\
                &=& (\cbnTApprox{t_1}\esub{y}{\cbnTApprox{t_2}})\isub{x}{\cbnTApprox{u}}
                \\
                &=& \cbnTApprox{t_1\esub{y}{t_2}}\isub{x}{\cbnTApprox{u}}
                &=& \cbnTApprox{t}\isub{x}{\cbnTApprox{u}}
            \end{array}
        \end{equation*}
        \qedhere
    \end{itemize}
\end{itemize}
\end{proof}
}

\subsection{Proofs of Section \ref{sec:genericity-cbn}}

\CbnMeaningfulApproximationLeqTerm*%
\label{prf:Cbn_MA(t)_<=_t}%
\stableProof{
    \begin{proof}
We distinguish two cases:
\begin{itemize}
\item[\bltI] $t$ is meaningless. Then $\cbnTApprox{t} = \cbnAprxBot
    \cbnAprxLeq t$.
\item[\bltI] $t$ is meaningful. We then reason by cases on $t$:
    \begin{itemize}
    \item[\bltII] $t = x$: Then $\cbnTApprox{t} = x \cbnAprxLeq x =
        t$.

    \item[\bltII] $t = \abs{x}{t'}$: By \ih on $t'$, one has that
        $\cbnTApprox{t'} \cbnAprxLeq t'$ thus $\cbnTApprox{t} =
        \abs{x}{\cbnTApprox{t'}} \cbnAprxLeq \abs{x}{t'} = t$.

    \item[\bltII] $t = \app{t_1}{t_2}$: By \ih on $t_1$ and $t_2$,
        one has that $\cbnTApprox{t_1} \cbnAprxLeq t_1$ and
        $\cbnTApprox{t_2} \cbnAprxLeq t_2$ so that $\cbnTApprox{t} =
        \app{\cbnTApprox{t_1}}{\cbnTApprox{t_2}} \cbnAprxLeq
        \app{t_1}{t_2} = t$.

    \item[\bltII] $t = t_1\esub{x}{t_2}$: By \ih on $t_1$ and
        $t_2$, $\cbnTApprox{t_1} \cbnAprxLeq t_1$ and
        $\cbnTApprox{t_2} \cbnAprxLeq t_2$. So, $\cbnTApprox{t} =
        \cbnTApprox{t_1}\esub{x}{\cbnTApprox{t_2}} \cbnAprxLeq
        t_1\esub{x}{t_2} = t$.
    \qedhere
    \end{itemize}
\end{itemize} 
\end{proof}

}

\CbnApproximantGenericity*
\label{prf:cbn_Approximant_Genericity}%
\stableProof{
    \begin{proof}%
Let $\cbnCCtxt$ be a full context and $t, u \in \setCbvTerms$ be such
that $t$ is meaningless. We distinguish two cases:
\begin{itemize}
\item[\bltI] $\cbnCCtxt<t>$ is meaningless.  Then
    $\cbnTApprox{\cbnCCtxt<t>} = \cbnAprxBot \cbnAprxLeq
    \cbnTApprox{\cbnCCtxt<u>}$.

\item[\bltI] $\cbnCCtxt<t>$ is meaningful: Then by surface genericity
    (\Cref{lem:Cbn_Qualitative_Surface_Genericity}), one has that
    $\cbnCCtxt<u>$ is also meaningful. We reason  by induction on
    $\cbnCCtxt$:
    \begin{itemize}
    \item[\bltII] $\cbnCCtxt = \Hole$: This case is impossible because
        it would mean $\cbnCCtxt<t> = t$ is meaningless (by
        hypothesis) and meaningful at the same time.

    \item[\bltII] $\cbnCCtxt = \abs{x}{\cbnCCtxt'}$: By \ih on
        $\cbnCCtxt'$, one has $\cbnTApprox{\cbnCCtxt'<t>} \cbnAprxLeq
        \cbnTApprox{\cbnCCtxt'<u>}$, thus $\cbnTApprox{\cbnCCtxt<t>} =
        \abs{x}{\cbnTApprox{\cbnCCtxt'<t>}} \cbnAprxLeqIh
        \abs{x}{\cbnTApprox{\cbnCCtxt'<u>}} =
        \cbnTApprox{\cbnCCtxt<u'>}$ by contextual closure.

    \item[\bltII] $\cbnCCtxt = \app{\cbnCCtxt'}{s}$: By \ih on
        $\cbnCCtxt'$, one has $\cbnTApprox{\cbnCCtxt'<t>} \cbnAprxLeq
        \cbnTApprox{\cbnCCtxt'<u>}$, thus $\cbnTApprox{\cbnCCtxt<t>} =
        \app{\cbnTApprox{\cbnCCtxt'<t>}}{\cbnTApprox{s}} \cbnAprxLeqIh
        \app{\cbnTApprox{\cbnCCtxt'<u>}}{\cbnTApprox{s}} \allowbreak=
        \cbnTApprox{\cbnCCtxt<u'>}$ by contextual closure.

    \item[\bltII] $\cbnCCtxt = \app[\,]{s}{\cbnCCtxt'}$: By \ih on
        $\cbnCCtxt'$, one has $\cbnTApprox{\cbnCCtxt'<t>} \cbnAprxLeq
        \cbnTApprox{\cbnCCtxt'<u>}$, thus $\cbnTApprox{\cbnCCtxt<t>} =
        \app{\cbnTApprox{s}}{\cbnTApprox{\cbnCCtxt'<t>}} \cbnAprxLeqIh
        \app{\cbnTApprox{s}}{\cbnTApprox{\cbnCCtxt'<u>}} \allowbreak=
        \cbnTApprox{\cbnCCtxt<u'>}$ by contextual closure.

    \item[\bltII] $\cbnCCtxt = \cbnCCtxt'\esub{x}{s}$: By \ih on
        $\cbnCCtxt'$, one has $\cbnTApprox{\cbnCCtxt'<t>} \cbnAprxLeq
        \cbnTApprox{\cbnCCtxt'<u>}$, and hence
        $\cbnTApprox{\cbnCCtxt<t>} =
        \cbnTApprox{\cbnCCtxt'<t>}\esub{x}{}{\cbnTApprox{s}}
        \cbnAprxLeqIh
        \cbnTApprox{\cbnCCtxt'<u>}\esub{x}{}{\cbnTApprox{s}} =
        \cbnTApprox{\cbnCCtxt<u'>}$ by contextual closure.

    \item[\bltII] $\cbnCCtxt = s\esub{x}{\cbnCCtxt'}$: By \ih on
        $\cbnCCtxt'$, one has $\cbnTApprox{\cbnCCtxt'<t>} \cbnAprxLeq
        \cbnTApprox{\cbnCCtxt'<u>}$, and hence
        $\cbnTApprox{\cbnCCtxt<t>} =
        \cbnTApprox{s}\esub{x}{\cbnTApprox{\cbnCCtxt'<t>}}
        \cbnAprxLeqIh
        \cbnTApprox{s}\esub{x}{\cbnTApprox{\cbnCCtxt'<u>}} =
        \cbnTApprox{\cbnCCtxt<u'>}$ by contextual closure.
        \qedhere
    \end{itemize}
\end{itemize}
\end{proof}

}

\paragraph*{\textbf{Dynamic Approximation and Lifting.}}

We now prove the instance of
\Cref{ax:Gen_t_->Sn_u_=>_OA(t)_->*Sn_|u_OA(u)} for \CBNSymb.

\begin{lemma}[\CBNSymb Dynamic Approximation]
    \label{lem:cbn_t_->Sn_u_=>_MA(t)_->*Sn_MA(u)}%
    Let $\indexK \in \setOrdinals$ and $t, u \in \setCbnTerms$ such
    that $t \cbnArr_{S_\indexK} u$, then there exists $\aprxU \in
    \setCbnAprxTerms$ such that $\cbnTApprox{t}
    \cbnAprxArr^=_{S_\indexK} \aprxU \cbnAprxGeq \cbnTApprox{u}$.
\end{lemma}
\stableProof{
    \begin{proof}
Let $t, u \in \setCbnTerms$ such that $t \cbnArr_{S_\indexK} u$. We
distinguish two cases:
\begin{itemize}
\item[\bltI] $t$ \CBNSymb-meaningless: Then $u$ is \CBNSymb-meaningful
    using~\Cref{lem:cbn_t_meaningful_and_t_->Sn_u_=>_u_meaningful}
    thus $\cbnTApprox{t} = \cbnAprxBot$ and $\cbnTApprox{u} =
    \cbnAprxBot$. We set $\aprxU \coloneqq \cbnTApprox{u}$ therefore
    by reflexivity of both relations, $\cbnTApprox{t}
    \cbnAprxArr^=_{S_\indexK} \aprxU \cbnAprxGeq \cbnTApprox{u}$.

\item[\bltI] $t$ \CBNSymb-meaningful: Then $u$ is \CBNSymb-meaningful
    using~\Cref{lem:cbn_t_meaningful_and_t_->Sn_u_=>_u_meaningful}. By
    definition,  there exist a context $\cbnSCtxt_\indexK$ and  $t',
    u' \in \setCbnTerms$ such that $t = \cbnSCtxt_\indexK<t'>$, $u =
    \cbnSCtxt_\indexK<u'>$ and $t' \mapstoR u'$ for some $\rel \in
    \{\cbnSymbBeta, \cbnSymbSubs\}$. We reason by induction on
    $\cbnSCtxt_\indexK$:
    \begin{itemize}
    \item[\bltII] $\cbnSCtxt_\indexK = \Hole$: Then $t = t'$ and $u =
        u'$. We distinguish two cases:
        \begin{itemize}
        \item[\bltIII] $\rel = \cbnSymbBeta$: Then $t' =
            \app{\cbnLCtxt<\abs{x}{s_1}>}{s_2}$ and $u' =
            \cbnLCtxt<s_1\esub{x}{s_2}>$ for some list context
            $\cbnLCtxt$ and some $s_1, s_2 \in \setCbnTerms$. Since
            $t'$ is \CBNSymb-meaningful, then using
            \Cref{lem:cbn_solvability_on_app_and_es}.(\ref{lem:cbn_meaningful_on_app}),
            one deduces in particular that $\cbnLCtxt<\abs{x}{s_1}>$
            is \CBNSymb-meaningful. Since $u'$ is \CBNSymb-meaningful,
            then using
            \Cref{lem:cbn_solvability_on_app_and_es}.(\ref{lem:cbn_meaningful_on_esub}),
            one deduces in particular that $s_1\esub{x}{s_2}$ is
            \CBNSymb-meaningful. Thus, by
            \Cref{lem:cbn_MA(L<t>)_=_MA(L)<MA(t)>}, one has
            $\cbnTApprox{t'} =
            \app{\cbnTApprox{\cbnLCtxt<\abs{x}{s_1}>}}{\cbnTApprox{s_2}}
            =
            \app{\cbnTApprox{\cbnLCtxt}\cbnCtxtPlug{\abs{x}{\cbnTApprox{s_1}}}}{\cbnTApprox{s_2}}$
            and $\cbnTApprox{u'} =
            \cbnTApprox{\cbnLCtxt<s_1\esub{x}{s_2}>} =
            \cbnTApprox{\cbnLCtxt}\cbnCtxtPlug{
            \cbnTApprox{s_1\esub{x}{s_2}}} =
            \cbnTApprox{\cbnLCtxt}\cbnCtxtPlug{\cbnTApprox{s_1}\esub{x}{\cbnTApprox{s_2}}}$.
            Hence, $\cbnTApprox{t} = \cbnTApprox{t'}
            \mapstoR[\cbnSymbBeta] \cbnTApprox{u'} = \cbnTApprox{u}$.
            We set $\aprxU \coloneqq \cbnTApprox{u}$, therefore by
            reflexivity $\cbnTApprox{t} \cbnAprxArr_{S_\indexK} \aprxU
            \cbnAprxGeq \cbnTApprox{u}$.

        \item[\bltIII] $\rel = \cbnSymbSubs$: Then $t' =
            s_1\esub{x}{s_2}$ and $u' = s_1\isub{x}{s_2}$ for some
            $s_1, s_2 \in \setCbnTerms$. Using
            \Cref{lem:cbn_MA(t{x:=u})_<=_MA(t){x:=MA(u)}}, one has
            that $\cbnTApprox{t'} =
            \cbnTApprox{s_1}\esub{x}{\cbnTApprox{s_2}}$ and
            $\cbnTApprox{u'} \cbnAprxLeq
            \cbnTApprox{s_1}\isub{x}{\cbnTApprox{s_2}}$. Thus
            $\cbnTApprox{t} = \cbnTApprox{t'} \mapstoR[\cbnSymbSubs]
            \cbnTApprox{s_1}\isub{x}{\cbnTApprox{s_2}} \cbnAprxGeq
            \cbnTApprox{u'} = \cbnTApprox{u}$. We set $\aprxU
            \coloneqq \cbnTApprox{s_1}\isub{x}{\cbnTApprox{s_2}}$,
            therefore $\cbnTApprox{t} \cbnAprxArr_{S_\indexK} \aprxU
            \cbnAprxGeq \cbnTApprox{u}$.
        \end{itemize}

    \item[\bltII] $\cbnSCtxt_\indexK = \abs{x}{\cbnSCtxt'_\indexK}$:
        By \ih on $\cbnSCtxt'_\indexK$, one has $\aprxU'  \in
        \setCbnAprxTerms$ such that
        $\cbnTApprox{\cbnSCtxt'_\indexK<t'>} \cbnAprxArr^=_{S_\indexK}
        \aprxU' \cbnAprxGeq \cbnTApprox{\cbnSCtxt'_\indexK<u'>}$. We
        set $\aprxU \coloneqq \abs{x}{\aprxU'}$, therefore
        $\cbnTApprox{\cbnSCtxt_\indexK<t'>} =
        \abs{x}{\cbnTApprox{\cbnSCtxt'_\indexK<t'>}}
        \cbnAprxArr^=_{S_\indexK} \aprxU \cbnAprxGeq
        \abs{x}{\cbnTApprox{\cbnSCtxt'_\indexK<u'>}} =
        \cbnTApprox{\cbnSCtxt_\indexK<u'>}$.

    \item[\bltII] $\cbnSCtxt_\indexK =
        \app[\,]{\cbnSCtxt'_\indexK}{s}$: By \ih on
        $\cbnSCtxt'_\indexK$, one has $\aprxU'  \in \setCbnAprxTerms$
        such that $\cbnTApprox{\cbnSCtxt'_\indexK<t'>}
        \cbnAprxArr^=_{S_\indexK} \aprxU' \cbnAprxGeq
        \cbnTApprox{\cbnSCtxt'_\indexK<u'>}$. We set $\aprxU \coloneqq
        \app{\aprxU'}{\cbnTApprox{s}}$, therefore
        $\cbnTApprox{\cbnSCtxt_\indexK<t'>} =
        \app{\cbnTApprox{\cbnSCtxt'_\indexK<t'>}}{\cbnTApprox{s}}
        \cbnAprxArr^=_{S_\indexK}  \aprxU \cbnAprxGeq
        \app{\cbnTApprox{\cbnSCtxt'_\indexK<u'>}}{\cbnTApprox{s}} =
        \cbnTApprox{\cbnSCtxt_\indexK<u'>}$.

    \item[\bltII] $\cbnSCtxt_\indexK = \app[\,]{s}{\cbnSCtxt_m}$ with
        either $\indexK > 0$ and $m = \indexK - 1$, or $\indexK = m =
        \omega$: By \ih on $\cbnSCtxt'_m$, one has $\aprxU'  \in
        \setCbnAprxTerms$ such that $\cbnTApprox{\cbnSCtxt_m<t'>}
        \cbnAprxArr^=_{S_m} \aprxU' \cbnAprxGeq
        \cbnTApprox{\cbnSCtxt_m<u'>}$. We set $\aprxU \coloneqq
        \app{\cbnTApprox{s}}{\aprxU'}$, therefore
        $\cbnTApprox{\cbnSCtxt_\indexK<t'>} =
        \app{\cbnTApprox{s}}{\cbnTApprox{\cbnSCtxt_m<t'>}}
        \cbnAprxArr^=_{S_\indexK}  \aprxU \cbnAprxGeq
        \app{\cbnTApprox{s}}{\cbnTApprox{\cbnSCtxt_m<u'>}} =
        \cbnTApprox{\cbnSCtxt_\indexK<u'>}$.

    \item[\bltII] $\cbnSCtxt_\indexK = \cbnSCtxt'_\indexK\esub{x}{s}$:
        By \ih on $\cbnSCtxt'_\indexK$, one has $\aprxU'  \in
        \setCbnAprxTerms$ such that
        $\cbnTApprox{\cbnSCtxt'_\indexK<t'>} \cbnAprxArr^=_{S_\indexK}
        \aprxU' \cbnAprxGeq \cbnTApprox{\cbnSCtxt'_\indexK<u'>}$. We
        set $\aprxU \coloneqq \aprxU'\esub{x}{\cbnTApprox{s}}$,
        therefore $\cbnTApprox{\cbnSCtxt_\indexK<t'>} =
        \aprxU'\esub{x}{\cbnTApprox{s}} \cbnAprxArr^=_{S_\indexK}
        \aprxU \cbnAprxGeq
        \cbnTApprox{\cbnSCtxt'_\indexK<u'>}\esub{x}{\cbnTApprox{s}} =
        \cbnTApprox{\cbnSCtxt_\indexK<u'>}$.

    \item[\bltII] $\cbnSCtxt_\indexK = s\esub{x}{\cbnSCtxt_m}$ with
        either $\indexK > 0$ and $m = \indexK - 1$, or $\indexK = m =
        \omega$: By \ih on $\cbnSCtxt_m$, one has $\aprxU' \in
        \setCbnAprxTerms$ such that $\cbnTApprox{\cbnSCtxt_m<t'>}
        \cbnAprxArr^=_{S_m} \aprxU' \cbnAprxGeq
        \cbnTApprox{\cbnSCtxt_m<u'>}$. We set $\aprxU \coloneqq
        \cbnTApprox{s}\esub{x}{\aprxU'}$, therefore
        $\cbnTApprox{\cbnSCtxt_\indexK<t'>} =
        \cbnTApprox{s}\esub{x}{\cbnTApprox{\cbnSCtxt_m<t'>}}
        \cbnAprxArr^=_{S_\indexK}  \aprxU \cbnAprxGeq
        \cbnTApprox{s}\esub{x}{\cbnTApprox{\cbnSCtxt_m<u'>}} =
        \cbnTApprox{\cbnSCtxt_\indexK<u'>}$.
        \qedhere
    \end{itemize}
\end{itemize}
\end{proof}

}

We now prove the instance of
\Cref{ax:Gen_t_->Sn_u_and_t_<=_t'_==>_t'_->Sn_u'_and_u_<=_u'} for
\CBNSymb.

\begin{lemma}[\CBNSymb Dynamic \tPartial^ Lift]{lemma}
    \label{lem:cbn_t_->Sn_u_and_t_<=_t'_==>_t'_->Sn_u'_and_u_<=_u'}%
    Let $\indexK \in \setOrdinals$ and $\aprxT, \aprxU, \aprxT' \in
    \setCbnAprxTerms$. If $\aprxT \cbnAprxArr_{S_\indexK} \aprxU$ and
    $\aprxT \cbnAprxLeq \aprxT'$, then there is $\aprxU' \in
    \setCbnAprxTerms$ such that $\aprxT' \cbnAprxArr_{S_\indexK}
    \aprxU'$ and $\aprxU \cbnAprxLeq \aprxU'$.%
\end{lemma}
\stableProof{
    \begin{proof}
Let $\aprxT, \aprxU, \aprxT' \in \setCbnAprxTerms$ such that $\aprxT
\cbnAprxLeq \aprxT'$ and $\aprxT \cbnAprxArr_{S_\indexK} \aprxU$. Then
$\aprxT = \cbnAprxSCtxt_\indexK\cbnCtxtPlug{\aprxT_r}$ and $\aprxU =
\cbnAprxSCtxt_\indexK\cbnCtxtPlug{\aprxU_r}$ with $\aprxT_r \mapstoR
\aprxU_r$ for some $\aprxT_r, \aprxU_r \in \setCbnAprxTerms$ and $\rel
\in \{\cbnSymbBeta, \cbnSymbSubs\}$. Let us proceed by induction on
$\cbnAprxSCtxt_\indexK$:
\begin{itemize}
\item[\bltI] $\cbnAprxSCtxt_\indexK = \Hole$: Then $\aprxT = \aprxT_r$
    and $\aprxU = \aprxU_r$. We distinguish two cases:
    \begin{itemize}
    \item[\bltII] $\rel = \cbnSymbBeta$: Then $\aprxT_r =
        \app{\cbnAprxLCtxt<\abs{x}{\aprxS_1}>}{\aprxS_2}$ and
        $\aprxU_r = \cbnAprxLCtxt<\aprxS_1\esub{x}{\aprxS_2}>$ for
        some list context $\cbnAprxLCtxt$ and some  $\aprxS_1,
        \aprxS_2 \in \setCbnAprxTerms$. Since $\aprxT \cbnAprxLeq
        \aprxT'$, by induction on $\cbnAprxLCtxt$, one has that
        $\aprxT = \app{\cbnAprxLCtxt'<\abs{x}{\aprxS'_1}>}{\aprxS'_2}$
        with $\cbnAprxLCtxt \cbnAprxLeq \cbnAprxLCtxt'$, $\aprxS_1
        \cbnAprxLeq \aprxS'_1$ and $\aprxS_2 \cbnAprxLeq \aprxS'_2$
        for some list context $\cbnAprxLCtxt'$ and some $\aprxS'_1,
        \aprxS'_2 \in \setCbnAprxTerms$. Taking $\aprxU' =
        \cbnAprxLCtxt'<\aprxS'_1\esub{x}{\aprxS'_2}>$ concludes this
        case since $\aprxT' =
        \app{\cbnAprxLCtxt'<\abs{x}{\aprxS'_1}>}{\aprxS'_2}
        \cbnAprxArr_{S_\indexK}
        \cbnAprxLCtxt'<\aprxS'_1\esub{x}{\aprxS'_2}> = \aprxU'$ and by
        transitivity $\aprxS_1\esub{x}{\aprxS_2} \cbnAprxLeq
        \aprxS'_1\esub{x}{\aprxS'_2}$ thus
        $\cbnAprxLCtxt<\aprxS_1\esub{x}{\aprxS_2}> \cbnAprxLeq
        \cbnAprxLCtxt'<\aprxS'_1\esub{x}{\aprxS'_2}> = \aprxU'$ using
        \cref{lem:cbn_L1_<=_L2_and_t1_<=_t2_==>_L1<t1>_<=_L2<t2>}.

    \item[\bltII] $\rel = \cbnSymbSubs$: Then $\aprxT_r =
        \aprxS_1\esub{x}{\aprxS_2}$ and $\aprxU_r =
        \aprxS_1\isub{x}{\aprxS_2}$ for some partial terms $\aprxS_1,
        \aprxS_2 \in \setCbnAprxTerms$. Since $\aprxT \cbnAprxLeq
        \aprxT'$, one has that $\aprxT' =
        \aprxS'_1\esub{x}{\aprxS'_2}$ with $\aprxS_1 \cbnAprxLeq
        \aprxS'_1$ and $\aprxS_2 \cbnAprxLeq \aprxS'_2$ for some some
        partial terms $\aprxS'_1, \aprxS'_2 \in \setCbnAprxTerms$.
        Taking $\aprxU' = \aprxS'_1\isub{x}{\aprxS'_2}$ concludes this
        case since $\aprxT' = \aprxS'_1\esub{x}{\aprxS'_2}
        \cbnAprxArr_{S_\indexK} \aprxS'_1\isub{x}{\aprxS'_2}$ and
        using
        \cref{lem:Cbn_t1_<=_t2_and_u1_<=_u2_==>_t1{x:=u1}_<=_t2{x:=u2}},
        one has that $\aprxS_1\isub{x}{\aprxS_2} \cbnAprxLeq
        \aprxS'_1\isub{x}{\aprxS'_2}$.
    \end{itemize}

\item[\bltI] $\cbnAprxSCtxt_\indexK =
    \abs{x}{\cbnAprxSCtxt'_\indexK}$: Then $\aprxT =
    \abs{x}{\cbnAprxSCtxt'_\indexK<\aprxT_r>}$ and $\aprxU =
    \abs{x}{\cbnAprxSCtxt'_\indexK<\aprxU_r>}$ thus $\aprxT' =
    \abs{x}{\aprxT'_1}$ for some $\aprxT'_1 \in \setCbnAprxTerms$ such
    that $\cbnAprxSCtxt'_\indexK<\aprxT_r> \cbnAprxLeq \aprxT'_1$.
    Since $\cbnAprxSCtxt'_\indexK<\aprxT_r> \cbnAprxArr_{S_\indexK}
    \cbnAprxSCtxt'_\indexK<\aprxU_r>$, by \ih on
    $\cbnAprxSCtxt'_\indexK$, one obtains $\aprxU'_1 \in
    \setCbnAprxTerms$ such that $\aprxT'_1 \cbnAprxArr_{S_\indexK}
    \aprxU'_1$ and $\cbnAprxSCtxt'_\indexK<\aprxU_r> \cbnAprxLeq
    \aprxU'_1$. Taking $\aprxU' = \abs{x}{\aprxU'_1}$ concludes this
    case since $\aprxT' = \abs{x}{\aprxT'_1} \cbnAprxArr_{S_\indexK}
    \abs{x}{\aprxU'_1} = \aprxU'$ and $\aprxU =
    \abs{x}{\cbnAprxSCtxt'_\indexK<\aprxU_r>} \cbnAprxLeq
    \abs{x}{\aprxU'_1} = \aprxU'$.

\item[\bltI] $\cbnAprxSCtxt_\indexK =
    \app[\,]{\cbnAprxSCtxt'_\indexK}{\aprxS}$: Then $\aprxT =
    \app{\cbnAprxSCtxt'_\indexK<\aprxT_r>}{\aprxS}$ and $\aprxU =
    \app{\cbnAprxSCtxt'_\indexK<\aprxU_r>}{\aprxS}$ thus $\aprxT' =
    \app[\,]{\aprxT'_1}{\aprxS'}$ for some partial term $\aprxT'_1,
    \aprxS \in \setCbnAprxTerms$ such that
    $\cbnAprxSCtxt'_\indexK<\aprxT_r> \cbnAprxLeq \aprxT'_1$ and
    $\aprxS \cbnAprxLeq \aprxS'$. Since
    $\cbnAprxSCtxt'_\indexK<\aprxT_r> \cbnAprxArr_{S_\indexK}
    \cbnAprxSCtxt'_\indexK<\aprxU_r>$, by \ih on
    $\cbnAprxSCtxt'_\indexK$, one obtains $\aprxU'_1 \in
    \setCbnAprxTerms$ such that $\aprxT'_1 \cbnAprxArr_{S_\indexK}
    \aprxU'_1$ and $\cbnAprxSCtxt'_\indexK<\aprxU_r> \cbnAprxLeq
    \aprxU'_1$. Taking $\aprxU' = \app{\aprxU'_1}{\aprxS'}$ concludes
    this case since $\aprxT' = \app{\aprxT'_1}{\aprxS'}
    \cbnAprxArr_{S_\indexK} \app{\aprxU'_1}{\aprxS'} = \aprxU'$ and
    $\aprxU = \app{\cbnAprxSCtxt'_\indexK<\aprxU_r>}{\aprxS}
    \cbnAprxLeq \app{\aprxU'_1}{\aprxS'} = \aprxU'$ by transitivity.

\item[\bltI] $\cbnAprxSCtxt_\indexK =
    \app[\,]{\aprxS}{\cbnAprxSCtxt'_m}$ with either $\indexK > 0$ and
    $m = \indexK - 1$, or $\indexK = m = \indexOmega$: Then $\aprxT =
    \app[\,]{\aprxS}{\cbnAprxSCtxt'_m<\aprxT_r>}$ and $\aprxU =
    \app[\,]{\aprxS}{\cbnAprxSCtxt'_m<\aprxU_r>}$ thus $\aprxT' =
    \app[\,]{\aprxS'}{\aprxT'_1}$ for some $\aprxS, \aprxT'_1 \in
    \setCbnAprxTerms$ such that $\aprxS \cbnAprxLeq \aprxS'$ and
    $\cbnAprxSCtxt'_m<\aprxT_r> \cbnAprxLeq \aprxT'_1$. Since
    $\cbnAprxSCtxt'_m<\aprxT_r> \cbnAprxArr_{S_{m}}
    \cbnAprxSCtxt'_m<\aprxU_r>$, by \ih on $\cbnAprxSCtxt'_m$, one
    obtains $\aprxU'_1 \in \setCbnAprxTerms$ such that $\aprxT'_1
    \cbnAprxArr_{S_m} \aprxU'_1$ and $\cbnAprxSCtxt'_m<\aprxU_r>
    \cbnAprxLeq \aprxU'_1$. Taking $\aprxU' =
    \app{\aprxS'}{\aprxU'_1}$ concludes this case since $\aprxT' =
    \app{\aprxS'}{\aprxT'_1} \cbnAprxArr_{S_\indexK}
    \app{\aprxS'}{\aprxU'_1} = \aprxU'$ and $\aprxU =
    \app[\,]{\aprxS}{\cbnAprxSCtxt'_\indexK<\aprxU_r>} \cbnAprxLeq
    \app{\aprxS'}{\aprxU'_1} = \aprxU'$ by transitivity.

\item[\bltI] $\cbnAprxSCtxt_\indexK =
    \cbnAprxSCtxt'_\indexK\esub{x}{\aprxS}$: Then $\aprxT =
    \cbnAprxSCtxt'_\indexK<\aprxT_r>\esub{x}{\aprxS}$ and $\aprxU =
    \cbnAprxSCtxt'_\indexK<\aprxU_r>\esub{x}{\aprxS}$ thus $\aprxT' =
    \aprxT'_1\esub{x}{\aprxS'}$ for some $\aprxT'_1, \aprxS \in
    \setCbnAprxTerms$ such that $\cbnAprxSCtxt'_\indexK<\aprxT_r>
    \cbnAprxLeq \aprxT'_1$ and $\aprxS \cbnAprxLeq \aprxS'$. Since
    $\cbnAprxSCtxt'_\indexK<\aprxT_r> \cbnAprxArr_{S_\indexK}
    \cbnAprxSCtxt'_\indexK<\aprxU_r>$, by \ih on
    $\cbnAprxSCtxt'_\indexK$, one obtains $\aprxU'_1 \in
    \setCbnAprxTerms$ such that $\aprxT'_1 \cbnAprxArr_{S_\indexK}
    \aprxU'_1$ and $\cbnAprxSCtxt'_\indexK<\aprxU_r> \cbnAprxLeq
    \aprxU'_1$. Taking $\aprxU' = \aprxU'_1\esub{x}{\aprxS'}$
    concludes this case since $\aprxT' = \aprxT'_1\esub{x}{\aprxS'}
    \cbnAprxArr_{S_\indexK} \aprxU'_1\esub{x}{\aprxS'} = \aprxU'$ and
    $\aprxU = \cbnAprxSCtxt'_\indexK<\aprxU_r>\esub{x}{\aprxS}
    \cbnAprxLeq \aprxU'_1\esub{x}{\aprxS'} = \aprxU'$ by transitivity.

\item[\bltI] $\cbnAprxSCtxt_\indexK =
    \aprxS\esub{x}{\cbnAprxSCtxt'_m}$ with either $\indexK > 0$ and $m
    = \indexK - 1$, or $\indexK = m = \indexOmega$: Then $\aprxT =
    \aprxS\esub{x}{\cbnAprxSCtxt'_m<\aprxT_r>}$ and $\aprxU =
    \aprxS\esub{x}{\cbnAprxSCtxt'_m<\aprxU_r>}$ thus $\aprxT' =
    \aprxS'\esub{x}{\aprxT'_1}$ for some $\aprxS, \aprxT'_1 \in
    \setCbnAprxTerms$ such that $\aprxS \cbnAprxLeq \aprxS'$ and
    $\cbnAprxSCtxt'_m<\aprxT_r> \cbnAprxLeq \aprxT'_1$. Since
    $\cbnAprxSCtxt'_m<\aprxT_r> \cbnAprxArr_{S_{m}}
    \cbnAprxSCtxt'_m<\aprxU_r>$, by \ih on $\cbnAprxSCtxt'_m$, one
    obtains $\aprxU'_1 \in \setCbnAprxTerms$ such that $\aprxT'_1
    \cbnAprxArr_{S_m} \aprxU'_1$ and $\cbnAprxSCtxt'_m<\aprxU_r>
    \cbnAprxLeq \aprxU'_1$. Taking $\aprxU' =
    \aprxS'\esub{x}{\aprxU'_1}$ concludes this case since $\aprxT' =
    \aprxS'\esub{x}{\aprxT'_1} \cbnAprxArr_{S_\indexK}
    \aprxS'\esub{x}{\aprxU'_1} = \aprxU'$ and $\aprxU =
    \aprxS\esub{x}{\cbnAprxSCtxt'_\indexK<\aprxU_r>} \cbnAprxLeq
    \aprxS'\esub{x}{\aprxU'_1} = \aprxU'$ by transitivity.
    \qedhere
\end{itemize}
\end{proof}

}

\paragraph*{\textbf{Observability.}} Let us recall the mutual
inductive definitions of $\cbnAprxBneS_\indexK$ and
$\cbnAprxBnoS_\indexK$:
\begin{equation*}
    \begin{array}{rcl}
        \cbnAprxBneS_0 &\coloneqq& x
            \vsep \app{\cbnAprxBneS_0}{\aprxT}
    \\
        \cbnAprxBneS_{\indexI+1} &\coloneqq& x
            \vsep \app{\cbnAprxBneS_{\indexI+1}}{\cbnAprxBnoS_\indexI}
                    \qquad (\indexI \in \setIntegers)
    \\
        \cbnAprxBneS_\indexOmega &\coloneqq& x
            \vsep \app{\cbnAprxBneS_\indexOmega}{\cbnAprxBnoS_\indexOmega}
    \\[0.2cm]
        \cbnAprxBnoS_\indexK &\coloneqq& \abs{x}{\cbnAprxBnoS_\indexK}
            \vsep \cbnAprxBneS_\indexK
            \qquad (\indexK \in \setOrdinals)
    \end{array}
\end{equation*}

We now prove the instance of
\Cref{ax:Gen_t_in_BnoSn_==>_OA(t)_in_noSn} for \CBNSymb.

\begin{lemma}[\CBNSymb Observability of Normal Form Approximants]
    \label{lem:cbn_t_in_BnoSn_==>_MA(t)_in_noSn}%
    Let $\indexK \in \setOrdinals$ and $\aprxT \in \cbnNoS_\indexK$,
    then $\cbnTApprox{t} \in \cbnAprxBnoS_\indexK$.%
\end{lemma}
\stableProof{
    \begin{proof}
We strengthen the induction hypothesis:
\begin{center}
    Let $t \in \setCbnTerms$ such that $\cbnPredSNF_\indexK{t}$, then
    $\cbnTApprox{t} \in \cbnAprxBnoS_\indexK$. And if, moreover,
    $\neg\cbnAbsPred{t}$ then $\cbnTApprox{t} \in
    \cbnAprxBneS_\indexK$.
\end{center}

Since $\cbnPredSNF_\indexK{t}$, then in particular $\cbnPredSNF{t}$,
so that $t$ is meaningful. We reason by induction on $t \in
\setCbnTerms$:
\begin{itemize}
\item[\bltI] $t = x$: Then $\neg\cbnAbsPred{t}$ and $\cbnTApprox{t} =
    x \in \cbnAprxBneS_\indexK$.

\item[\bltI] $t = \abs{x}{t'}$: Then $\cbnPredSNF_\indexK{t'}$ and
    $\cbnAbsPred{t}$, thus $\cbnTApprox{t'} \in \cbnAprxBnoS_\indexK$
    by \ih on $t'$. Hence, $\cbnTApprox{t} = \abs{y}{\cbnTApprox{t'}}
    \in \abs{y}{\cbnAprxBnoS_\indexK} \subseteq \cbnAprxBnoS_\indexK$.

\item[\bltI] $t = \app{t_1}{t_2}$: Then $\neg\cbnAbsPred{t}$,
    $\cbnPredSNF_\indexK{t_1}$ and $\neg\cbnAbsPred{t_1}$. By \ih on
    $t_1$, one has that $\cbnTApprox{t_1} \in \cbnAprxBneS_\indexK$.
    We distinguish two cases:
    \begin{itemize}
    \item $\indexK = 0$: Then $\cbnTApprox{t} =
        \app{\cbnTApprox{t_1}}{\cbnTApprox{t_2}} \in
        \app{\cbnAprxBneS_0}{\aprxT} \subseteq \cbnAprxBnoS_0$.

    \item $\indexK \in \Nat \setminus \{0\}$ (\resp $\indexK =
        \indexOmega$): Then $\cbnPredSNF_m{t_2}$ with $m = \indexK -
        1$ (\resp $m = \indexOmega$). By \ih on $t_2$, one has that
        $\cbnTApprox{t_2} \in \cbnAprxBnoS_m$ thus $\cbnTApprox{t} =
        \app{\cbnTApprox{t_1}}{\cbnTApprox{t_2}} \in
        \app{\cbnAprxBneS_\indexK}{\cbnAprxBnoS_m} \subseteq
        \cbnAprxBneS_\indexK$.
    \end{itemize}

\item[\bltI] $t = t_1\esub{x}{t_2}$: Impossible since
    $\cbnPredSNF_\indexK{t}$.
    \qedhere
\end{itemize}
\end{proof}

}

Let us recall the inductive definition of \CBNSymb stratified
equality:
\begin{equation*}
    \begin{array}{c}
        \begin{prooftree}
            \hypo{\phantom{\cbvEq[\indexI]}}
            \inferCbnEqVar[1]{x \cbvEq[\indexI] x}
        \end{prooftree}
    \hspace{1.5cm}
        \begin{prooftree}
            \hypo{t \cbvEq[\indexI] u }
            \inferCbnEqAbs{\abs{x}{t} \cbvEq[\indexI] \abs{x}{u}}
        \end{prooftree}
    \\[0.5cm]
        \begin{prooftree}
            \hypo{t_1 \cbvEq[0] u_1}
            \inferCbnEqApp[1]{\app{t_1}{t_2} \cbvEq[0] \app{u_1}{u_2}}
        \end{prooftree}
    \hspace{0.5cm}
        \begin{prooftree}
            \hypo{t_1 \cbvEq[\indexI+1] u_1}
            \hypo{t_2 \cbvEq[\indexI] u_2}
            \inferCbnEqApp{\app{t_1}{t_2} \cbvEq[\indexI+1] \app{u_1}{u_2}}
        \end{prooftree}
    \hspace{0.5cm}
        \begin{prooftree}
            \hypo{t_1 \cbvEq[\indexOmega] u_1}
            \hypo{t_2 \cbvEq[\indexOmega] u_2}
            \inferCbnEqApp{\app{t_1}{t_2} \cbvEq[\indexOmega] \app{u_1}{u_2}}
        \end{prooftree}
    \\[0.5cm]
    \hspace{-0.4cm}
        \begin{prooftree}
            \hypo{t_1 \cbvEq[0] u_1}
            \inferCbnEqEs[1]{t_1\esub{x}{t_2} \cbvEq[0] u_1\esub{x}{u_2}}
        \end{prooftree}
    \hspace{0.1cm}
        \begin{prooftree}
            \hypo{t_1 \cbvEq[\indexI+1] u_1}
            \hypo{t_2 \cbvEq[\indexI] u_2}
            \inferCbnEqEs{t_1\esub{x}{t_2} \cbvEq[\indexI+1] u_1\esub{x}{u_2}}
        \end{prooftree}
    \hspace{0.1cm}
        \begin{prooftree}
            \hypo{t_1 \cbvEq[\indexOmega] u_1}
            \hypo{t_2 \cbvEq[\indexOmega] u_2}
            \inferCbnEqEs{t_1\esub{x}{t_2} \cbvEq[\indexOmega] u_1\esub{x}{u_2}}
        \end{prooftree}
    \end{array}
\end{equation*}

\noindent We now prove the instance of
\Cref{ax:Gen_|t_in_bnoSn_and_t_<=_u_==>_u_in_BnoSn} for \CBNSymb.

\begin{lemma}[\CBNSymb Stability of Meaningful Observables]
    \label{lem:cbn_t_in_bnoSn_and_t_<=_u_==>_u_in_BnoSn}%
    Let $\indexK \in \setOrdinals$ and $\aprxT \in
    \cbnAprxBnoS_\indexK$. Then for any $\aprxU \in \setCbnAprxTerms$
    such that $\aprxT \cbnAprxLeq \aprxU$, one has that $\aprxU \in
    \cbnAprxBnoS_\indexK$ and $\aprxT \cbnEq[\indexK] \aprxU$.%
\end{lemma}
\stableProof{
    \begin{proof}
Let $\aprxT, \aprxU \in \setCbnAprxTerms$ such that $\aprxT
\cbnAprxLeq \aprxU$. We strengthen the induction hypothesis:
\begin{itemize}
\item[\bltI] If $\aprxT \in \cbnAprxBneS_\indexK$ then $\aprxU \in
    \cbnAprxBneS_\indexK$.

\item[\bltI] If $\aprxT \in \cbnAprxBnoS_\indexK$ then $\aprxU \in
    \cbnAprxBnoS_\indexK$.
\end{itemize}
Moreover, $\aprxT \cbnEq[\indexK] \aprxU$.

By mutual induction on $\aprxT \in \cbnAprxBnoS_\indexK$ and $\aprxT
\in \cbnAprxBneS_\indexK$:
\begin{itemize}
\item[\bltI] $\aprxT \in \cbnAprxBneS_\indexK$:
    \begin{itemize}
    \item[\bltII] $\aprxT = x$: Then $\aprxU = x \in \cbnAprxBneS$ and
        $\aprxT \cbnEq[\indexK] \aprxU$ using $(\cbnEqVarRuleName)$.
 
    \item[\bltII] $\aprxT = \app{\aprxT_1}{\aprxT_2}$: Then
        necessarily $\aprxU = \app[\,]{\aprxU_1}{\aprxU_2}$ with
        $\aprxT_1 \cbnAprxLeq \aprxU_1$ and $\aprxT_2 \cbnAprxLeq
        \aprxU_2$. Moreover $\aprxT_1 \in \cbnAprxBneS_\indexK$ thus
        by \ih on $\aprxT_1$, one has that $\aprxU_1 \in
        \cbnAprxBneS_\indexK$ and $\aprxT_1 \cbnEq[\indexK] \aprxU_1$.
        We distinguish two cases:
        \begin{itemize}
            \item[\bltIII] $\indexK = 0$: Then $\aprxU =
            \app{\aprxU_1}{\aprxU_2} \in \app{\cbnAprxBneS_0}{\aprxT}
            \subseteq \cbnAprxBneS_0$ and $\aprxT \cbnEq[\indexK]
            \aprxU$ using $(\cbnEqAppRuleName)$.

            \item[\bltIII] $\indexK \in \Nat \setminus \{0\}$ (\resp
            $\indexK = \indexOmega$): Then $m = \indexK - 1$ (\resp $m
            = \indexOmega$) and $\aprxT_2 \in \cbnAprxBnoS_m$. By \ih
            on $t_2$, one has that $\aprxU_2 \in \cbnAprxBnoS_m$ and
            $\aprxT_2 \cbnEq[m] \aprxU_2$. Thus $\aprxU =
            \app{\aprxU_1}{\aprxU_2} \in
            \app{\cbnAprxBneS_\indexK}{\cbnAprxBnoS_m} \subseteq
            \cbnAprxBneS_\indexK$ and $\aprxT \cbnEq[\indexK] \aprxU$
            using $(\cbnEqAppRuleName)$.
        \end{itemize}
    \end{itemize}

\item[\bltI] $\aprxT \in \cbnAprxBnoS_\indexK$:
    \begin{itemize}
    \item[\bltII] $\aprxT = \abs{x}{\aprxT'}$: Then $\aprxT' \in
        \cbnAprxBnoS_\indexK$ and $\aprxU = \abs{x}{\aprxU'}$ with
        $\aprxT' \cbnAprxLeq \aprxU'$. By \ih on $\aprxT'$, one has
        $\aprxU' \in \cbnAprxBnoS_\indexK$ and $\aprxT'
        \cbnEq[\indexK] \aprxU'$. Thus $\aprxU = \abs{x}{\aprxU'} \in
        \abs{x}{\cbnAprxBnoS_\indexK} \subseteq \cbnAprxBnoS_\indexK$
        and $\aprxT \cbnEq[\indexK] \aprxU$ using
        $(\cbnEqAbsRuleName)$.

    \item[\bltII] $\aprxT \in \cbnAprxBneS_\indexK$: By \ih on
        $\aprxT$, one has that $\aprxU \in \cbnAprxBneS_\indexK
        \subseteq \cbnAprxBnoS_\indexK$ and $\aprxT \cbnEq[\indexK]
        \aprxU$.
        \qedhere
    \end{itemize}
\end{itemize}
\end{proof}
}

\CbnQuantitativeStratifiedGenericity*
\label{prf:cbn_Quantitative_Stratified_Genericity}%
\stableProof{
    \begin{proof}
By definition, there exists $s \in \cbnNoS_\indexK$ such that
$\cbnCCtxt<t> \cbnArr*_{S_\indexK} s$. By iterating
\Cref{lem:cbn_t_->Sn_u_=>_MA(t)_->*Sn_MA(u),lem:cbn_t_->Sn_u_and_t_<=_t'_==>_t'_->Sn_u'_and_u_<=_u'},
there exist $\aprxS \in \setCbnAprxTerms$ and $n \in \mathbb{N}$ such
that $\cbnTApprox{\cbnCCtxt<t>} \cbnAprxArr^n_{S_\indexK} \aprxS
\cbnAprxGeq \cbnTApprox{s}$. Since $s \in \cbnNoS_\indexK$, then
$\cbnTApprox{s} \in \cbnAprxBnoS_\indexK$ using
\Cref{lem:cbn_t_in_BnoSn_==>_MA(t)_in_noSn}, thus $\aprxS \in
\cbnAprxBnoS_\indexK$ using
\Cref{lem:cbn_t_in_bnoSn_and_t_<=_u_==>_u_in_BnoSn}. On one hand, by
\Cref{lem:Cbn_MA(t)_<=_t}, one has that $\cbnTApprox{\cbnCCtxt<t>}
\cbnAprxLeq \cbnCCtxt<t>$ thus, by iterating
\Cref{lem:cbn_t_->Sn_u_and_t_<=_t'_==>_t'_->Sn_u'_and_u_<=_u'}, one
has that $\cbnCCtxt<t> \cbnArr^n_{S_\indexK} t'$ for some $t' \in
\setCbnTerms$ such that $\aprxS \cbnAprxLeq t'$. By
\Cref{lem:cbn_t_in_bnoSn_and_t_<=_u_==>_u_in_BnoSn}, one deduces that
$t' \in \cbnNoS_\indexK$ and $t' \cbnEq[\indexK] \aprxS$. One the
other hand, using \Cref{lem:cbn_Approximant_Genericity}, one has that
$\cbnTApprox{\cbnCCtxt<t>} \cbnAprxLeq \cbnTApprox{\cbnCCtxt<u>}$ thus
$\cbnTApprox{\cbnCCtxt<t>} \cbnAprxLeq \cbnCCtxt<u>$ using
\Cref{lem:Cbn_MA(t)_<=_t} and transitivity. By iterating
\Cref{lem:cbn_t_->Sn_u_and_t_<=_t'_==>_t'_->Sn_u'_and_u_<=_u'} again,
one has that $\cbnCCtxt<u> \cbnArr^n_{S_\indexK} u'$ for some $u' \in
\setCbnTerms$ such that $\aprxS \cbnAprxLeq u'$. By
\Cref{lem:cbn_t_in_bnoSn_and_t_<=_u_==>_u_in_BnoSn}, one deduces that
$u' \in \cbnNoS_\indexK$ and $u' \cbnEq[\indexK] \aprxS$. Finally, by
transitivity of $\cbnEq[\indexK]$, one deduces that $t'
\cbnEq[\indexK] u'$.

Graphically, these are the main steps of the proof:
\begin{center}
    \begin{tikzpicture}
        \node               (A)         at (0, 0.25)    {$\cbnTApprox{\cbnCCtxt<t>}$};
        \node               (A')        at (0, -2)      {$\aprxS\!$};

        \node               (Ct)        at (-1, -0.5)   {$\cbnCCtxt<t>$};
        \node               (t')        at (-1, -2.5)   {$t'\!$};
        \node[rotate=-90]   (t'_in)     at (-1, -2.85)  {$\in$};
        \node               (t'_set)    at (-1, -3.2)   {$\cbnNoS_\indexK$};

        \node               (Cu)        at (1, -0.5)    {$\cbnCCtxt<u>$};
        \node               (u')        at (1, -2.5)    {$u'\!$};
        \node[rotate=-90]   (u'_in)     at (1, -2.85)   {$\in$};
        \node               (u'_set)    at (1, -3.2)    {$\cbnNoS_\indexK$};

        \draw[->>, line width=0.2mm] (A) to (A'.north);
        \node at ([yshift=-2.05cm, xshift=1.8mm]A.north)
            {{$\scriptscriptstyle n$}};
        \node at ([yshift=-2.05cm, xshift=-1.8mm]A.north)
            {{$\scriptscriptstyle \cbnAprxSymbSurface_\indexK$}};

        \draw[->>, line width=0.2mm] (Ct) to (t'.north);
        \node at ([yshift=-1.8cm, xshift=1.8mm]Cu.north)
            {{$\scriptscriptstyle n$}};
        \node at ([yshift=-1.8cm, xshift=-1.8mm]Cu.north)
            {{$\scriptscriptstyle \cbnSymbSurface_\indexK$}};

        \draw[->>, line width=0.2mm] (Cu) to (u'.north);
        \node at ([yshift=-1.8cm, xshift=-1.8mm]Ct.north)
            {{$\scriptscriptstyle n$}};
        \node at ([yshift=-1.8cm, xshift=1.8mm]Ct.north)
            {{$\scriptscriptstyle \cbnSymbSurface_\indexK$}};

        \node[rotate=25]    (deduceArrowLeft)   at (-0.55, -1.5) {$\leftsquigarrow$};
        \node at ([xshift=1mm, yshift=0.8mm]deduceArrowLeft.north)
        {{$\scriptscriptstyle {\tt Lem}.\ref{lem:cbn_t_->Sn_u_and_t_<=_t'_==>_t'_->Sn_u'_and_u_<=_u'}$}};

        \node[rotate=25]    (deduceArrowLeft)   at (-0.55, -1) {$\rightsquigarrow$};
        \node at ([xshift=1mm, yshift=0.8mm]deduceArrowLeft.north)
        {{$\scriptscriptstyle {\tt Lem}.\ref{lem:cbn_t_->Sn_u_=>_MA(t)_->*Sn_MA(u)}$}};

        \node[rotate=-25]   (deduceArrowRight)  at (0.5, -1.25)  {$\rightsquigarrow$};
        \node at ([xshift=-1mm, yshift=0.8mm]deduceArrowRight.north)
        {{$\scriptscriptstyle {\tt Lem}.\ref{lem:cbn_t_->Sn_u_and_t_<=_t'_==>_t'_->Sn_u'_and_u_<=_u'}$}};

        \node[rotate=35]    (LeqCt)     at (-0.6, -0.15)    {$\cbnAprxGeq$};
        \node   at ([yshift=1mm, xshift=-5mm]LeqCt)
            {$\scriptscriptstyle {\tt Lem.} \ref{lem:Cbn_MA(t)_<=_t}$};

        \node[rotate=35]    (LeqT')     at (-0.5, -2.13)    {$\cbnAprxGeq$};

        \node[rotate=-40]   (LeqCu)  at (0.6, -0.15)     {$\cbnAprxLeq$}; 
        \node   at ([yshift=1mm, xshift=8mm]LeqCu)
            {$\scriptscriptstyle {\tt Lem.} \ref{lem:cbn_Approximant_Genericity} \;\&\; {\tt Lem.} \ref{lem:Cbn_MA(t)_<=_t}$};

        \node[rotate=-40]   (LeqU')  at (0.5, -2.13)     {$\cbnAprxLeq$};


        \node               (eq)         at (0, -2.5)     {$\cbnLongEq[\indexK]$};
        \node   at ([yshift=-2.75mm]eq)
            {$\scriptscriptstyle {\tt Lem.} \ref{lem:cbn_t_in_bnoSn_and_t_<=_u_==>_u_in_BnoSn}$};
    \end{tikzpicture}
\end{center}
\end{proof}
}%

\section{Proofs of Section \ref{sec:theories}}
\label{sec:ProofTheories}%

\subsection{Proofs of Subsection \ref{sec:generic-theories}}


Let us recall the rules of the congruence $\genThCongr$.
\begin{equation*}
    \begin{array}{c}
        \begin{prooftree}
            \inferGenThCongrRefl{t}{t}
        \end{prooftree}
    \hspace{0.6cm}
        \begin{prooftree}
            \hypoGenThCongr{t}{u}
            \inferGenThCongrSym{u}{t}
        \end{prooftree}
    \hspace{0.6cm}
        \begin{prooftree}
            \hypoGenThCongr{t}{s}
            \hypoGenThCongr{s}{u}
            \inferGenThCongrTrans{t}{u}
        \end{prooftree}
    \hspace{0.6cm}
        \begin{prooftree}
            \hypoGenThCongr{t}{u}
            \inferGenThCongrCtxt{\genCCtxt<t>}{\genCCtxt<u>}
        \end{prooftree}
    \hspace{0.6cm}
        \begin{prooftree}
            \inferGenThCongrStep{t}{u}
        \end{prooftree}
    \end{array}
\end{equation*}
We define the theory $\genThHax$ to be the set of equations generated
by the following rules.
\begin{equation*}
    \begin{array}{c}
        \begin{prooftree}
            \inferGenThHaxRefl{t}{t}
        \end{prooftree}
    \hspace{0.5cm}
        \begin{prooftree}
            \inferGenThHaxBase{t}{u}
        \end{prooftree}
    \hspace{0.5cm}
        \begin{prooftree}
            \hypoGenThHaxDSim{t}{u}
            \inferGenThHaxSym{u}{t}
        \end{prooftree}
    \hspace{0.5cm}
        \begin{prooftree}
            \hypoGenThHaxDSim{t}{u}
            \inferGenThHaxConvSSim{t}{u}
        \end{prooftree}
    \hspace{0.5cm}
        \begin{prooftree}
            \hypoGenThHaxSSim{t}{u}
            \inferGenThHaxCtxt{\genCCtxt<t>}{\genCCtxt<u>}
        \end{prooftree}
    \hspace{0.5cm}
        \begin{prooftree}
            \hypoGenThHaxSSim{t}{u}
            \inferGenThHaxConvEq{t}{u}
        \end{prooftree}
    \hspace{0.5cm}
        \begin{prooftree}
            \hypoGenThHaxEq{t}{s}
            \hypoGenThHaxEq{s}{u}
            \inferGenThHaxTrans{t}{u}
        \end{prooftree}
    \end{array}
\end{equation*}

We write $\genThHaxTail{t}{u}$ (\resp
$\genThHaxTail[\genThHaxDSim]{t}{u}$,
$\genThHaxTail[\genThHaxSSim]{t}{u}$) if $t \genThHaxEq u \in
\genThHax$ (\resp $t \genThHaxDSim u \in \genThHax$, $t \genThHaxSSim
u \in \genThHax$).

\begin{restatable}{lemma}{GenTheoryTaxEquivalentTheoryH}
    \label{lem:Theories_Tax_|-_t_=_u_iff_H_|-_t_=_u}%
    Let $t, u \in \setGenTerms$. Then $\genThHaxTail{t}{u}$ iff
     $\genThHTail{t}{u}$.
\end{restatable}
\label{prf:Theories_Tax_|-_t_=_u_iff_H_|-_t_=_u}%
\stableProof{
    \begin{proof}
By double implication:
\begin{itemize}
\item[$(\Rightarrow)$] By mutual induction on the length of the proof
    of $\genThHaxTail{t}{u}$, $\genThHaxTail[\genThHaxDSim]{t}{u}$ and
    $\genThHaxTail[\genThHaxSSim]{t}{u}$, one shows that
    $\genThHTail{t}{u}$. Cases on the last rule applied:
    \begin{itemize}
    \item[\bltI] $(\genThHaxReflSymbDSim)$: Then $t = u$ and using
        $(\genThCongrReflSymb)$, one concludes that
        $\genThHTail{t}{u}$.

    \item[\bltI] $(\genThHaxBaseSymbDSim)$: Then $(t, u) \in
        \genThHBaseEquations$ and using $(\genThCongrStepSymb)$, one
        concludes that $\genThHTail{t}{u}$.

    \item[\bltI] $(\genThHaxSymSymbDSim)$: Then
        $\genThHaxTail[\genThHaxDSim]{u}{t}$ and by \ih, one obtains
        $\genThHTail{u}{t}$ thus using $(\genThCongrSymSymb)$, one
        concludes that $\genThHTail{t}{u}$.

    \item[\bltI] $(\genThHaxConvSymbSSim)$: Then
        $\genThHaxTail[\genThHaxDSim]{t}{u}$ thus by \ih, one
        concludes that $\genThHTail{t}{u}$.

    \item[\bltI] $(\genThHaxCtxtSymbSSim)$: Then $t = \genCCtxt<t'>$
        and $u = \genCCtxt<u'>$ for some $t', u' \in \setGenTerms$
        such that $\genThHaxTail[\genThHaxSSim]{t'}{u'}$. By \ih, one
        obtains that $\genThHTail{t'}{u'}$ thus using
        $(\genThCongrCtxtSymb)$, one concludes that
        $\genThHTail{t}{u}$.

    \item[\bltI] $(\genThHaxConvSymbEq)$: Then
        $\genThHaxTail[\genThHaxSSim]{t}{u}$ and by \ih, one concludes
        that $\genThHTail{t}{u}$.

    \item[\bltI] $(\genThHaxTransSymbEq)$: Then $\genThHaxTail{t}{s}$
        and $\genThHaxTail{s}{u}$ for some $s \in \setGenTerms$. By
        \ih on both, one has that $\genThHTail{t}{s}$ and
        $\genThHTail{s}{u}$ thus using $(\genThCongrTransSymb)$, one
        concludes that $\genThHTail{t}{u}$.
    \end{itemize}

\item[$(\Leftarrow)$] Let us first show the following property
    (\textbf{P1}): provability of $\genThHaxTail[\genThHaxEq]{t}{u}$
    and $\genThHaxTail[\genThHaxSSim]{t}{u}$ is symmetric. By mutual
    induction on $\genThHaxTail[\genThHaxEq]{t}{u}$ and
    $\genThHaxTail[\genThHaxSSim]{t}{u}$. Cases on the last rule
    applied:
    \begin{itemize}
    \item[\bltI] $(\genThHaxConvSymbSSim)$: Then
        $\genThHaxTail[\genThHaxDSim]{t}{u}$ and using
        $(\genThHaxSymSymbDSim)$, one has that
        $\genThHaxTail[\genThHaxDSim]{u}{t}$ thus using
        $(\genThHaxConvSymbSSim)$, one concludes that
        $\genThHaxTail[\genThHaxSSim]{u}{t}$.

    \item[\bltI] $(\genThHaxCtxtSymbSSim)$: Then $t = \genCCtxt<t'>$,
        $u = \genCCtxt<u'>$ for some $t', u' \in \setGenTerms$ such
        that $\genThHaxTail[\genThHaxSSim]{t'}{u'}$. By \ih, one has
        that $\genThHaxTail[\genThHaxSSim]{u'}{t'}$ thus using
        $(\genThHaxCtxtSymbSSim)$, one concludes that
        $\genThHaxTail[\genThHaxSSim]{u}{t}$.

    \item[\bltI] $(\genThHaxConvSymbEq)$: Then
        $\genThHaxTail[\genThHaxSSim]{t}{u}$ and by \ih, one has that
        $\genThHaxTail[\genThHaxSSim]{u}{t}$ thus using
        $(\genThHaxConvSymbEq)$, one concludes that
        $\genThHaxTail{u}{t}$.

    \item[\bltI] $(\genThHaxTransSymbEq)$: Then $\genThHaxTail{t}{s}$
        and $\genThHaxTail{s}{u}$ for some $s \in \setGenTerms$. By
        \ih on both, one has that $\genThHaxTail{s}{t}$ and
        $\genThHaxTail{u}{s}$ thus, using $(\genThHaxTransSymbEq)$,
        one concludes that $\genThHaxTail{u}{t}$.

    \end{itemize}
    We then show the following property (\textbf{P2}): provability of
    $\genThHaxTail{t}{u}$ is a congruence. Let $\genThHaxTail{t}{u}$,
    let us show that $\genThHaxTail{\genCCtxt<t>}{\genCCtxt<u>}$. By
    induction on the length of the proof of $\genThHaxTail{t}{u}$.
    Cases on the last rule applied:
    \begin{itemize}
    \item[\bltI] $(\genThHaxConvSymbEq)$: Then
        $\genThHaxTail[\genThHaxSSim]{t}{u}$ thus using
        $(\genThHaxCtxtSymbSSim)$, one obtains
        $\genThHaxTail[\genThHaxSSim]{\genCCtxt<t>}{\genCCtxt<u>}$ and
        using $(\genThHaxConvSymbEq)$, one concludes that
        $\genThHaxTail{\genCCtxt<t>}{\genCCtxt<u>}$.

    \item[\bltI] $(\genThHaxTransSymbEq)$: Then $\genThHaxTail{t}{s}$
        and $\genThHaxTail{s}{u}$ for some $s \in \setGenTerms$. By
        \ih on both, one has that
        $\genThHaxTail{\genCCtxt<t>}{\genCCtxt<s>}$ and
        $\genThHaxTail{\genCCtxt<s>}{\genCCtxt<u>}$ thus using
        $(\genThHaxTransSymbEq)$, one concludes that
        $\genThHaxTail{\genCCtxt<t>}{\genCCtxt<u>}$.
    \end{itemize}

    Let us now show the main property. By induction on the length of
    the proof of $\genThHTail{t}{u}$. Cases on the last rule applied:
    \begin{itemize}
    \item[\bltI] $(\genThCongrReflSymb)$: Then $t = u$ and using
        $(\genThHaxReflSymbDSim)$, $(\genThHaxConvSymbSSim)$ and
        $(\genThHaxConvSymbEq)$, one concludes that
        $\genThHaxTail{t}{u}$.

    \item[\bltI] $(\genThCongrStepSymb)$: Then $(t, u) \in
        \genThHBaseEquations$ and using $(\genThHaxBaseSymbDSim)$,
        $(\genThHaxConvSymbSSim)$ and $(\genThHaxConvSymbEq)$, one
        concludes that $\genThHaxTail{t}{u}$.

    \item[\bltI] $(\genThCongrSymSymb)$: Then $\genThHTail{u}{t}$ thus
        by \ih $\genThHaxTail{u}{t}$ and by property (\textbf{P1}),
        one concludes that $\genThHaxTail{t}{u}$.

    \item[\bltI] $(\genThCongrTransSymb)$: Then $\genThHTail{t}{s}$
        and $\genThHTail{s}{u}$ and by \ih on both, one has that
        $\genThHaxTail{t}{s}$ and $\genThHaxTail{s}{u}$ thus using
        $(\genThHaxTransSymbEq)$, one concludes that
        $\genThHaxTail{t}{u}$.

    \item[\bltI] $(\genThCongrCtxtSymb)$: Then $t = \genCCtxt<t'>$ and
        $u = \genCCtxt<u'>$ for some $t', u' \in \setGenTerms$ such
        that and $\genThHTail{t'}{u'}$. By \ih, one has
        $\genThHaxTail{t'}{u'}$ and by property (\textbf{P2}), one
        concludes that $\genThHaxTail{t}{u}$.
        \qedhere
    \end{itemize}
\end{itemize}
\end{proof}

}

\begin{lemma}
    \label{lem:gen_Tax_Decompose}%
    Let $t, u \in \setGenTerms$, then
    \begin{enumerate}
    \item If $\genThHaxTail{t}{u}$, then there exist $n \in
        \mathbb{N}$ and $s_0, \dots, s_n \in \setGenTerms$ such that
        $t = s_0$, $u = s_n$ and
        $\genThHaxTail[\genThHaxSSim]{s_i}{s_{i+1}}$ for $0 \leq i <
        n$. \label{lem:gen_TaxEq_Decompose}%

    \item $\genThHaxTail[\genThHaxDSim]{t}{u}$ \;if and only if\;
        $\genThLambdaTail{t}{u}$ or $\genThLambdaTail{u}{t}$ from an
        axiom, or $t$ and $u$ are meaningless.
        \label{lem:gen_TaxDSim_Decompose}%
    \end{enumerate}
\end{lemma}
\stableProof{
    \begin{proof} ~
\begin{enumerate}
\item By induction on the length of the proof of
    $\genThHaxTail{t}{u}$. Cases on the last rule applied:
    \begin{itemize}
    \item[\bltI] $(\genThHaxConvSymbEq)$: Then
        $\genThHaxTail[\genThHaxSSim]{t}{u}$ and we set $n := 0$ with
        $s_0 := t = u$ which concludes this case since using
        $(\genThHaxReflSymbDSim)$ and $(\genThHaxConvSymbSSim)$, one
        has $\genThHaxTail[\genThHaxSSim]{s_0}{s_0}$.

    \item[\bltI] $(\genThHaxTransSymbEq)$: Then $\genThHaxTail{t}{s'}$
        and $\genThHaxTail{s'}{u}$ for some $s' \in \setCbnTerms$. By
        \ih on both, one has $n_1, n_2 \in \mathbb{N}$ and $s_1^0,
        \dots, s_1^{n_1}, s_2^0, \dots, s_2^{n_2} \in \setCbnTerms$
        such that $t = s_1^0$, $s' = s_1^{n_1}$, $s' = s_2^0$, $u =
        s_2^{n_2}$, $\genThHaxTail[\genThHaxSSim]{s_1^i}{s_1^{i+1}}$
        for all $0 \leq i < n_1$ and
        $\genThHaxTail[\genThHaxSSim]{s_2^i}{s_2^{i+1}}$ for all $0
        \leq i < n_2$. We set $n := n_1 + n_2$, $s_i := s_1^i$ for $0
        \leq i \leq n_1$ and $s_{n_1 + i} := s_2^i$ for $0 < i \leq
        n_2$ which concludes this case since $s_1^{n_1} = s' = s_2^0$.
    \end{itemize}

\item By double implication:
    \begin{itemize}
    \item[$(\Rightarrow)$] By induction on the length of the proof of
        $\genThHaxTail[\genThHaxDSim]{t}{u}$. Cases on the last rule
        applied:
        \begin{itemize}
        \item[\bltI] $(\genThHaxReflSymbDSim)$: Then
            $\genThLambdaTail{t}{u}$ using $(\genThLambdaReflSymb)$.

        \item[\bltI] $(\genThHaxBaseSymbDSim)$: Then $(t, u) \in
            \genThHBaseEquations$. We distinguish two cases:
            \begin{itemize}
            \item[\bltII] $(t, u \in \genThLambdaBaseEquations)$: Then
                $\genThLambdaTail{t}{u}$ using
                $(\genThLambdaStepSymb)$.
            \item[\bltII] $t, u$ meaningless.
            \end{itemize}

        \item[\bltI] $(\genThHaxSymSymbDSim)$: Then
            $\genThHaxTail[\genThHaxDSim]{u}{t}$ and by \ih, one has
            that either $\genThLambdaTail{u}{t}$ or
            $\genThLambdaTail{t}{u}$ from an axiom, or $t$ and $u$ are
            meaningless, thus concluding this case.
        \end{itemize}

    \item[$(\Leftarrow)$] We distinguish three cases:
        \begin{itemize}
        \item[\bltI] $\genThLambdaTail{t}{u}$ from an axiom: We
            distinguish two cases:
            \begin{itemize}
            \item[\bltII] $(\genThLambdaReflSymb)$: Then $t = u$ and
                using $(\genThHaxReflSymbDSim)$, one concludes that
                $\genThHaxTail[\genThHaxDSim]{t}{u}$.

            \item[\bltII] $(\genThLambdaStepSymb)$: Then $(t, u) \in
                \genThLambdaBaseEquations$ thus
                $\genThHaxTail[\genThHaxDSim]{t}{u}$ using
                $(\genThHaxBaseSymbDSim)$.
            \end{itemize}

        \item[\bltI] $\genThLambdaTail{u}{t}$ from an axiom: We
            distinguish two cases:
            \begin{itemize}
            \item[\bltII] $(\genThLambdaReflSymb)$: Then $u = t$ and
                using $(\genThHaxReflSymbDSim)$, one concludes that
                $\genThHaxTail[\genThHaxDSim]{t}{u}$.

            \item[\bltII] $(\genThLambdaStepSymb)$: Then $(t, u) \in
            \genThLambdaBaseEquations$ thus
            $\genThHaxTail[\genThHaxDSim]{t}{u}$ using
            $(\genThHaxBaseSymbDSim)$ and $(\genThHaxSymSymbDSim)$.
            \end{itemize}

        \item[\bltI] $t$ and $u$ are meaningless: Then using
            $(\genThHaxBaseSymbDSim)$, one concludes that
            $\genThHaxTail{t}{u}$.
            \qedhere
        \end{itemize}
    \end{itemize}
\end{enumerate}
\end{proof}
}

\begin{lemma}
    \label{lem:gen_TaxSSim_Decompose}%
    Let $t, u \in \setGenTerms$ such that
    $\genThHaxTail[\genThHaxSSim]{t}{u}$, then there exist a context
    $\genCCtxt$ and terms $t', u' \in \setGenTerms$ such that $t =
    \genCCtxt<t'>$, $u = \genCCtxt<u'>$ and
    $\genThHaxTail[\genThHaxDSim]{t'}{u'}$.
\end{lemma}
\stableProof{
    \begin{proof}
By induction on the length of the proof of
$\genThHaxTail[\genThHaxSSim]{t}{u}$. Cases on the last rule applied:
\begin{itemize}
\item[\bltI] $(\genThHaxConvSymbSSim)$: Then
    $\genThHaxTail[\genThHaxDSim]{t}{u}$ and we set $\genCCtxt :=
    \Hole$, $t' := t$ and $u' := u$ which concludes this case since $t
    = \Hole\CtxtPlug{t} = \genCCtxt<t'>$, $u = \Hole\CtxtPlug{u} =
    \genCCtxt<u'>$ and $\genThHaxTail[\genThHaxDSim]{t'}{u'}$ by
    hypothesis.

\item[\bltI] $(\genThHaxCtxtSymbSSim)$: Then $t = \genCCtxt''<t''>$
    and $u = \genCCtxt''<u''>$ for some terms $t', u' \in
    \setGenTerms$ such that $\genThHaxTail[\genThHaxSSim]{t''}{u''}$.
    By \ih, there exists a context $\genCCtxt'$ and terms $t', u' \in
    \setCbnTerms$ such that $t'' = \genCCtxt'<t'>$, $u'' =
    \genCCtxt'<u'>$ and $\genThHaxTail[\genThHaxDSim]{t'}{u'}$. We set
    $\genCCtxt := \genCCtxt''<\genCCtxt'>$ which concludes this case
    since $t = \genCCtxt''<t''> = \genCCtxt''<\genCCtxt'<t'>> =
    (\genCCtxt''<\genCCtxt'>)\CtxtPlug{t'} = \genCCtxt<t'>$ and $u =
    \genCCtxt''<u''> = \genCCtxt''<\genCCtxt'<u'>> =
    (\genCCtxt''<\genCCtxt'>)\CtxtPlug{t'} = \genCCtxt<u'>$.
    \qedhere
\end{itemize}
\end{proof}
}

Let us recall the first axiom:

\AssumptionTheoryFullConfluence*

\begin{lemma}
    \label{lem:gen_lambda_|-_t_=_u_==>_t_has_FNF_iff_u_has_FNF}%
    Let $t, u \in \setGenTerms$ such that $\genThLambdaTail{t}{u}$.
    Then $t$ has a $\indexOmega$-normal form if and only if $u$ has a
    $\indexOmega$-normal form.

    Moreover, they share the same $\genSCtxt_\indexOmega$-normal form.
\end{lemma}
\stableProof{
    \begin{proof}
Let $t, u \in \setGenTerms$ such that $\genThLambdaTail{t}{u}$. Using
\Cref{lem:Theories_lambda_|-_t_=_u_iff_t_=F_u}, one deduces that $t
\genSymArr_{\setrules} u$. By induction on the length of the sequence
$t \genSymArr_{\setrules} u$.
\begin{itemize}
\item[\bltI] $({\tt refl})$: Then $t$ is equal to $u$ and trivially
    $t$ has a $\indexOmega$-normal form if and only if $u$ has a
    $\indexOmega$-normal form. Moreover, they also share the same
    $\indexOmega$-normal form.

\item[\bltI] $({\tt one-step})$: Then $t \genArr_\indexOmega u$ and
    one concludes by Confluence (\Cref{ax:Confluence}).

\item[\bltI] $({\tt sym})$: Then $t \genSymArr_{\setrules} u$ and one
    concludes by \ih.

  \item[\bltI] $({\tt trans})$: Then $t \genSymArr_{\setrules} s$ and
    $s \genSymArr_{\setrules} u$ for some $s \in \setGenTerms$. We
    distinguish two cases:
    \begin{itemize}
    \item[\bltII] $t$ has a $\indexOmega$-normal form: Then by \ih on
        $t \genSymArr_{\setrules} s$, one deduces that $s$ has the
        same $\indexOmega$-normal form as $t$ thus by \ih on $s
        \genSymArr_{\setrules} u$, one concludes that $u$ has the same
        $\indexOmega$-normal form as $t$.

    \item[\bltII] $u$ has a $\indexOmega$-normal form: Symmetric of
        the previous case.
    \end{itemize}
\end{itemize}
\end{proof}

}

\begin{restatable}{lemma}{GenTheoryHAndTaxAgreeFNF}
    \label{lem:Theories_Tax_|-_t_=_u_and_u_FNF_==>_lambda_|-_t_=_u}%
    If $u \in \setGenTerms$ has a $\indexOmega$-normal form, then for
    any $t \in \setGenTerms$, if $\genThHaxTail{t}{u}$ then
    $\genThLambdaTail{t}{u}$.
\end{restatable}
\label{prf:Theories_Tax_|-_t_=_u_and_u_FNF_==>_lambda_|-_t_=_u}%
\stableProof{
    \begin{proof}
Let $t, u \in \setGenTerms$ such that $u$ has a $\indexOmega$-normal
form and $\genThHaxTail{t}{u}$. Using
\Cref{lem:gen_Tax_Decompose}:\ref{lem:gen_TaxEq_Decompose}, one has $n
\in \mathbb{N}$ and $s_0, \dots, s_n \in \setGenTerms$ such that $t =
s_0$, $u = s_n$ and $\genThHaxTail[\genThHaxSSim]{s_i}{s_{i+1}}$ for
$0 \leq i < n$. Let us show the following:
\begin{center}
    If $s_n$ has a $\indexOmega$-normal form, and
    $\genThHaxTail[\genThHaxSSim]{s_i}{s_{i+1}}$ for $0 \leq i < n$,
    then $\genThLambdaTail{s_0}{s_n}$. 
\end{center}

We reason by induction on $n \in \mathbb{N}$. Cases:
\begin{itemize}
\item[\bltI] $n = 0$: Then $s_0 = s_n$ and thus
    $\genThLambdaTail{s_0}{s_n}$ by $(\genThLambdaReflSymb)$.

\item[\bltI] $n = m + 1$: Then using \Cref{lem:gen_TaxSSim_Decompose},
    there exist a context $\genCCtxt$ and terms $s'_m, s'_{m+1} \in
    \setGenTerms$ such that $s_m = \genCCtxt<s'_m>$, $s_{m+1} =
    \genCCtxt<s'_{m+1}>$ and
    $\genThHaxTail[\genThHaxDSim]{s'_m}{s'_{m+1}}$. By
    \Cref{lem:gen_Tax_Decompose}:\ref{lem:gen_TaxDSim_Decompose}, we
    distinguish two cases:
    \begin{itemize}
    \item[\bltII] $\genThLambdaTail{s'_m}{s'_{m+1}}$ (\resp
        $\genThLambdaTail{s'_{m+1}}{s'_m}$) from an axiom: Then
        $\genThLambdaBaseEquations$ $\vdash \genCCtxt<s'_m> \doteq
        \genCCtxt<s'_{m+1}>$ by contextuality (\resp and symmetry)
        thus using
        \Cref{lem:gen_lambda_|-_t_=_u_==>_t_has_FNF_iff_u_has_FNF},
        one deduces that $s_m = \genCCtxt<s'_m>$ has a
        $\indexOmega$-norm form. By \ih, one concludes that
        $\genThLambdaTail{s_0}{s_m}$, thus
        $\genThLambdaTail{s_0}{s_{m+1}}$ holds by
        $(\genThLambdaTransSymb)$.

    \item[\bltII] $s'_m$ and $s'_{m+1}$ are meaningless: Since
        $s_{m+1} = \genCCtxt<s'_{m+1}>$ has a $\indexOmega$-normal
        form, let say $s \in \setGenTerms$, then $s_{m+1}
        \genSymArr_{\setrules} s$, which implies
        $\genThLambdaTail{s_{m+1}}{s}$ by
        \Cref{lem:Theories_lambda_|-_t_=_u_iff_t_=F_u}. Using Full
        Genericity (\Cref{ax:Full_Genericity}), one deduces that $s_m
        = \genCCtxt<s'_m>$ has the same $\indexOmega$-normal form $s$
        thus $\genThLambdaTail{s_m}{s}$ by
        \Cref{lem:Theories_lambda_|-_t_=_u_iff_t_=F_u}. Therefore
        $\genThLambdaTail{s_m}{s_{m+1}}$ using $(\genThLambdaSymSymb)$
        and $(\genThLambdaTransSymb)$. By \ih, one has that
        $\genThLambdaTail{s_0}{s_m}$ and we conclude that
        $\genThLambdaTail{s_0}{s_{m+1}}$ using
        $(\genThLambdaTransSymb)$.
        \qedhere
    \end{itemize}
\end{itemize}
\end{proof}

}

\GenTheoryHAndLambdaAgreeFNF*
\label{prf:Theories_T_|-_t_=_u_and_u_FNF_==>_lambda_|-_t_=_u}
\begin{proof}
    Using
    \Cref{lem:Theories_Tax_|-_t_=_u_and_u_FNF_==>_lambda_|-_t_=_u} and
    \Cref{lem:Theories_Tax_|-_t_=_u_iff_H_|-_t_=_u}.
\end{proof}

\begin{lemma} ~
    \label{lem:Theories_Subseteq_H*}%
    \begin{enumerate}
    \item $\genThLambda \subseteq \genThH*$.
        \label{lem:Theories_Subseteq_H*_Lambda}%
    \item $\genThH \subseteq \genThH*$.
        \label{lem:Theories_Subseteq_H*_H}%
    \end{enumerate}
\end{lemma}
\stableProof{
    \begin{proof} ~
\begin{enumerate}
\item Let $\genThLambdaTail{t}{u}$. Let us show that $(t, u) \in
    \genThH*$ by induction on the length of the proof
    $\genThLambdaTail{t}{u}$. Cases on the last rule applied:
    \begin{itemize}
    \item[\bltI] $(\genThCongrReflSymb)$: Then $t = u$ thus $(t, u)
        \in \genThH*$ trivially holds.

    \item[\bltI] $(\genThCongrStepSymb)$: Then $(t, u) \in
        \genThLambdaBaseEquations$ and using
        \Cref{ax:Meaningfulness_Stable_Reduction}, one deduces that
        $t$ is meaningful if and only if $u$ is meaningful so that
        $(t, u) \in \genThH*$.

    \item[\bltI] $(\genThCongrSymSymb)$: Then $\genThLambdaTail{u}{t}$
        and, by \ih,  $(u,t) \in \genThH*$; one
        concludes by symmetry of the logical equivalence in the definition of $\genThH*$.

    \item[\bltI] $(\genThCongrTransSymb)$: Then
        $\genThLambdaTail{t}{s}$ and $\genThLambdaTail{s}{u}$ for some $s \in \setGenTerms$, and by
        \ih on both, one deduces that $(t, s) \in \genThH*$ and $(s,
        u) \in \genThH*$ respectively. Therefore one concludes by
        transitivity of the logical equivalence in the definition of $\genThH*$.
    \end{itemize}

\item By (\ref{lem:Theories_Subseteq_H*_Lambda}) one has that
    $\genThLambda \subseteq \genThH*$, then it is sufficient to show
    that $\genThH*$ equates all meaningless terms. Let $t, u \in
    \setCbnTerms$ be both meaningless. By Surface Genericity
    (\Cref{ax:Surface_Genericity}), one has that $\genCCtxt<t>$ is
    meaningful if and only if $\genCCtxt<u>$ is meaningful, thus $(t,
    u) \in \genThH*$.
    \qedhere
\end{enumerate}
\end{proof}

}

%
%
%

\subsection{Proofs of Subsection \ref{sec:consistent-sensible-cbv}}

\begin{lemma}
    \label{lem:Cbv_App_Meaningful_==>_Arg_Meaningful}%
    Let $t, u \in \setCbvTerms$ and $x \in \setCbvVariables$ such that
    $t$ is meaningful, $x$ is fresh in $t$ and $u$ and
    $\app{\abs{x}{t}}{u}$ is meaningful, then $u$ is meaningful.
\end{lemma}
\begin{proof}
    Let $t, u \in \setCbvTerms$ and $x \in \setCbvVariables$ such that
    $t$ is meaningful, $x$ is fresh in $t$ and $u$ and
    $\app{(\abs{x}{t})}{u}$ is meaningful. Then by
    \Cref{lem:cbvBKRV_characterizes_meaningfulness}, there exists $\Pi
    \cbvTrBKRV \Gamma \vdash \app{(\abs{x}{t})}{u} : \sigma$. Then
    $\Pi$ is necessarily of the following form:
    \begin{equation*}
        \begin{prooftree}
            \hypo{\Pi_1 \cbvTrBKRV \Gamma_1 \vdash \abs{x}{t} : \M \typeArrow \sigma}
            \hypo{\Pi_2 \cbvTrBKRV \Gamma_2 \vdash u : \M}
            \inferCbvAGKApp{\Gamma_1 + \Gamma_2 \vdash \app{(\abs{x}{t})}{u} : \sigma}
        \end{prooftree}
    \end{equation*}
    with $\Gamma = \Gamma_1 + \Gamma_2$. Since $\Pi_2 \cbvTrBKRV
    \Gamma_2 \vdash u : \M$, one then deduces that $u$ is meaningful
    using \Cref{lem:cbvBKRV_characterizes_meaningfulness_logical}.
\end{proof}

\CbvTheoryHExtended*%
\label{prf:Cbv_Theory_H_Extended_Single}%
\begin{proof}
    Let $t, u \in \setCbvTerms$ such that
    $\genThCongr_{\mathsf{H}_\mathsf{v} \cup \{(t, u)\}}$ is
    consistent. Suppose by absurd that there exists a context
    $\cbvCCtxt$ such that $\cbvCCtxt<t>$ is meaningful and
    $\cbvCCtxt<u>$ is meaningless. Let us show that
    $\genThCongr_{\mathsf{H}_\mathsf{v} \cup \{(t, u)\}}$ is
    inconsistent.

    Let $s \in \setCbvTerms$ be a meaningful term and $x \in
    \setCbvVariables$ fresh. On one hand, by definition of
    meaningfulness, there exists a testing context $\cbvTCtxt$ such
    that $\cbvTCtxt<\cbvCCtxt<t>> \cbvArr*_S v$ for some value $v \in
    \setCbvValues$. In particular,
    $\app{(\abs{x}{s})}{\cbvTCtxt<\cbvCCtxt<t>>} \cbvArr*_S
    \app{(\abs{x}{s})}{v} \cbvArr*_S s$ thus using
    \Cref{lem:Theories_lambda_|-_t_=_u_iff_t_=F_u}, one deduces that
    $\cbvThLambdaTail{\app{(\abs{x}{s})}{\cbvTCtxt<\cbvCCtxt<t>>}}{s}$.
    Since $\cbvThH$ is a \cbvLambdaTheoryTxt, then
    $\cbvThHTail{\app{(\abs{x}{s})}{\cbvTCtxt<\cbvCCtxt<t>>}}{s}$ and
    thus $\genThCongr_{\mathsf{H}_\mathsf{v} \cup \{(t, u)\}}$ $\vdash
    \app{(\abs{x}{s})}{\cbvTCtxt<\cbvCCtxt<t>>} \doteq s$. On the
    other hand, by definition of meaninglesness,
    $\cbvTCtxt<\cbvCCtxt<u>>$ necessarily meaningless. Suppose by
    absurd that $\app{(\abs{x}{s})}{\cbvTCtxt<\cbvCCtxt<u>>}$ is
    meaningful, then using
    \Cref{lem:Cbv_App_Meaningful_==>_Arg_Meaningful}, one deduces that
    $\cbvTCtxt<\cbvCCtxt<t>>$ is meaningful which is impossible thus
    $\app{(\abs{x}{s})}{\cbvTCtxt<\cbvCCtxt<u>>}$ is therefore
    meaningless. By contextuality and transitivity
    $\genThCongr_{\mathsf{H}_\mathsf{v} \cup \{(t, u)\}}$ $\vdash
    \app{(\abs{x}{s})}{\cbvTCtxt<\cbvCCtxt<u>>} \cbvThHEq s$ and by
    sensibility $\genThCongr_{\mathsf{H}_\mathsf{v} \cup \{(t, u)\}}
    \vdash \Omega \cbvThHEq s$. Since $s$ is an arbitrary meaningful
    term, one deduces that $\genThCongr_{\mathsf{H}_\mathsf{v} \cup
    \{(t, u)\}}$ is inconsistent contradicting the hypothesis and
    concluding this proof.
\end{proof}

\CbvTheoryHStarHPComplete*%
\label{prf:cbvTheories_HStar_HP_Complete}%
\begin{proof}
    We first prove that $\cbvThH*$ is a maximal consistent and
    sensible \cbvLambdaTheoryTxt and then show that it is unique.
    \begin{enumerate}
    \item Maximality: By \Cref{cor:consistency}, $\cbvThH*$ is a
         consistent and sensible \cbvLambdaTheoryTxt. Let $t, u \in
         \setCbvTerms$ such that $\cbvThH* \cup \{(t, u)\}$ is
         consistent. By inclusion, one deduces that
         $\genThCongr_{\mathsf{H}_\mathsf{v} \cup \{(t, u)\}}$ is
         consistent, thus that $\cbvThH* \vdash t \doteq u$ using
         \Cref{lem:Cbv_Theory_H_Extended_Single}. The theory
         $\cbvThH*$ is therefore maximal.

    \item Unicity: Let $\mathcal{E}$ be a maximal consistent and
        sensible \cbvLambdaTheoryTxt. Let $t, u \in \setCbvTerms$ such
        that $\mathcal{E} \vdash t \doteq u$. By inclusion, one
        deduces that $\genThCongr_{\mathsf{H}_\mathsf{v} \cup \{(t,
        u)\}}$ is consistent thus using
        \Cref{lem:Cbv_Theory_H_Extended_Single}, one has that
        $\cbvThH* \vdash t \doteq u$ and therefore $\mathcal{E}
        \subseteq \cbvThH*$. Suppose by absurd that there exists $t, u
        \in \setCbvTerms$ such that $\cbvThH* \vdash t \doteq u$ but
        not $\mathcal{E} \vdash t \doteq u$. Since $\cbvThH*$ is
        consistent, then by inclusion, one deduces that $\mathcal{E}
        \cup \{(t, u)\}$ is consistent which contradict the fact that
        $\mathcal{E}$ is a maximal consistent and sensible
        \cbvLambdaTheoryTxt. One concludes that $\mathcal{E} =
        \cbvThH*$.
        \qedhere
    \end{enumerate}
\end{proof}

\end{document}